\documentclass[11pt]{article}
\usepackage[margin=0.8in]{geometry}
\usepackage{booktabs}
\usepackage[ruled,linesnumbered]{algorithm2e}

\SetAlFnt{\small}
\SetAlCapFnt{\small}
\SetAlCapNameFnt{\small}
\SetAlCapHSkip{0pt}
\IncMargin{-\parindent}

\usepackage{bbm}
\usepackage{framed}
\usepackage{url}
\usepackage{complexity}
\usepackage[T1]{fontenc}
\usepackage{amsmath}
\usepackage{amssymb}
\usepackage{amsfonts}
\usepackage{amsthm}
\usepackage{bm}
\usepackage{thmtools}
\usepackage{thm-restate}
\usepackage{nicefrac}
\usepackage{calc}
\usepackage{enumerate}
\usepackage{enumitem}
\usepackage{xcolor}
\usepackage{natbib}
\renewcommand{\cite}{\citep}
\usepackage{setspace}
\allowdisplaybreaks

\usepackage{graphicx}
\usepackage[font=footnotesize,labelfont=bf]{subcaption}
\usepackage[font=footnotesize,labelfont=bf]{caption}
\usepackage[nobreak=true,skipabove=7pt]{mdframed}

\usepackage{xr}
\usepackage{array}
\usepackage{xspace}

\usepackage{chngcntr}
\usepackage{mathtools}
\usepackage{pifont}

\usepackage{tikz}
\usetikzlibrary{patterns}
\usetikzlibrary{positioning}

\DeclareMathOperator*{\expectation}{\mathbb{E}}
\renewcommand{\E}[2]{\mathbb{E}_{#1}\left[#2\right]}

\let\poly\relax
\DeclareMathOperator*{\poly}{poly}
\DeclareMathOperator*{\cov}{cov}
\DeclareMathOperator*{\probability}{\Pr}

\DeclareMathOperator*{\dtv}{d_{TV}}
\DeclareMathOperator*{\dkl}{D_{KL}}

\newcommand\eps{\epsilon}

\newcommand\abs[1]{\lvert #1 \rvert}

\newcommand{\prob}{\probability\probarg}
\DeclarePairedDelimiterX{\probarg}[1]{(}{)}{%
	\ifnum\currentgrouptype=16 \else\begingroup\fi
	\activatebar#1
	\ifnum\currentgrouptype=16 \else\endgroup\fi
}
\newcommand{\probover}[2]{\probability_{#1}\left( #2 \right)}
\newcommand{\expect}{\expectation\expectarg}
\DeclarePairedDelimiterX{\expectarg}[1]{[}{]}{%
	\ifnum\currentgrouptype=16 \else\begingroup\fi
	\activatebar#1
	\ifnum\currentgrouptype=16 \else\endgroup\fi
}

\DeclarePairedDelimiterX{\wklarg}[1]{(}{)}{%
	\ifnum\currentgrouptype=16 \else\begingroup\fi
	\activatebars#1
	\ifnum\currentgrouptype=16 \else\endgroup\fi
}
\newcommand{\innermid}{\nonscript\;\delimsize\vert\nonscript\;}
\newcommand{\activatebar}{%
	\begingroup\lccode`\~=`\|
	\lowercase{\endgroup\let~}\innermid
	\mathcode`|=\string"8000
}
\newcommand{\innermids}{\nonscript\;\delimsize\vert\delimsize\vert\nonscript\;}
\newcommand{\activatebars}{%
	\begingroup\lccode`\~=`\|
	\lowercase{\endgroup\let~}\innermids
	\mathcode`|=\string"8000
}

\newcommand{\defeq}{\coloneq}

\counterwithin{equation}{section}

\usepackage{placeins}
\SetKwInOut{Input}{Input}
\SetKwInOut{Output}{Output}
\DontPrintSemicolon

\usepackage[hidelinks]{hyperref}
\usepackage[nameinlink,capitalise]{cleveref} % load after algorithm2e
\crefrangelabelformat{enumi}{#3#1#4--#5#2#6}
\crefname{claim}{Claim}{Claims}
\crefname{appendix}{Appendix}{Appendices}
\Crefname{appendix}{Appendix}{Appendices}

\newcommand{\lineref}[1]{\hyperref[#1]{line~\ref*{#1}}}
\newcommand{\Lineref}[1]{\hyperref[#1]{Line~\ref*{#1}}}

\newcommand{\cA}{\mathcal{A}}
\newcommand{\cC}{\mathcal{C}}
\newcommand{\cD}{\mathcal{D}}
\newcommand{\cE}{\mathcal{E}}
\newcommand{\cF}{\mathcal{F}}
\newcommand{\cH}{\mathcal{H}}
\newcommand{\cI}{\mathcal{I}}
\newcommand{\cJ}{\mathcal{J}}

\renewcommand\cP{\mathcal{P}}

\newcommand\cQ{\mathcal{Q}}

\newcommand\cU{\mathcal{U}}
\newcommand\bu{\bm{u}}

\newcommand{\cW}{\W}
\newcommand{\cX}{\mathcal{X}}

\usepackage{dsfont}
\DeclareMathOperator{\1}{\mathds{1}}
\DeclareMathOperator{\Bern}{\mathrm{Bern}}
\DeclareMathOperator{\Bin}{\mathrm{Bin}}
\DeclareMathOperator{\Multin}{\mathrm{Multinom}}

\theoremstyle{plain}
\newtheorem{theorem}{Theorem}[section]
\newtheorem{definition}[theorem]{Definition}
\newtheorem{proposition}[theorem]{Proposition}

\newtheorem{lemma}[theorem]{Lemma}

\newtheorem{observation}[theorem]{Observation}

\newtheorem{example}[theorem]{Example}

\newtheorem{corollary}[theorem]{Corollary}

\DeclareMathOperator{\cvg}{cvg}

\DeclareMathOperator{\err}{err}
\newcommand\nf{\nicefrac}
\renewcommand\epsilon{\varepsilon}

\newcommand{\algstyle}[1]{\textsc{#1}\xspace}

\newcommand{\prof}{\cP}
\newcommand{\comms}{\cW}
\newcommand{\cands}{\cC}
\renewcommand{\epsilon}{\varepsilon}
\newcommand{\ax}{\mathrm{Ax}}

\newcommand{\psma}{p_\downarrow}
\newcommand{\pbig}{p_\uparrow}

\newcommand{\jr}{\text{JR}}
\newcommand{\ejr}{\text{EJR}}
\newcommand{\ejrp}{\text{EJR+}\xspace}
\newcommand{\psc}{\text{PSC}}

\newcommand{\jre}{\ensuremath{\text{JR}_{\eps}}\xspace}
\newcommand{\bjr}{\text{BJR}}
\newcommand{\fpjr}{\text{FPJR}}
\newcommand{\ejrpe}{\ensuremath{\text{EJR}_{\eps}+}\xspace}
\newcommand{\psce}{\ensuremath{\text{PSC}_{\eps}}\xspace}

\newcommand{\av}{\algstyle{AV}}
\newcommand{\greedy}{\algstyle{Greedy-AV}}
\newcommand{\samplegreedy}{\algstyle{Sample-Greedy-AV}}
\newcommand{\gjcr}{\algstyle{GJCR}}
\newcommand{\samplegjcr}{\algstyle{Sample-GJCR}}
\newcommand{\pav}{\algstyle{PAV}}
\newcommand{\seqpav}{\algstyle{SeqPAV}}
\newcommand{\samplepav}{\algstyle{Sample-PAV}}
\newcommand{\lspav}{\algstyle{ls-PAV}}
\newcommand{\samplelspav}{\algstyle{Sample-ls-PAV}}
\newcommand{\samplecspav}{\algstyle{Sample-Convergent-Search-PAV}}
\newcommand{\samplewslspav}{\algstyle{Sample-PAV-then-Pivot}}
\newcommand{\samplemaxkcover}{\algstyle{Sample-Max k-Cover}}
\newcommand{\chambcou}{\algstyle{CC}}
\newcommand{\mes}{\algstyle{MES}}

\newcommand{\scr}{\mathrm{sc}}
\DeclareMathOperator{\pavsc}{\scr_{\pav}}

\newcommand\mkc{\textsc{Max k-Coverage}\xspace}
\newcommand\MkC{\mkc}

\usepackage{soul}

\title{Proportionality from Sampled Approvals}
\author{
    Gregory Kehne\thanks{Washington University in St. Louis. {\tt kehne@wustl.edu}}
}
\date{}

\begin{document}

\maketitle
\begin{abstract}
    How much voter input is necessary in order to ensure representation in multiwinner elections?
    If voters are randomly selected from an underlying population, how many draws are necessary to find a proportional committee of $k$ candidates, with high probability?

    Sample-based adaptations of standard multiwinner voting rules that satisfy the justified representation (JR) proportionality axiom use $\tilde O(k^5 \log \nf{m}{\delta})$ sampled approval ballots over $m$ candidates, where $\delta$ is a probability of failure and $\tilde O$ suppresses $\polylog(k)$ factors.
    We present a rule for which the sample complexity of JR-family proportional committee selection is $\tilde O(k^{4}\log \nf{m}{\delta})$.
    This separates the sample complexity of JR from that of the natural corresponding additive approximation to the voter coverage (Chamberlin-Courant) objective, which we show requires $\Theta(k^5\log \nf{m}{\delta})$ samples.

    For lower bounds, we present a family of instances with $m, \nf{1}{\delta} \in \poly(k)$ for which
    $\Omega(k^3)$
    sampled ballots are necessary in order to identify a JR committee.
    We also show a dependence on $\log m$ is necessary.
    This lower bound is versatile, and also applies to Hare proportionality for solid coalitions (PSC) for ranked ballots.
    Unfortunately, no number of sampled ballots suffices to satisfy the slightly stronger Droop JR and Droop PSC axioms with high probability.
    But mild relaxations of JR require fewer samples, as do the beyond-worst-case domains and actual approval preferences we evaluate.
\end{abstract}

\section{Introduction}\label{sec:intro}

We set out to determine how much voter input is required in order to identify representative outcomes in multiwinner elections.
If voters are randomly selected from an electorate, how many draws are necessary to find a committee that is proportional \emph{with respect to the underlying electorate}, with high probability, and how does this requirement increase with the winning committee size?

This question is particularly well motivated in settings such as participatory budgeting and polling, where participation is costly and alternative spaces are complex.
The aims of promoting high and representative rates of participation and providing the information and context to voters sufficient to elicit meaningful preference data are often in tension with one another, and it can be difficult to know how to resolve these tradeoffs in practice.
A better understanding of how many randomly selected voters or respondents are necessary, as indexed by the aims of the social choice process, is useful for efficiently obtaining more meaningful and informed preference inputs.
This is part of the recipe for the success of citizens' assemblies in producing population-representative recommendations on complex policy issues: randomly selecting a limited number of participants enables them to be brought together to both deliberate and learn about the issues they are asked to provide recommendations on---and to be paid for the effort.

Even within high-stakes and conventional applications of multiwinner approval voting to committee elections and to participatory budgeting, sample complexity is a natural criterion by which to both differentiate between rules and guide the design of new rules.
Participatory budgeting rules that are more sample-efficient in producing proportional outcomes could be expected to be more robust to changes in turnout or to misspecification of voter preferences, as measured by the axiomatic guarantees that motivate their design and deployment \cite{peters25mathematical}.

Establishing the sample complexity of existing rules is also a theoretical step towards understanding when voter turnout in a multiwinner election becomes low enough that the proportionality guarantees of the outcome can become compromised, as a function of the committee size $k$.
In particular, if each voter has an independent probability of participation, then an election designer might hope to identify a proportional outcome with respect to the underlying population, as weighted by their ex ante turnout probabilities.
In this case, the sample complexity of the voting rule used determines necessary and sufficient bounds on the turnout required for satisfying a given proportionality axiom, as a function of the axiom and the size of the winning committee.

This shift in focus from producing the outcome of a specific rule to ensuring the outcome directly satisfies axiomatic properties is also aligned with and motivated by the emerging application domains of proportional representation.
Prominent recent applications include clustering in metric spaces \cite{chen19proportionally} \cite{caragiannis24proportional},
fairly choosing collections of matchings \cite{boehmer25proportional},
selecting central nodes in a network \cite{papasotiropoulos25proportional},
and models of online deliberation that entail choosing representative subsets from vast spaces of alternatives that contain many near-clones \cite{halpern23representation, fish24generative}.
In each of these domains, proportionality remains an important goal; but the choices a particular rule makes between (near-)clones may be inconsequential, and it may be infeasible to target the outcome of a particular rule if it requires knowing every voter's approval of every candidate in order to choose a winning committee.

Whether conducting public opinion polling, organizing a citizens' assembly, shortlisting a large slate of alternatives, or modeling stochastic voter turnout, attaining outcomes that are proportionally representative \emph{of the underlying voter population} is a natural goal.
To do so, we must understand the sample complexity of our axioms: how many sampled voter preferences are necessary and sufficient for their satisfaction?

%%%
\subsection{Our Contributions}

We initiate the study of the sample complexity of satisfying proportionality axioms in multiwinner elections, where samples are voters (or their ballots) drawn from an unknown profile.
We focus on approval-based multiwinner elections and the justified representation (\jr{}) family of proportionality axioms, and identify implications for ranked-ballot elections and proportionality axioms.

Broadly, we make three main contributions.
First, in \Cref{sec:greedyrules} we establish a novel upper bound on the sample complexity of \ejrp{}, a strengthening of extended justified representation (\ejr{}) (defined below).
%EJR+ upper bound
\begin{theorem}[Informal statement of \Cref{thm:cs-pav-correctness}]
	\label{thm:sample-complexity-ejrp-ub-informal}
	The sample complexity of \ejrp{} is $\tilde O(k^4 \log \nf{m}{\delta})$, where $\tilde O$ hides $\poly\log k$ factors.
\end{theorem}
We accomplish this via a careful combination of the sample-based adaptations of known rules, and observe that it is more sample efficient than our adaptations of other known proportional rules.

Along the way, in \Cref{sec:coverage} we analyze the sample complexity of satisfying \jr{} via the Chamberlin-Courant rule, which maximizes voter coverage.
We present matching upper and lower bounds, which combine with our results in \Cref{sec:greedyrules} to separate the complexity of approximate \emph{global} voter coverage maximization from approximate \emph{local} voter coverage maximization (\jr{}).

Third, in \Cref{sec:lowerbounds} we establish a nontrivial lower bound on the sample complexity of \jr{}.
% JR lower bound
\begin{theorem}[Informal statement of \Cref{thm:k3-lower-bound-jr}]
	\label{thm:kcubed-jr-lb-informal}
	For any $k \geq 2$ there are elections with $m, \nf{1}{\delta} = \poly(k)$ for which satisfying \jr{} requires $\Omega(k^3)$ samples.
\end{theorem}
This is accomplished via a distribution over partition instances, and improves by a factor of $k^2$ upon the
bound obtained by reducing to multiple hypothesis testing.

We complement these main contributions with beyond-worst-case guarantees for structured instances (\Cref{sec:vcdim}) and implications for the proportionality for solid coalitions (PSC) axiom in ranked-preference domains (\cref{sec:ranked}).
Finally, we evaluate the sample complexity of ensuring proportional outcomes on a selection of approval profiles from participatory budgeting elections (\Cref{sec:experiments}).

%%%
\subsection{Related Work}
This work draws from and builds upon a few areas of research within computational social choice.

\subsubsection*{Sortition and Sampling in Social Choice}

The question of how many randomly chosen voters are needed to make collective decisions while representing the underlying population is central to sortition and the design of citizens' assemblies.
Most recent theoretical and practical work on sortition uses a formulation that is, in some sense, dual to ours.
For a fixed number of individuals, representativeness constraints are placed on the outcome of the participant sampling procedure, and the sampling procedure is made maximally uniform subject to these constraints \cite{flanigan20neutralizing,flanigan21fair}.

Other recent work models sortition as proceeding via uniformly random selection, and studies the extent to which this selection satisfies definitions of proportional representation of metric preferences \cite{ebadian25boosting} and other demographic, metric, and welfarist objectives \cite{brustle2025panel}.
The sample complexity questions \citet{brustle2025panel} ask are close to ours in spirit.
For instance, they bound the extent and likelihood that a budget allocation in the approximate core for the sampled voters is in the (less) approximate core for the population.
Materially, the most similar prior results are due to \citet[Corollary 1]{chen19proportionally}.
They work in the clustering setting with the related axiom of \emph{proportional fairness} (PF), and (informally) they show that $\tilde O(k \eps^{-2})$ samples suffice to find a $k$-subset satisfying PF with high probability, provided the proportionality quota is increased by an $\eps$ proportion of the population.

A distinct line of work proposes to find proportional committees in large-scale multiwinner approval voting instances by sampling voters i.i.d. but sequentially, and making adaptive and sparse queries of voter preferences \cite{halpern23representation,brill23robust}.
As their queries are sparse they require $\Omega(m)$ sampled voters, and obtain incomparable bounds.
This has also inspired generative approaches to constructing slates of proportionally representative opinion statements \cite{fish24generative,boehmer25generative}, where access to textual representations of all voters' preferences is assumed.
Working in a similar model to that of \citet{halpern23representation}, \citet{lindeboom25diverse} study the number of voter-candidate approvals necessary to implement polynomial-time approximations to the coverage objective.

\subsubsection*{Learning Axioms and Rules from Samples}
Prior work has studied both approval-based multiwinner voting rules and proportionality axioms from the probably-approximately correct (PAC) learning framework.
\citet{caragiannis22complexity} showed committee scoring rules and sequential Thiele rules have polynomial sample complexity, and \citet{xia25linear} both extended these results and considered the sample complexity of proportionality axioms such as \jr{}.
This was in the course of developing a unified theory of many ABC voting rules and axioms that, like us, takes profiles to be distributions over approval sets.
This line of work is distinct in that their samples are pairs $(E,W)$ from an unknown distribution over elections and committees output by an unknown rule, and the PAC guarantees are with respect to this profile distribution.

\subsubsection*{Stability and Robustness of Multiwinner Rules}
Another area of research seeks to characterize how changes to approval profiles affect the outputs of various rules, or degrade the axiomatic guarantees of rules' outputs.
Within this area, closest to our perspective is perhaps \cite[Theorem 1]{kraiczy23properties}, which observes that committees with sufficiently high \pav{} scores continue to satisfy \ejr{} even when a $\nf{1}{k^2}$ proportion of approvals in the profile are adversarially changed.
Relatedly, \citet{kehne25robust} produce randomized rules that exhibit bounded change (in distribution) under small changes to the profile.
\citet{casey25justified} conduct an in-depth analysis of Hare vs. Droop rules and axioms, and observe that committees satisfying the more stringent (Droop) proportionality axioms are empirically much more scarce.
Beyond the Droop quota, \citet{becker25minimal} study standard rules' approximation of the instance-optimal quota for which \jr{} or \ejr{} is satisfiable.
\citet{imber25approval} consider the question of determining if a given committee possibly/necessarily satisfies \jr{} or is the outcome of a rule for an election for which only a fixed subset of the voter-candidate approvals are known.
Finally, \citet{gawron19robustness} and \citet{faliszewski22robustness} study the extent to which the outputs of specific approval-based committee rules are affected by small changes to the approval profile.

\section{Preliminaries} \label{sec:prelims}

\subsection{Election Models}
\label{sec:model}

Throughout we will use \emph{candidates} and \emph{alternatives} interchangeably, as well as \emph{voters} and \emph{agents}.
We consider a finite set $\cC$ of candidates.
Our aim is to output a committee $W \subseteq \cC$, $\abs{W} = k$ for some specified $k \in \mathbb{N}$, based on the preferences of voters in a fixed electorate.

In approval-based committee (ABC) elections, each voter $v$ is identified with a set of approved candidates $A_v \subseteq \cC$.
Our focus is the setting where the voter population is large; as such, we view the population as a distribution over approvals $\prof \in \Delta(2^{\cC})$.
An \emph{election} is then fully defined by the tuple $E = (\cC, \prof, k)$.
Let $\prof(A) \defeq \Pr_{A' \sim \prof}\left[A' = A\right]$ denote the proportion of the population with approval set $A\subseteq \cC$.
Then a sample of $t$ voters $\cA = (A_1, \ldots, A_t)$ is a draw from the product distribution $\prof^t$.

For a discrete population $N$ of $n$ voters, this corresponds to uniform sampling with replacement.
Although drawing $t$ samples from $N$ with vs. without replacement results in distinct distributions, they converge as $n$ becomes large when---as is the case for our lower bound instance families---the support is small \cite{johnson2024relative}.
Furthermore, recent work in metric settings shows that samples without replacement are more informative \cite{brustle2025panel}, suggesting that our positive results are in the more challenging of the two models.

A typical ABC rule $f:  \Delta(2^{\cC}) \rightarrow \comms$ receives the full preferences of the population and outputs a committee in $\comms \defeq \binom{\cands}{k}$ for a specified committee size $k$.
By contrast, our interest is in \emph{sample-based rules} of the form $\bar f:  (2^{\cC})^t \rightarrow \comms$ which see only $t$ voters' preferences.
Our goal will be for such $\bar f$ to output committees which satisfy axioms defined with respect to the population profile $\prof \in \Delta(2^{\cC})$.
In general, a multiwinner approval voting axiom $\ax: \Delta(2^\cC) \rightarrow 2^{\comms}$ specifies a set of feasible output committees for each input preference profile.
(We will introduce specific axioms shortly.)
We can now formalize what it means for a sample-based rule to satisfy an axiom.

\begin{definition}[Axiom satisfaction w.h.p.]
\label{def:axiom-satisfaction-samples}
	We will say that a  rule $\bar f$ \emph{$(t,\delta)$-satisfies} an axiom $\ax$ if
	\[
		\Pr_{\cA \sim \prof^t} \left[ \bar f(\cA) \in \ax(\prof) \right] \geq 1 - \delta
	\]
	for all profiles $\prof \in \Delta(2^\cC)$.
	Further, $\ax$ is \emph{$(t,\delta)$-satisfiable} if there exists a rule that $(t,\delta)$-satisfies it.
\end{definition}
This means that regardless of the population profile, the approvals $\cA$ of $t$ independently sampled voters suffice for $\bar f$ to identify a committee that satisfies $\ax$ \emph{with respect to the underlying population} at least a $1-\delta$ proportion of the time.
Rules $\bar f$ are more sample-efficient for $\ax$ if, for a given $\delta$, they $(t,\delta)$-satisfy $\ax$ for smaller $t$ (or if, for a given $t$, they $(t,\delta)$-satisfy $\ax$ for smaller $\delta$).

%%%
\subsection{Sample Complexity}

At a high level, if we have an algorithmic task defined with respect to a distribution over data, then it makes sense to ask whether this task can be accomplished by an algorithm given some number of samples from this distribution, and if so, how the best-possible expected performance or likelihood of succeeding at this task improves with the number of samples.
An influential framework for understanding such problems is probably approximately correct (PAC) learning \cite{valiant84theory}, which initially considered the problem of approximately learning an approximation to a `true' classifier $h^*:\cX \rightarrow \{0,1\}$ of instances $\cX$. Given labelled samples $(x, h^*(x))$ from some unknown instance distribution $\cD$, the goal is then to choose an $h \in \cH$ from some larger class, which is likely to agree with $h^*$ on a $1-\eps$ fraction of future samples $x \sim \cD$.
This admits myriad generalizations, for instance by replacing agreement with $h^*$ by more general error metrics \cite{haussler92decision}.

A simple such problem is that of distinguishing a Bernoulli source with rate $p + \eps$ from $m-1$ bad sources with rate at most $p$, given a uniform number of samples from each of the distributions.\footnote{Note that this is distinct from best-arm identification in the bandit setting, where adaptive sampling of indices is allowed.}
The number of samples needed to identify the good source with high probability is well understood.
It coincides with that of the more general problem of selecting a source which is within $\eps$ of optimal, which will be useful to us.

\begin{restatable}[$\eps$-Maximum Bernoulli sources]{fact}{distinguishingsourcesnew}
\label{lem:separating-means-from-samples-new}
	Given probabilities $(p_1,\ldots, p_m)$, consider the product distribution $\cD \defeq \prod\Bern(p_i)$.
	If $i^* \defeq \arg \max_i p_i$ is a maximum-probability index, then $\Theta(\nf{1}{\eps^2} \log \nf{m}{\delta})$ samples from $\cD$ are necessary and sufficient to identify an $i \in [m]$ for which $p_i \geq p_{i^*} - \eps$ with probability at least $1-\delta$.
\end{restatable}

From this perspective, we should approach axiom satisfaction (\Cref{def:axiom-satisfaction-samples}) by viewing each $k$-committee $W\subseteq C$ as a hypothesis, and from the sample identifying a $W$ that is sufficiently good in terms of an objective related to our proportionality axiom.
This is what \samplepav (\Cref{alg:sample-pav}) does, and this approach turns out to be optimal for the related coverage problem we consider in \Cref{sec:coverage}.

But for axioms in the \jr{} family, we can do better by engaging with the structure of our space of possible outputs $\binom{C}{k}$ and its relationship to our objective (axiom satisfaction) directly.
In particular, our objective is not characterized by approximately minimizing any linear function over the voters.
This situates us in a broader line of work on combinatorial optimization from samples, which includes finding approximately influence-maximizing sets of nodes on networks \cite{sadeh20sample} and sample-based submodular function maximization
\cite{rosenfeld18learning}.

%%%
\subsubsection{Sampling with vs. without Replacement}
We consider sampling with replacement, even though sampling without replacement is more closely aligned to our motivation \cite{brustle2025panel}.
This is for convenience, but we argue sample complexity in the two models is asymptotically equivalent for our parameter regime of interest.
For a finite population, sampling with and without replacement coincide when the number of samples is bounded away from the total population size.

\begin{restatable}[]{observation}{withvswithoutreplacement}
\label{claim:with-vs-without-replacement}
	Suppose voter distribution $\cD$ is uniform over a finite population $[n]$.
	If $t$ and $[n]$ are known, then $t$ samples without replacement (w.o.r.) can simulate $t$ samples with replacement (w.r.).
	If the $i \in [n]$ are distinguishable, then for any $c > 0$ and $\eps >0$,  if $t \leq (1 - \eps) n$ then $t' = t/\eps + O(c \cdot \log n)$ samples w.r. can simulate $t$ samples w.o.r. with failure probability $n^{-c}$.
\end{restatable}
So if we are interested in $t = \Omega(\log n),  t =(1 - \Omega(1)) \cdot n$ samples \emph{without} replacement, then samples \emph{with} replacement are weaker by at most a constant, and lower bounds carry over directly.

%%%
\subsection{Approval-Based Axioms and Rules}

%%%
\subsubsection{Proportionality Axioms}
\label{sec:proportionality}
We now define the axiomatic operationalizations of proportionality which we will use.
This presentation is nonstandard only in that it is for profiles $\prof \in \Delta(2^\cC)$, rather than with respect to a finite set of voters.
For a candidate $c$, let $\cA_c \defeq \{A \in 2^\cC: c \in A\}$ denote the collection of approval sets containing $c$.
We are now ready to define \emph{justified representation} (\jr{}), introduced by \citet{aziz17justified}.
Since we will be interested in different thresholds between $\nf{1}{k}$ and $\nf{1}{k+1}$, we will use the notation \jre{} to refer to \jr{} with intermediate thresholds.
\begin{definition}[\jre{}]
\label{def:jr-eps}
	For $\eps > 0$, a committee $W \in \comms$ satisfies \jre{} for an election $E$ if
	\[
	    \Pr_{A \sim \prof}\left[A \in \cA_c \: \wedge \: \abs{A \cap W} = 0 \right] < \frac{1}{k+1} + \eps
	\]
	for all candidates $c \not\in W$.
\end{definition}
The standard axiom of \emph{justified representation} corresponds to \jre{} for $\eps = \frac{1}{k} - \frac{1}{k+1} = \frac{1}{k(k+1)}$.
We will refer to \jre{} with this choice of $\eps$ as \emph{Hare \jr{}}, or simply \emph{\jr{}}.
The axiom \emph{Droop \jr{}} is then the intersection of \jre{} for all $\eps > 0$; that is, Droop \jr{} is satisfied exactly when \jre{} is satisfied for all $\eps > 0$.
Some prior work relaxes \jr{} by taking $\eps > \frac{1}{k(k+1)}$, referring to this as $\alpha$-\jr{}, where $\alpha \defeq (\frac{k}{k+1} + k\eps)^{-1}$.

% EJR+
The stronger proportionality notion we will consider is $\ejrp{}$, due to \citet{brill23robust}.
This is a refinement of \emph{extended} \jr{}, also introduced by \citet{aziz17justified} and itself stronger than \jr{}.
We again parameterize the choice of quota by $\eps$.
\begin{definition}[\ejrpe{}]
\label{def:ejrp-eps}
	For $\eps > 0$, a committee $W \in \comms$ satisfies \emph{\ejrpe{}} with respect to profile $\prof$ if
	\[
	    \Pr_{A \sim \prof}\left[A \in \cA_c \: \wedge \: \abs{A \cap W} < \ell \right] < \ell \cdot \left( \frac{1}{k+1} + \eps\right)
	\]
	for all candidates $c \not \in W$ and for all $\ell \in [k]$.
\end{definition}
Again, \ejrp{} corresponds to \ejrpe{} for $\eps = \frac{1}{k(k+1)}$.
It is easy to see that \ejrp{} strengthens \jr{}, since we recover \Cref{def:jr-eps} from \Cref{def:ejrp-eps} by specifying $\ell = 1$.
In the formulation of $\ax: \Delta(2^\cC) \rightarrow 2^{\comms}$ above, the sets of $W$ satisfying \Cref{def:jr-eps,def:ejrp-eps} define the values of \jre{}$(\prof)$ and \ejrpe{}$(\prof)$, respectively.

A large and ever-expanding family of strengthenings of \jr{} have been defined and studied for multiwinner approval voting rules and in related settings.
See \citet[Fig. 4.1]{lackner2023multi} for a family picture as of 2023.
Together with \ejrp{}, subsequent arrivals include the Monroe-esque \emph{balanced} \jr{}
(\bjr{})
\cite{fish24generative} and \emph{full proportional} \jr{}
(\fpjr{})
\cite{kalayci25full}.

Of these, \ejrp{} is strength-maximal among the axioms which are known to be satisfiable for all profiles, is satisfied by several appealing rules, and is relatively discerning on real and synthetic profiles \cite[Table 2]{boehmer2024approval} \cite[Figs. 2-5]{brill23robust}.
Our lower bounds hold for \jr{}, and hence for every proportionality axiom in the \jr{} family.
Our algorithmic focus in \Cref{sec:upperbounds} will be on satisfying \ejrp{}.

\subsubsection{Committee Voting Rules}

Several prominent multiwinner voting rules satisfy \jr{} and its strengthenings.
We will work with the following such rules.
\begin{definition}[Thiele methods  \cite{thiele95om}]
\label{def:thiele-methods}
	For a nonnegative and non-decreasing scoring vector $w = (w_1, w_2, \ldots)$, the \emph{$w$-Thiele method} chooses a committee $W \in \comms$ maximizing the score $\scr_w(W) \defeq \probover{A\sim \prof}{w_{\abs{A\cap W}}}$.
	The \emph{(utilitarian) approval voting rule} (\av{}) is given by $w_\ell = \ell$.
	The \emph{Chamberlin-Courant rule} (\chambcou) is given by $w_\ell = 1$.
	The \emph{proportional approval voting} (\pav{}) rule is given by $w_\ell = H_\ell$, where $H_\ell$ is the $\ell$'th harmonic number.
	Let $\scr_{\av}$, $\scr_{\cc}$, and $\scr_{\pav}$ denote their scores.
\end{definition}
When sample-based rules calculate scores, it is with respect to sampled approvals rather than $\prof$.
In such cases we will let $\scr(W ; \cA)$ denote the score of $W$ on the ballots $\cA$ (normalized by $\abs{\cA}$).

Thiele methods admit more computationally tractable greedy variants:
\begin{definition}[Sequential Thiele methods \cite{thiele95om}]
\label{def:sequential-thiele-methods}
	For a scoring vector $w$ as above, the \emph{sequential $w$-Thiele method} greedily constructs a winning committee $W$ according to $\scr_w$.
	The voting rule \greedy{} is the sequential variant of \chambcou{}, and \emph{Sequential PAV} (\seqpav{}) is that of \pav{}.
\end{definition}
While \seqpav{} does not enjoy the same guarantees as \pav{}, a different tractable variant does:
\begin{definition}[Local-search \pav{} \cite{aziz18complexity}]
\label{def:ls-pav}
	For fixed $\eps > 0$, and a starting committee $W$, \emph{$\eps$-local-search \pav{}} ($\eps$-\lspav{}) iteratively pivots $W \gets W \setminus \{c'\} \cup \{c\}$ so long as there are pairs of candidates $(c,c') \in \cands \times W$ for which $\scr_{\pav{}}(W \setminus \{c'\} \cup \{c\}) \geq \scr_{\pav{}}(W) + \eps$.
\end{definition}

Finally, we also consider (a slight variant of) the EJR-Exact rule Simple EJR, also known as GJCR.
\begin{definition}[\gjcr{} \cite{fernandez17committees} \cite{brill23robust}]
	\label{def:gjcr-eps}
	For fixed $\eps > 0$, the \emph{greedy justified candidate rule} (\gjcr{}$_\eps$)
	constructs $W$ by sweeping from $\ell = k$ down to $\ell = 1$, at each level iteratively taking any candidates that violate the inequality in \Cref{def:ejrp-eps} with respect to $W$.
\end{definition}

We do not work directly with \seqpav{} or \mes{} \cite{peters21proportional}, but evaluate them in \Cref{sec:experiments}.

%%%%
\section{Warmup: Sampling for Chamberlin-Courant and Voter Coverage}
\label{sec:coverage}

\jr{} is a relaxation of the coverage guarantee provided by the Chamberlin-Courant rule for approval ballots, which returns the committee of $k$ candidates which maximizes the number of voters who approve of at least one chosen candidate.
In the approval setting, it is immediate that the Chamberlin-Courant rule is equivalent to \MkC, the problem of choosing $k$ sets from a set family so as to maximize the cardinality of their union \cite{skowron17chamberlin}.

For a given committee, satisfying \jr{} can be seen as a certificate of (approximate) local optimality with respect to voter coverage, while Chamberlin-Courant-winning committees are globally optimal with respect to voter coverage.
Viewed this way, the fact that Chamberlin-Courant-winning committees satisfy \jr{} is somewhat intuitive, and can be formalized via an exchange argument.
\begin{restatable}[]{observation}{cvgdroopjr}
\label{obs:cvg-droop-jr}
    If $W$ is a Chamberlin-Courant winning committee, then $W$ satisfies Droop \jr{}.
\end{restatable}

In the same sense, (Hare) \jr{} is a relaxation of approximately optimal global voter coverage:
\begin{restatable}[]{observation}{cvgharejr}
\label{obs:approx-cvg-hare-jr}
    Let $W^*$ be a Chamberlin-Courant winner for the approval election $E$.
    For any committee $W$, if $\cvg(W) > \cvg(W^*) - \nf{1}{k^2}$, then $W$ satisfies \jr{} for $E$.
\end{restatable}

One benefit of this relaxation from coverage to \jr{} is computational tractability.
The sequential greedy algorithm efficiently finds a (Droop) \jr{} committee in polynomial time, while obtaining better than a multiplicative $(1-\nf{1}{e})$-approximation to the \mkc (i.e. the Chamberlin-Courant) objective is NP-hard \cite{feige98threshold}, and under suitable assumptions this problem remains inapproximable even for small $k$ \cite{manurangsi20tight}.

Computational complexity aside, in light of \Cref{obs:approx-cvg-hare-jr} it is natural to ask whether the sample complexity of \jr{} and the corresponding approximate maximum voter coverage problem coincide or come apart.
Formally, the corresponding additively approximate coverage problem is the following:
\begin{definition}[$\eps$-\MkC]
    \label{def:eps-approx-mkc}
    Given a set system $\mathcal{S}$ over a ground set $\mathcal{U}$ and $k \in \mathbb{N}$,
    let $\mathcal{C} \defeq \binom{\mathcal{S}}{k}$ be the family of all collections of $k$ sets in $\mathcal{S}$.
    Then the $\eps$-approximate Maximum-$k$-Coverage problem is the task of identifying a $C \in \mathcal{C}$ for which
    $\abs{\cup C} \geq \max_{C' \in \mathcal{C}} \abs{\cup C'} - \eps \cdot \abs{\mathcal{U}}$.
\end{definition}

We begin by establishing tight bounds on the sample complexity of this problem.
\begin{theorem}
\label{thm:eps-mkc-sample-complexity}
    The sample complexity of
    $\eps$-\MkC is $\Theta\left(\nf{k}{\eps^2}  \log \abs{\cS} \right)$, provided $\delta = \abs{\cS}^{- O(k)}$.
\end{theorem}

We now spell out the sample complexity implications for the global coverage problem that implies \jr{}.
\Cref{obs:approx-cvg-hare-jr} shows that $\nf{1}{k^2}$-\MkC implies \jr{}, and this is tight: simple instances (such as \Cref{ex:hard-partition} below) imply that if $\eps$-\MkC implies \jr{} then $\eps = O(\nf{1}{k^2})$.
Therefore, for the easiest version of $\eps$-\MkC that implies \jr{} we have:
\begin{corollary}
\label{cor:1/k2-mkc-sample-complexity}
    For any $\delta = \abs{\cS}^{- O(k)}$, the sample complexity of $\nf{1}{k^2}$-\MkC is $\Theta\left(k^5 \log \abs{\mathcal{S}}\right)$.
\end{corollary}

This is therefore the dependence on $k$ to beat.
In the next section we will set out to separate the sample complexity of \jr{} from that of its global coverage analog.

%%%
\subsection{\texorpdfstring{$\eps$}{e}-\MkC from Samples}
We will split the proof of \Cref{thm:eps-mkc-sample-complexity} into two parts.
We start by showing a sample complexity upper bound via the brute-force algorithm for $\eps$-\MkC:

\begin{algorithm}
	\caption{\samplemaxkcover}
	\label{alg:eps-mkc-brute-force}
	\KwIn{Sets $\cS$, $k$, number of samples $t$}
	Sample elements $u^1, \ldots, u^t \sim \cU$ independently and u.a.r.\;
	\Return $\arg\max_{W \in \binom{\cS}{k}}\abs{\{i \in [t]: u^i \in \cup W\}}$
\end{algorithm}

\begin{restatable}[]{lemma}{epsmkcsampleub}
\label{lem:eps-mkc-sample-upper-bound}
    The empirical maximizer brute-force algorithm for $\eps$-\MkC (\cref{alg:eps-mkc-brute-force}) has sample complexity
    $O\left(\nf{1}{\eps^2} \left(k \log \abs{\mathcal{S}} + \log \nf{1}{\delta}\right)\right)$.
\end{restatable}
This directly follows from the upper bound portion of \Cref{lem:separating-means-from-samples-new}.

%%%
\subsection{Sample-Hard Instances of \texorpdfstring{$\eps$-\MkC}{e-Max k-Coverage}}

The lower bound is slightly more involved.
It consists of instances which divide the universe into $k$ disjoint equal parts, and hide (in the same fashion as in the lower bound portion of the proof of \Cref{lem:separating-means-from-samples-new}) a distinguished high-coverage set for each part among many sets which provide $\Theta(\nf{\eps}{k})$ less coverage.
In order to provide an $\eps$-approximate solution to $\eps$-\MkC, any sample-based algorithm must then succeed on a constant proportion of the disjoint parts.
This essentially follows the construction of the general agnostic PAC learning lower bound of \citet{simon96general} \cite[Theorem 6.8]{shwartz14understanding}.

\begin{restatable}[]{lemma}{epsmkcsamplelb}
\label{lem:eps-mkc-sample-lower-bound}
	Suppose we take $\delta = \nf{1}{2}$.
	For any $k\geq 2$ and for any $\abs{\cS}= k^{1 + \Omega(1)}$, the sample complexity of $\eps$-\MkC is $\Omega\left( \nf{k}{\eps^2} \log \abs{\cS} \right)$.
\end{restatable}

Finally, it is worth asking whether these instances that are difficult from the perspective of $\nf{1}{k^2}$-\MkC are difficult for \jr{}.
For the construction that proves \Cref{lem:eps-mkc-sample-lower-bound}, the answer is no.
Trivially, this is because all sets $S \in \cS$ in this construction have size less than $\abs{\cU}/k$.
But even if we were to use only $\nf k 2$ sub-instances so that each set has size $\abs{S} \geq \abs{\cU}/k$, any choice $C$ of $\nf k 2$ subsets that includes one set from each sub-instance would satisfy \jr{}, even though all such $C$ fail $\eps$-\MkC.
This is because although globally such $C$ have poor coverage, no unchosen $S \not \in C$ have marginal coverage $\abs{\cU}/k$.

\section{Sample-Efficient Algorithms for the \jr{} Family}
\label{sec:upperbounds}

The most natural way to construct sample-efficient rules for the \jr{} family is to adapt existing rules to make their decisions based only on sampled ballots.
We begin by briefly presenting the sample complexity of the natural sample versions of known rules.

\subsection{Adaptations of Existing Rules}
\label{sec:greedyrules}

When designing a sample-based variant of a coverage-based rule, we must ensure that enough samples are used so that any inaccuracy in which candidate or set of candidates is chosen falls safely within the slack provided by the difference between the Hare and Droop quotas.

By way of illustration, suppose we implement the greedy rule for \jr{} \cite{aziz17justified}, i.e. the greedy algorithm for \MkC.
We need only choose additional candidates if the maximum marginal coverage available is at least $\nf{n}{k}$.
If, in each step of this rule, the chosen candidate (set) provides marginal coverage at least $n/(k+1)$, then the output satisfies \jr{}.
Therefore, it is sufficient to choose a candidate approximating the maximum coverage to (additive) extent $\nf n k -\nf{n}{k+1} \approx \nf{n}{k^2}$.
We denote the sample-based implementation of this rule by \samplegreedy (\Cref{alg:sample-greedy}).

This is the sequential greedy approximation to maximizing coverage (\chambcou{}), and both satisfy \jr{}.
For \ejrp{}, the \pav{} rule can be seen as analogous to \chambcou{} in that it is maximizing a related harmonic welfarist monotone submodular function, and \gjcr{} is a sequential greedy generalization of \greedy{}.
Remarkably, sample-based implementations of all of these rules have the same dependence on $k$ as the sample-based implementation of \chambcou{} (\Cref{alg:eps-mkc-brute-force}):

\begin{restatable}[]{observation}{standardrulecomplexity}
\label{obs:complexity-of-standard-rules}
    The natural sample-based implementations of \samplegreedy{} (\Cref{alg:sample-greedy}), \samplepav{} (\Cref{alg:sample-pav}), and \samplegjcr{} (\Cref{alg:sample-gjcr}) for \jr{} and \ejrp{} all use $\tilde \Theta \left(k^5 \log m \right)$ samples, where $\tilde \Theta$ suppresses $\log k$ and $\log \nf 1 \delta$ factors.
\end{restatable}

Pseudocode for these sample-based rules is provided in \Cref{app:upper-bounds-known}.

This coincidence can be attributed to the fact that independent samples are used for each of the $k$ rounds of marginal coverage estimation in the greedy rules.
This is to avoid correlations due to conditioning; otherwise the error in the marginal coverage estimates in step $i+1$ may be correlated with the choice made in step $i$, and the estimates in step $i+1$ may not concentrate properly.
Sampling enough to union bound over $m$ things $k$ times then takes the same $\nf{1}{\eps^2} \cdot k \log m$ samples as sampling enough to union bound over $m^k$ things once.
This tradeoff is evocative of the design of approval-continuous randomized rules \cite{kehne25robust}, as is the apparent relative difficulty of designing sample-based implementations of share-based rules such as Phragm\'en's method or the Method of Equal Shares.

Though striking, this coincidence for welfare-based proportional rules is not total:
\begin{observation}
\label{obs:complexity-of-ls-pav}
    The natural sample-based implementation of \samplelspav{} (\Cref{alg:sample-ls-pav}) uses $\tilde \Theta \left(k^6 \log m\right)$ samples to satisfy \ejrp{}, where $\tilde \Theta$ suppresses $\log \nf{1}{\delta}$ and $\log k$ factors.
\end{observation}
Despite this strictly worse dependence on $k$, \samplelspav{} nevertheless merits mention as a source of inspiration and point of contrast for our sample-efficient rule.

%%%
\subsection{An \texorpdfstring{$\tilde O(k^4)$}{O(k\^4)}-Sample Rule}
\label{sec:cs-pav}

We now present an improvement on the standard adaptations of existing rules.
It follows from first running \samplepav{} with a specified error strictly larger than the $\eps=\nf{1}{k^2}$ necessary to guarantee \ejrp{}, then using its output committee as the starting point for a sequence of multi-candidate pivots of decreasing size, inspired by \samplelspav{}.
We have seen in \Cref{obs:complexity-of-standard-rules} and \Cref{obs:complexity-of-ls-pav} that individually \samplepav{} and \samplelspav{} require $\tilde O (k^5)$ and $\tilde O(k^6)$ samples to satisfy $\ejrp{}$, respectively.
By carefully interpolating between these rules, we can improve the sample complexity by roughly a factor of $k$.

\begin{algorithm}[h]
	\caption{\samplecspav{}}
	\label{alg:sample-cs-pav}
	\KwIn{Candidates $C$,  $k$, tolerance $\eps$, failure probability $\delta$}
	\KwOut{committee $W$ satisfying $\ejrpe{}$ w.p. $1-\delta$}
	$\eps_P \gets \eps k^{1/2}/(12\log_2^3 k ) $ \;
	Sample $\cA = (A^1, \ldots, A^t) \sim \cP^t$ i.i.d. for $t = \nf{1}{2} (\log k + 1)^2 \cdot \eps_P^{-2} \cdot (k \log \abs{C} + \log \nf{1}{\delta} + 1)$ \label{line:sample-cs-pav-first-samples} \;
	$W \gets \arg\max_{W \in \binom{C}{k}} \pavsc(W; \cA)$ \label{line:pav-step-output-committee} \;
	\For{$r \in \{2, 3, \ldots, \lceil \log_2 \log_2 k \rceil \}$ \label{line:sample-cs-pav-rounds}}{
		$\ell_r \gets \lfloor k^{2^{-r}}\rfloor $ \;
		$\eps_r \gets \eps \cdot 6^r k^{2^{-r}}/(12\log_2^3 k )$ \;
		\For{$\ell_r $ steps \label{line:sample-cs-pav-steps} }{
			Sample $\cA \sim \cP^t$ i.i.d. for $t =18 (\log k + 1)^2 \cdot \eps_r^{-2} \cdot (\ell_r \log \abs{C} + (\ell_r + 2) \log k + \log \nf{1}{\delta} + 1)$ \label{line:sample-cs-pav-second-samples} \;
			$(X^*, Y^*) \gets \arg\max_{(X,Y) \in \binom{W}{\ell_r} \times \binom{C }{\ell_r}} \pavsc(W\setminus X \cup Y; \cA) - \pavsc(W; \cA)$ \label{line:sample-cs-pav-pivot} \;
			\eIf{$\pavsc(W\setminus X^* \cup Y^*; \cA) - \pavsc(W; \cA) \geq \frac{2}{3} \cdot \eps_r$}{
			    $W \gets W\setminus X^* \cup Y^*$ \;
			}{
			\textbf{break} \label{line:sample-cs-pav-break} \;
			}
		}
	}
	\Return $W$ \;
\end{algorithm}

In \samplecspav (\Cref{alg:sample-cs-pav}) we start by calling \samplepav with a coarse but suitably sample-efficient larger tolerance $\eps_P$ in order to provide a warm start.
We then execute $\log \log k$ rounds of pivots which begin by exchanging large portions of the committee in pursuit of commensurately large score improvements, and end just short of single-candidate \samplelspav steps.
These pivots of decreasing size efficiently converge upon an approximately \pav{}-score-maximal committee.
A direct combination of \samplepav and \samplelspav is perhaps more illustrative but slightly less sample-efficient; for details see \samplewslspav in \Cref{app:upper-bounds-pav-pivot}.

We present the algorithm pseudocode and analysis in terms of more general $\eps$.
\begin{restatable}{theorem}{cs-pav-correct}
\label{thm:cs-pav-correctness}
%\begin{theorem}
	\samplecspav{} (\Cref{alg:sample-cs-pav}) satisfies \ejrpe{} with probability at least $1-\delta$.
	Furthermore, it does so using $O\left( \eps^{-2} \cdot \log^9 k \log\log k \cdot \log \nf{m}{\delta}\right)$ samples.
%\end{theorem}
\end{restatable}

Recalling that \ejrp{} corresponds to \ejrpe{} for $\eps = \nf{1}{k^2}$, we have our main result (\Cref{thm:sample-complexity-ejrp-ub-informal}).

To prove this result we will make use of the following lemma, which is a modest generalization of the argument that \pav{} satisfies \ejr{} \cite[Theorem 1]{aziz18complexity} and indeed also \ejrp{} \cite[Proposition 5]{brill23robust}.
A proof appears in \Cref{app:upper-bounds}.
\begin{restatable}[]{proposition}{pavtojerp}
\label{prop:pav-score-to-ejrp-guarantee}
	Consider an election $E$ and a committee $W$.
	If for all $c' \in W$ and $c \not \in W$ we have $\pavsc(W) \geq \pavsc(W \setminus \{c'\} \cup \{c\}) - \eps$, then $W$ satisfies \ejrpe{}.
\end{restatable}

\begin{lemma}[Pivot combination]
\label{lem:combined-pivots}
	Given a committee $W=W_0$, if there exists a sequence of $t$ pivots $\{(X_i, Y_i)\}_{i \in [t]}$ which each improve the \pav{} score from $W_{i-1}$ to $W_i$ by $\eps_i$, then the pivot $(X, Y)$ improves the score of $W$ by $\eps = \sum_i \eps_i$, where $X \defeq W_0 \setminus W_t$ and $Y \defeq W_t \setminus W_0$.
	Moreover, $|X| = |Y| \leq \sum_i |X_i|$.
\end{lemma}
\begin{proof}
	This follows from the definition of a pivot.
	For pivots $(X, Y)$ and $(X',Y')$, suppose that
	\begin{align*}
		\pavsc(W\setminus X \cup Y) - \pavsc(W) &\geq \eps, \\
		\pavsc(W\setminus X \cup Y\setminus X' \cup Y') - \pavsc(W\setminus X \cup Y) &\geq \eps'.
	\end{align*}
	Then directly summing these two inequalities yields
	\begin{align*}
		\pavsc(W\setminus X \cup Y\setminus X' \cup Y') - \pavsc(W) &\geq \eps + \eps'.
	\end{align*}
	Applying this fact inductively over all $i \in [t]$ gives the improvement guarantee.

	The fact that $|X|=|Y|$ follows from the definition of a pivot, and $|X| \leq \sum_i |X_i|$ because each candidate in $X$ must appear in some $X_i$.
\end{proof}

With this in hand, we are ready to analyze the correctness and sample-efficiency of \samplecspav.
\begin{proof}[Proof of \Cref{thm:cs-pav-correctness}]
	The sample complexity claim follows directly from the way in which \Cref{alg:sample-cs-pav} is structured.
	By our choice of $\eps_P= \eps k^{1/2}/(12\log_2^3 k )$, the \pav{} step uses $\tilde O(\eps^{-2})$ samples (\lineref{line:sample-cs-pav-first-samples}).
	How many samples are used by each of the $\log\log k$ rounds of converging pivot steps?
	Each of the $\ell_r$ steps within the round takes $\tilde O(\eps_r^{-2}) = \tilde O(\eps^{-2} k^{-2^{1-r}} \cdot \ell_r)$ samples (\lineref{line:sample-cs-pav-second-samples}).
	By our choice of $\ell_r$, this is $\tilde O(\eps^{-2})$ samples per round overall.
	Summing over the \pav{} start and the $\log\log k$ rounds implies the stated sample bound.

	The main task is to prove correctness.
	This amounts to showing two things: first, that with high probability both the approximate \samplepav{} step and all of the subsequent pivot steps are sufficiently good with respect to their impact on $\pavsc(W)$ for the purposes of constructing the running committee $W$.
	Second, that if all steps are sufficiently good, then \samplecspav terminates with a committee that satisfies \ejrpe{}.
%	We will tackle these in order.

	We begin by defining our failure events.
	Let the event $\cE_{\pav{}}$ denote the failure event for the first step, and let $\{\cE_{r, s}\}$ denote the failure events for the pivot we consider in the $s$th step of round $r$.
	In particular, let $\cE_{\pav{}}$ be the event that the warm start committee $W$ (\lineref{line:pav-step-output-committee}) has true \pav{} score $\pavsc(W) < \max_{W'} \pavsc(W') - \eps_P$, where again we choose $\eps_P \defeq \eps \cdot k^{1/2}/(12\log_2^3 k )$.
	For the per-step failures, we take $\{\cE_{r, s}\}$ to be the event that either (a) a pivot $(X,Y)$ exists that would yield a \pav{} score improvement of at least $\eps_r$, but we do not make a pivot; or (b) we make a pivot but the pivot $(X,Y)$ we choose yields \pav{} score improvement less than $\eps_r / 3$.
	Here we will take $\eps_r \defeq \eps \cdot 6^r\cdot k^{2^{-r}} / (12\log_2^3 k)$.

	Our goal  will be to show that the probability of a failure over the course of the algorithm is at most $\delta$.
	Upper bounding this event will follow the analysis for \samplepav (\Cref{prop:sample-pav-correct}) and \samplelspav (\Cref{prop:sample-ls-pav-correct}).

	Provided this failure probability is appropriately bounded, the second step is to show that if neither 	$\cE_{\pav{}}$ nor any of the $\cE_{r,s}$ occur, then \samplecspav{} satisfies \ejrpe{}.
	To this end, suppose that neither $\cE_{\pav{}}$ nor any of the $\cE_{r,s}$ hold.
	Then first, the committee $W$ chosen by the initial \pav{} warm start (\lineref{line:pav-step-output-committee}) satisfies $s_W \geq s_{W^*} - \eps_P$, and so the sum total of all \pav{} score improvements from subsequent pivots must be at most $\eps_P = \eps\cdot k^{1/2}/(12\log_2^3 k ) $. (Here we use the shorthand $s_W = \pavsc(W)$.)

	Let us now go round-by-round (\lineref{line:sample-cs-pav-rounds}).
	Assuming no failures, we will use \Cref{lem:combined-pivots} to prove inductively that at the end of round $r$, no $\eps_r$-good pivots of size $\ell_r$ are available.
	We can exit the round either via the break condition (\lineref{line:sample-cs-pav-break}) or by exhausting the loop (\lineref{line:sample-cs-pav-steps}).
	By our non-failure assumption, if we exit via the break condition then it is because no $\eps_r$-improving pivots of size $\ell_r$ are possible for our current $W$.
	We will argue that (assuming non-failure) we necessarily exit the round via the break condition.
	For how many consecutive steps can a good pivot exist?
	Consider the base case of $r=2$.
	Provided $\cE_{\pav{}}$ does not hold, $\pavsc(W)$ is within $\eps_P = \eps k^{1/2}/(12\log_2^3 k ) $ of optimal.
	If we then make $\ell_2$ consecutive pivots of size $\ell_2$, then by the non-failure assumption and by \Cref{lem:combined-pivots} they will comprise a single pivot $(X,Y)$ of $\ell_2^2 \leq \sqrt{k}$ candidates and true \pav{} score improvement at least
	\[
		\ell_2 \cdot \nf{1}{3} \cdot \eps_2
		= \lfloor k^{2^{-2}} \rfloor \cdot \nf{1}{3} \cdot  \eps \cdot 6^2 \cdot  k^{2^{-2}}/(12\log_2^3 k )
		> \eps \cdot  k^{1/2} /(12\log_2^3 k )
		= \eps_P.
	\]
	Since the starting committee $W$ was within $\eps_P$ of optimal, such a combined pivot cannot exist.
	Hence we exit round $r=2$ via the break condition (\lineref{line:sample-cs-pav-break}), with a committee for which no $ \eps_2$, $\ell_2$ pivots are possible.

	For general $r$, again suppose that the algorithm does not exit round $r$ via the break condition but rather exhausts the loop; by the non-failure assumption and by \Cref{lem:combined-pivots} the resulting $\ell_r$ pivots together comprise a single pivot $(X,Y)$
	of at most $\ell_r^2 = \lfloor k^{2^{-r}} \rfloor^2 \leq \lfloor k^{2^{1-r}} \rfloor = \ell_{r-1}$ candidates and true \pav{} score improvement at least
	\[
		\ell_r \cdot \nf{1}{3} \cdot \eps_r
		= \lfloor k^{2^{-r}} \rfloor \cdot \nf{1}{3} \cdot \eps \cdot  6^r \cdot  k^{2^{-r}}/(12\log_2^3 k )
		> \eps \cdot 6^{r-1} \cdot k^{2^{-(r-1)}} /(12\log_2^3 k )
		= \eps_{r-1},
	\]
	because $2\lfloor k^{2^{-r}} \rfloor > k^{2^{-r}}$ for all $r$.
	Since the starting committee $W$ admitted no $\eps_{r-1}$-improving pivots of at most $\ell_{r-1}$ candidates by the inductive hypothesis, such a combined pivot cannot exist.
	Hence the algorithm exits every round $r$ via the break (\lineref{line:sample-cs-pav-break}).

	What does this imply for the last round?
	At the end of round $r = \lceil \log_2 \log_2 k \rceil$, the resulting committee $W$ does not admit any pivot of size at most $\ell_r = 2$ that offers an $\eps_r$ improvement to the \pav{} score.
	Since
	\[
		\eps_r = \eps \cdot 6^{\lceil \log_2 \log_2 k \rceil} \cdot  k^{2^{-\lceil \log_2 \log_2 k \rceil}}/(12\log_2^3 k )
		\leq \eps \cdot 6^{1 + \log_2 \log_2 k} \cdot  k^{2^{-\log_2 \log_2 k }}/(12\log_2^3 k )
		=\eps,
	\]
	in particular no two-candidate pivots offer an $\eps$-improvement to the \pav{} score of $W$.
	Since \lineref{line:sample-cs-pav-pivot} considers pivots $(X,Y)$ for which $Y$ intersects $W$, it also quantifies over smaller pivots; in particular, at this point no one-candidate pivots offer an $\eps$-improvement to the \pav{} score of $W$ either.
	By \Cref{prop:pav-score-to-ejrp-guarantee}, the resulting $W$ therefore satisfies \ejrpe{}.

	It remains only to bound the likelihood of failure.
	We begin with $\cE_{\pav{}}$.
	Observe that if $\cE_{\pav{}}$ holds then either (i) the $\pavsc$-maximizing committee $W^*$ has sample score $\hat s_{W^*} \leq \pavsc(W^*) - \eps_P$, or (ii) at least one of the low-scoring committees $W$ with true score $\pavsc(W) \leq \pavsc(W^*) - 2\eps_P$ has sample score $\hat s_{W} \geq \pavsc(W^*) - \eps_P$.
	Let $s^* \defeq \max_{W} \pavsc(W)$ for convenience.
	Then by Hoeffding, the probability of (i) over the samples $\cA = (A^1, \ldots, A^t)$ is
	\begin{equation}
		\probover{\cA \sim \cP^t}{\hat s_{W^*} \leq s^* - \eps_P}
		\leq \prob{t \cdot \hat s_{W^*} \leq t\cdot \expect{\hat s_{W^*}} - t\cdot \eps_P}
		\leq \exp\left(- \frac{2\eps_P^2 t^2 }{t H_k^2 } \right) = \exp\left(- \frac{2 \eps_P^2 t }{H_k^2 } \right),		 \notag
	\end{equation}
	where $H_k \defeq \sum_{j=1}^k \nf{1}{j}$ is the $k$th harmonic number, and a single sample's impact on $t \cdot \hat s_W$ is between $-H_k$ and $H_k$.
	Similarly, for any low-score $W$ the probability of (ii) is at most
	\begin{equation}
		\probover{\cA\sim \cP^t}{\hat s_{W} \geq s^*  - \eps_P}
		\leq \prob{t \cdot \hat s_{W} \geq t \cdot \expect{\hat s_W} + t\cdot \eps_P}
		\leq \exp\left(- \frac{2\eps_P^2 t}{H_k^2} \right).		 \notag
	\end{equation}
	Union bounding over all $W$ and introducing a failure probability of at most $\nf\delta 2$ yields
	\[
		\probover{\cA\sim \cP^t}{\cE_{\pav{}}}
		\leq \binom{m}{k} \cdot \exp\left(-\frac{ 2\eps_P^2 t}{H_k^2}\right)
		\leq \exp\left(k \log m - \frac{ 2\eps_P^2 t}{H_k^2}\right)
		\leq \nf{\delta}{2}.
	\]
	Substituting
	$H_k \leq \log k + 1$ and rearranging, we find this is satisfied if
	\[
		t \geq \nf{1}{2} (\log k + 1)^2 \cdot \eps_P^{-2} \cdot \left( k \log m + \log \nf{1}{\delta} + 1 \right).
	\]
	This holds for our chosen number of samples (\lineref{line:sample-cs-pav-first-samples}).

	With our remaining failure probability budget of $\nf{\delta}{2}$, we now turn to the events $\{\cE_{r,s}\}$.
	Within each round $r$, for each of the $\ell_r$ constituent steps $s$, the number of pivots under consideration is $\binom{m}{\ell_r} \times \binom{k}{\ell_r} \leq (mk)^{\ell_r}$.

	What is the failure probability within each step?
	According to our goals we can identify the sufficient conditions for success that (a) no $\eps_r$-good pivots appear worse than $\nf{2\eps_r}{3}$, and (b) no $\nf{\eps_r}{3}$-bad pivots appear better than $\nf{2\eps_r}{3}$.
	For fixed $W$ and a pivot $q=(X,Y)$, as in \lineref{line:sample-cs-pav-pivot} let
	\[
		\hat s_{XY} \defeq \pavsc(W\setminus X \cup Y; \cA) - \pavsc(W; \cA)
	\]
	denote the empirical estimate of the pivot's $\pavsc$ impact on $W$.
	In the first case, for a $\eps_r$-good pivot we can use Hoeffding to bound the probability of the failure-necessary event (a):
	\begin{equation}
		\probover{\cA\sim \cP^t}{\hat s_{XY} \leq \eps_r - \nf{\eps_r}{3} }
		\leq  \prob{t \cdot \hat s_{XY} \leq t \cdot \expect{\hat s_{XY}} - t \cdot \nf{\eps_r}{3} }
		\leq \exp\left(- \frac{2 \eps_r^2 t^2/9}{t (2H_k)^2} \right)
		= \exp\left(- \frac{ \eps_r^2 t}{18 H_k^2} \right).		 \notag
	\end{equation}
	We can similarly use Hoeffding to upper bound the probability of a $\nf{\eps_r}{3}$-bad pivot appearing at least $\nf{2\eps_r}{3}$-good.
	Union bounding over both this and the $\nf{\eps_r}{3}$-bad pivots, we have that the probability of failure within a step is at most
	\[
		\sum_{\eps_r\text{-good } q} \prob{q \text{ seems } \nf{2\eps_r}{3} \text{ bad}} + \sum_{\nf{\eps_r}{3}\text{-bad } q} \prob{q \text{ seems } \nf{2\eps_r}{3} \text{ good}}
		\leq  (mk)^{\ell_r} \cdot \exp\left(- \frac{\eps_r^2 t}{18 H_k^2} \right)
		\leq \nf{\delta'}{\ell_r},
	\]
	where $\delta'$ is a per-round failure probability.

	To make sure our failure is at most $\nf{\delta}{2}$ over all $\log_2\log_2 k$ rounds, we set $\delta' = \delta / (2\log_2\log_2 k )$, recall that $H_k \leq (\log k + 1)$, and rearrange to obtain
	\[
		t
		\geq 18 \cdot \eps_r^{-2} \cdot (\log k + 1)^2 \left( \ell_r \log mk + \log\nf 1 \delta + \log( 2  \ell_r \log_2 \log_2 k )\right).
	\]
	Since the number of samples we choose on \lineref{line:sample-cs-pav-second-samples} satisfies this bound, our combined failure probability over all rounds is at most $\nf{\delta}{2}$.

	This shows the total failure probability of \Cref{alg:sample-cs-pav} is at most $\delta$ and concludes the proof.
\end{proof}

It is also worth noting the guarantee provided by \Cref{thm:cs-pav-correctness} for relaxations of \jr{}.
Recall (\Cref{def:jr-eps}) that for $\alpha < 1$, the relaxed representation axiom $\alpha$-\ejrp{} equals \ejrpe{} for $\alpha \defeq (\frac{k}{k+1} + k\eps)^{-1}$.

\begin{restatable}[]{corollary}{relaxedejrpub}
	\label{cor:sample-complexity-ejrp-relaxed-ub}
	For any $\alpha < 1$, the sample complexity of $\alpha$-\ejrp{} is $\tilde O(k^{2})$.
\end{restatable}

To put this in context, consider the distinct problem of distinguishing a \emph{single} candidate with support $\probover{A \sim \cP}{c\in A} = \nf{2}{k}$ from a candidate $c'$ with support $\probover{A \sim \cP}{c'\in A} = \nf{1}{2k}$.
By \Cref{lem:separating-means-from-samples-new}, this has sample complexity $\Theta(k^2\log \nf{m}{\delta})$---nearly as large as our upper bound on that of $\alpha$-\jr{}.

As we will shortly see, this suffices to establish that $\alpha$-\jr{} is strictly easier than \jr{} from the perspective of sample complexity.

\section{Lower Bounds}
\label{sec:lowerbounds}

We approach our main lower bound construction indirectly.
We begin by demonstrating that \jre{}, and therefore the Droop-proximate versions of all axioms in the \jr{} family, require many samples to satisfy in the limit as $\eps$ approaches $0$.
For this, we will require a few definitions.
\begin{definition}[Total Variation distance]
	\label{def:tv-dist}
	For two distributions $\cP$ and $\cQ$ over a common ground set $\Omega$ and event $\sigma$-algebra $\cF$, the \emph{total variation distance} is $\dtv(\cP, \cQ) \defeq \sup_{F \in \cF} \abs{\cP(F) - \cQ(F)}$.
\end{definition}

\begin{definition}[Kullback-Leibler divergence]
	\label{def:kl-divergence}
	For two distributions $\cP$ and $\cQ$ over a common discrete sample space  $\Omega$, the \emph{Kullback-Leibler divergence} is given by $\dkl(\cP, \cQ) \defeq \sum_{x \in \Omega} \cP(x) \log\frac{\cP(x)}{\cQ(x)}$.
\end{definition}

A standard approach for establishing sample complexity lower bounds is to identify a family of instances for which the set of feasible algorithmic outputs is pairwise disjoint, but for which the task of distinguishing between the instances in the family requires many samples.
To this end, consider a family of \emph{partition profiles} $\cI =\{\cI_1, \ldots, \cI_{k+1}\} $.
For each partition profile $\cI_i$, we have sufficiently large support $\cI_i(\{j\}) = \frac{1}{k+1} + \eps$ for all $j \neq i$ to guarantee representation according to \jre{}, and $\cI_i(\{i\}) = \frac{1}{k+1} - k \eps$.
The profile $\cI_i$ is depicted in \Cref{fig:partition-all-large}.

This family of profiles is useful for us because on the one hand, the TV distance between each pair is small: we have $d_{TV}(\cI_i, \cI_j) = (k+1)\eps$.
But on the other hand, their committees satisfying \jre{} are disjoint: \jre{}$(\cI_i) = \{ [k+1] \setminus \{i\}\}$, while \jre{}$(\cI_j) = \{ [k+1] \setminus \{j\}\}$.
Therefore for this family distinguishing which $\cI_i$ the samples are coming from reduces to finding a \jre{} committee for the underlying unknown population, and so hypothesis testing lower bounds similar to \Cref{lem:separating-means-from-samples-new} apply.

\begin{figure}
\centering

\def\blockheight{.65 cm}
\def\largeblockwidth{1.2 cm}
\def\smallblockwidth{.9 cm}
\def\blockgap{-0.2 cm}

\begin{tikzpicture}[
    block/.style={
        draw,
        line width=0.8pt,
        minimum height=\blockheight,
        align=center,
        text width=\largeblockwidth,
        outer sep=0pt,
        fill=blue!10,
        rounded corners=2mm
    },
    smallblock/.style={
        draw,
        line width=0.8pt,
        minimum height=\blockheight,
        align=center,
        text width=\smallblockwidth,
        outer sep=0pt,
        rounded corners=2mm
    },
    labelstyle/.style={
        below=0.2cm of #1,
        anchor=north,
        font=\footnotesize,
        align=center,
    }
]

	% Block A_1
	\node[block] (A1) at (0,0) {$A_1$};
	\node[labelstyle=A1, text width=\largeblockwidth] {$ \frac{1}{k+1} + \epsilon $};

	% Block A_2
	\node[block] (A2) at ([xshift=\largeblockwidth+\blockgap]A1.east) {$A_2$};
	\node[labelstyle=A2, text width=\largeblockwidth] {$ \frac{1}{k+1} + \epsilon $};

	% Ellipses 1
	\node at ([xshift=0.7cm]A2.east) {\dots};

	% Block A_i
	\node[smallblock] (Ai) at ([xshift=\largeblockwidth+\blockgap+\largeblockwidth+\blockgap]A2.east) {$A_i$};
	\node[labelstyle=Ai, text width=\largeblockwidth] {${\frac{1}{k+1} - k \epsilon} $};

	% Ellipses 2
	\node at ([xshift=0.6 cm]Ai.east) {\dots};

	% Block A_{k+1}
	\node[block] (Akplus1) at ([xshift=\smallblockwidth+\largeblockwidth+\blockgap]Ai.east) {$A_{k+1}$};
	\node[labelstyle=Akplus1, text width=\largeblockwidth] {$ \frac{1}{k+1} + \epsilon $};
\end{tikzpicture}

\caption{The partition profile $\cI_i$. Here $\cC = [k+1]$, each approval set is $A_{j} = \{j\}$, and $\prof_i(A_j)$ is below each bloc.}
\label{fig:partition-all-large}
\end{figure}

\begin{restatable}[]{theorem}{smallepslb}
	\label{thm:small-eps-lower-bound}
	There exists a constant $c$ such that for all $k \geq 2$, and $\delta \leq \nf{1}{3}$, and $\eps \leq k^{-2}/8$, no sample-based rule $(t,\delta)$-satisfies \jre{} for $t \leq c\eps^{-2}k^{-3}$.
	That is, the sample complexity of \jre{} is $\Omega(\eps^{-2}k^{-3})$.
\end{restatable}

This instance family and proof approach suffice to show that the sample complexity of \jre{} increases at least quadratically in $\nf{1}{\eps}$, and that therefore Droop \jr{} is unboundedly large.
However it gives a particularly uninspiring lower bound for (Hare) \jr{}.
This is because for $\eps = (k(k+1))^{-1}$, the size of the small group in $\cI_i$ is $\cI_i(A_i) = \frac{1}{k + 1} - \frac{k}{k(k+1)} = 0$, meaning candidate $i$ has no support.
In this case, the coupon collector's problem tells us that distinguishing $\cI_i$ and $\cI_j$ requires only $\tilde O(k)$ samples.

Can this construction be saved?
If we decrease the gap between the largest and smallest blocs from $(k+1) \eps$ down to $\approx \eps$, then we can still constrain the behavior of a \jr{} rule by enforcing that a \emph{single} distinguished candidate be part of any \jr{} committee.
This new instance family $\cJ$ is illustrated in \Cref{fig:partition-all-small}.
Retracing the proof of \Cref{thm:small-eps-lower-bound} with this smaller gap yields an $\Omega(\eps^{-2}k^{-1})$ lower bound on the number of samples required to distinguish $\cJ_i$ from $\cJ_j$, which is $\Omega(k^3)$ for \jr{}.

However in replacing instance family $\cI$ with $\cJ$, we have introduced a new problem: producing a \jr{} committee for unknown $\cJ_i$ no longer requires distinguishing $\cJ_i$ from $\cJ_j$; each subset of instances $\cJ_1, \ldots, \cJ_\ell$ has $k-\ell$ \jr{} committees in common.
Worse yet, a uniformly random committee satisfies \jr{} for $\cJ_i$ with probability $1-\nf{1}{k+1}$.
Nevertheless, we can overcome these obstacles and derive a meaningful lower bound by taking a different proof approach.

%%%%
\subsection{An \texorpdfstring{$\tilde\Omega(k^3)$}{Omega(k\^3)} Lower Bound for the JR Family}
\label{sec:lowerbounds-k3}

We proceed by instead making all but one bloc too small to receive a guarantee.
We let these partition profiles be denoted $\cJ_1, \ldots \cJ_{k+1}$, and construct them as in \Cref{fig:partition-all-small}, and denote this family of instances by $\cJ$.
For \jr{}, this instance family is described in \Cref{ex:hard-partition}.

\begin{example}
\label{ex:hard-partition}
	Given $\eps$ and $k$, partition the support distribution $\cP$ into one bloc $B_i$ of size $\frac{1}{k+1} +\eps$ and $k$ blocs of $B_j$ of size $\frac{1}{k+1} - \nf{\eps}{k}$.
	The difference between these support sizes is $\eps \cdot \frac{k+1}{k}$.
	Note that any collection of $k$ candidates that satisfies \jr{} must include one supported by $B_i$.
	Let this instance be $\cJ_i$, and choose $i \sim [k+1]$ uniformly at random to form the randomized instance $\cJ$.
\end{example}

\begin{figure} %[h]
\centering
\def\blockheight{.65 cm}
\def\largeblockwidth{1.05 cm}
\def\smallblockwidth{1.3 cm}
\def\blockgap{-0.15 cm}

\begin{tikzpicture}[
	    block/.style={
	        draw,
	        line width=0.8pt,
	        minimum height=\blockheight,
	        align=center,
	        text width=\largeblockwidth,
	        outer sep=0pt,
	        rounded corners=2mm
	    },
	    smallblock/.style={
	        draw,
	        line width=0.8pt,
	        fill=blue!10,
	        minimum height=\blockheight,
	        align=center,
	        text width=\smallblockwidth,
	        outer sep=0pt,
	        rounded corners=2mm
	    },
	    labelstyle/.style={
	        below=0.2cm of #1,
	        anchor=north,
	        font=\footnotesize,
	        align=center,
	    }
	]

	% Block A_1
	\node[block] (A1) at (0,0) {$A_1$};
	\node[labelstyle=A1, text width=\smallblockwidth] {${\frac{1}{k+1} - \nf{\eps}{k}}$};

	% Block A_2
	\node[block] (A2) at ([xshift=\smallblockwidth+\blockgap]A1.east) {$A_2$};
	\node[labelstyle=A2, text width=\largeblockwidth] {${\frac{1}{k+1} - \nf{\eps}{k}}$};

	% Ellipses 1
	\node at ([xshift=0.6cm]A2.east) {\dots};

	% Block A_i
	\node[smallblock] (Ai) at ([xshift=\largeblockwidth+\largeblockwidth]A2.east) {$A_i$};
	\node[labelstyle=Ai, text width=\largeblockwidth] {${\frac{1}{k+1} + \eps} $};

	% Ellipses 2
	\node at ([xshift=0.6 cm]Ai.east) {\dots};

	% Block A_{k+1}
	\node[block] (Akplus1) at ([xshift=\smallblockwidth+\blockgap+\largeblockwidth+\blockgap]Ai.east) {$A_{k+1}$};
	\node[labelstyle=Akplus1, text width=\largeblockwidth] {${\frac{1}{k+1} - \nf{\eps}{k}}$};
\end{tikzpicture}

\caption{The partition profile $\cJ_i$. Here $\cC = [k+1]$, each approval set is $A_j = \{j\}$, and $\prof_i(A_j)$ is below each bloc.}
\label{fig:partition-all-small}
\end{figure}

We derive a nontrivial lower bound from $\cJ$ in two stages.
As before, we consider the uniform distribution $\cJ$ over the partition profiles $\cJ$; this will comprise our sample-hard randomized instance.
We first characterize the behavior of the optimal sample-based rule $f^*$ which faces an unknown choice of $\cJ_i \sim \cJ$, for any number of samples $t$.
A proof appears in \Cref{app:lower-bounds}.

\begin{restatable}[]{lemma}{optimalrulestructure}
	\label{lem:optimal-part-rule-structure}
	For $\cJ_i \sim \cJ$ u.a.r. and any number of samples $t$ from unknown $\cP_i$, the sample-optimal rule $f^*$ chooses the top $k$ candidates according to empirical support.
\end{restatable}

We now lower bound the probability that $f^*$ errs when facing a draw from $\cD_\cJ$, as a function of $t$.
This proceeds first by observing that the sample mean of the large bloc support and the minimum sample mean of all other blocs are negatively correlated.
Once this is established, we use a combination of concentration and anticoncentration inequalities to lower bound the failure event.
\begin{restatable}[]{lemma}{optimalrulemistakelb}
	\label{lem:optimal-part-rule-mistake-LB}
	There exist constants $c_1, c_2, c_3$ such that for any given $k \geq 2$, any given $\eps \leq c_1 /k$, and any $\delta \leq c_2 k^2 \eps^2$,
	if $t \leq c_3 k^{-1} \eps^{-2}$ then over $t$ samples the large bloc $i$ in the sampled profile $\cJ_i$ exhibits the smallest empirical support with probability at least $\delta$.
\end{restatable}
\begin{proof}
	As we observed in the proof of \Cref{thm:small-eps-lower-bound}, for a given $\cJ_i \in \cJ$ the distribution of $\bar n$, the histogram of the supporters of candidates $j \in [k+1]$ from the sample $\cA=(A^1, \ldots, A^t)$, follows a multinomial distribution $M_{k,\eps, i,t} = \Multin(t; \psma, \ldots, \pbig, \ldots, \psma)$, where for ease of notation we let $\psma \defeq \nf{1}{k+1} - \nf{\eps}{k}$ and $\pbig \defeq \nf{1}{k+1} + \eps$.
	The failure event for $f^*$ is precisely that $\bar n_i \leq \min_{j \neq i} \bar n_j$ when the unknown distribution is $\cJ_i$.
	We will pick a threshold $\tau$ as a function of $k$ and $\eps$, and lower bound the failure probability by
	\begin{align}
		 \probover{\substack{i \sim [k+1] \\ \cA\sim \prof_i^t}}{ f^*(\cA) \not \in \jr{}_\eps(\cJ_i)}
		 &=  \frac{1}{k+1} \sum_i \probover{\cA\sim \prof_i^t}{ \bar n_i \leq \min_{j \neq i} \bar n_j} \notag \\
		 &=  \probover{\cA\sim \prof_1^t}{ \bar n_1 \leq \min_{j \neq 1} \bar n_j} \notag \\
		 &\geq \probover{\cA\sim \prof_1^t}{ \bar n_1 \leq \tau \: \land \: \tau \leq \min_{j \geq 2} \bar n_j}. \label{eq:kcubed-lb-threshold-bound}
	\end{align}
	This proceeds in two steps.
	First, we show that these events are positively correlated; therefore our failure event is lower bounded by the product of their individual probabilities.
	Then we lower bound the probability of each event.

	Positive correlation can be shown via Chebyshev's sum inequality.
	For fixed $\tau$, let $\Gamma_t$ denote the event that $\bar n_1 < \tau$ over $t$ samples $\cA$, and let $\Lambda_t$ denote the event that $\min_{j \geq 2} \bar n_j \geq \tau$ over $\cA$.
	Then
	\begin{align*}
		\probover{\cA\sim \prof_1^t}{\Gamma_t \: \land \: \Lambda_t }
		&= \sum_{t' =0}^t \prob{\Lambda_t \:\land \: t' \leq \tau \:\vert\: \bar n_1 = t' } \cdot \prob{\bar n_1 = t' } \\
		&= \sum_{t'} \prob{\Lambda_t \:\vert\: \bar n_1 = t' } \cdot \1\{t' \leq \tau \} \cdot \prob{\bar n_1 = t' } \\
		&\geq \sum_{t'} \prob{\Lambda_t \:\vert\: \bar n_1 = t' }\cdot \prob{\bar n_1 = t' }
		\cdot \sum_{t'} \1\{t' \leq \tau \} \cdot \prob{\bar n_1 = t' } \\
		&= \prob{\Lambda_t} \cdot \prob{\Gamma_t}.
	\end{align*}
	Here our invocation of Chebyshev's sum inequality required that both $\prob{\Lambda_t \:\vert\: \bar n_1 = t' }$ and $\1\{t' \leq \tau \}$ are decreasing in $t'$.
	This is clearly true for the latter; for the former it holds because for fixed $t$, the remaining $t -t'$ samples are partitioned into $\bar n_2, \ldots, \bar n_{k+1}$ according to $\Multin(t-t'; \psma, \ldots, \psma)$.
	As $t'$ increases the number of samples to distribute decreases, and so $\min_{j \geq 2} \bar n_j$ will decrease also.

	With this in hand, we turn to lower bounding each probability.
	Recall (or observe) that for multinomials, each individual coordinate follows a binomial distribution; for $(\bar n_1, \ldots, \bar n_{k+1}) \sim \Multin(t; p_1, \ldots, p_{k+1})$, each $\bar n_j$ is distributed according to $\Bin(t,p_j)$.
	This will permit us to use binomial anticoncentration to lower bound $\prob{\Gamma_t}$.

	Before doing this, we address $\prob{\Lambda_t}$.
	Unfortunately, the negative association of multinomial coefficients works against us in our effort to lower bound $\prob{\Lambda_t}$ in terms of the behavior of its constituent binomial parts.
	We instead resort to a union bound, and find
	\begin{align}
		\prob{\Lambda_t} = \prob{\land_{j \geq 2} \bar n_j \geq \tau } = 1-\prob{\lor_{j \geq 2} \bar n_j < \tau } \geq 1 - \sum_{j \geq 2} \prob{\bar n_2 \leq \tau } = 1 - k\cdot \prob{X \leq \tau}, \label{eq:kcubed-union-bound-on-min}
	\end{align}
	where $X \sim \Bin(t, \psma)$.
	From the standard Chernoff lower tail bound, we can derive that $\prob{X \leq \tau} \leq \exp\left( - \frac{(\E{}{X} - \tau)^2}{2 \E{}{X}} \right)$.
	Upper bounding this by $\nf{1}{2k}$ and plugging it into \eqref{eq:kcubed-union-bound-on-min}, we find that $\prob{\Lambda_t} \geq \frac{1}{2}$ so long as
	\begin{equation}
	\label{eq:tau-first-inequality}
		\tau \geq t \cdot \psma - \left(2 t \cdot \psma \cdot \log 2k \right)^{1/2}.
	\end{equation}

	We now turn to lower bounding $\prob{\Gamma_t}$.
	For this we will use standard forms of binomial anti-concentration in the lower tail e.g. \cite[Equation 5]{zhu22nearly}.
	In particular, for $B \sim \Bin\left(t, \pbig \right)$ we have
	\begin{align}
		\prob{\Gamma_t} > \prob{B = \tau}
		&\geq \frac{\sqrt{t}}{\sqrt{8 \tau (t-\tau)}} \cdot \exp\left(-t \dkl \left(\tau/t \:\|\: \pbig \right) \right) \notag
	\end{align}
	From here, we seek to upper bound this KL divergence.
	Since $\dkl$ is upper bounded by the $\chi^2$-divergence, we can start with
	\begin{align*}
		\dkl(\tau/t \| p_\uparrow) &\leq \chi^2(\tau/t, p_\uparrow) = \frac{\left( \frac{\tau}{t} - p_\uparrow \right)^2}{\frac{\tau}{t}(1 - \frac{\tau}{t})},
	\end{align*}
	and so in order for $\prob{\Gamma_t} \geq \delta$ a sufficient condition is
	\begin{align}
		\frac{\sqrt{t}}{\sqrt{8 \tau (t-\tau)}} \cdot \exp\left(-t \cdot \frac{\left( \frac{\tau}{t} - p_\uparrow \right)^2}{\frac{\tau}{t}(1 - \frac{\tau}{t})} \right) \geq \delta. \label{eq:tau-second-inequality}
	\end{align}

	Our goal therefore reduces to simultaneously satisfying \eqref{eq:tau-first-inequality} and \eqref{eq:tau-second-inequality} for the largest possible $t$.
	To make this easier, suppose that $\tau = t \cdot \left(\frac{1}{k+1} - c \eps\right)$ for some $c$.
	Then \eqref{eq:tau-first-inequality} is true so long as
	\begin{equation}
		t \leq  \frac{ \log 2k }{2 c^2}
		\cdot (k+1)^{-1}\cdot \eps^{-2},
	\end{equation}
	provided $\eps \leq \nf{1}{4}$.
	Meanwhile \eqref{eq:tau-second-inequality} is true so long as
	\begin{equation}
		t \leq  \frac{\left(\log \nf{1}{\delta} - \nf{1}{2} \log 8 \tau \right)}{4 (c+1)^2}
		(k+1)^{-1}\cdot \eps^{-2},
	\end{equation}
	provided that $c\eps \leq (k+1)^{-1}/2$.

	Choosing $c = \log^{1/4} k$ to balance terms gives a slightly larger final $t$ but makes the parameter bounds more cumbersome.
	Therefore we fix $c$ to be a constant, and conclude that there are constants $c_1, c_2, c_3$ such that for all $k \geq 2$, for all $\eps \leq c_1\cdot k^{-1}$, for all $\delta \leq c_2 \cdot k^2 \cdot  \eps^2$, if $t \leq c_3 \cdot k^{-1} \cdot \eps^{-2}$ then $f^*$ with $t$ samples fails \jre{} with probability strictly more than $\delta$.
	This concludes the proof.
\end{proof}

We then combine \Cref{lem:optimal-part-rule-structure,lem:optimal-part-rule-mistake-LB} to argue about all sample-based rules:
\begin{theorem}
	\label{thm:k3-lower-bound-jr}
	There exist constants $c_1, c_2, c_3$ such that for all $k \geq 2$ and all $\eps \leq c_1 k^{-1}$, if $\delta \leq c_2 k^2 \eps^2$ then no sample-based rule $(t,\delta)$-satisfies \jre{} for $t \leq c_3\eps^{-2}k^{-1}$.
	That is, for all $k$ and sufficiently small $\eps$, there are instances with $m, \nf{1}{\delta} = \poly(k)$ for which the sample complexity of \jre{} is $\Omega(\eps^{-2}k^{-1})$.
\end{theorem}
\begin{proof} %[Proof of \Cref{thm:k3-lower-bound-jr}]
	Since $f^*$ is the optimal rule for this hard instance (\Cref{lem:optimal-part-rule-structure}), the lower bound on the failure probability of $f^*$ established in \Cref{lem:optimal-part-rule-mistake-LB} also holds for any rule when faced with $\cJ$.
\end{proof}

%%%
\subsection{Dependence on \texorpdfstring{$\log m$}{log m} is Necessary}
\label{sec:lowerbound-m-dependence}

Since our lower bound profiles have $m = k+1$ candidates, one might ask whether the sample complexity dependence on $\log m$ in our upper bounds is necessary.
Is this problem truly harder when $m$ is extremely large relative to $k$?
We can also show that a dependence on $m$ in the sample complexity is necessary.
To do this we will use \Cref{ex:sample-complexity-m-dependence}.

\begin{example}[Sampling \jr{} is hard for large $m$]
\label{ex:sample-complexity-m-dependence}
	For any $k \geq 2$ and $d \geq 1$ we will construct an election $E = (\cands, \prof, k)$.
	Consider a ground set of $k^d$ types of voters, and the collection of candidates $\cands$ consisting of one candidate supported by \emph{every} subset of $k^{d-1}$ such types.
	The profile $\cP$ is uniform over the voter types.
	Now consider $t$ and samples $\cA \sim \cP^t$.
	If only $\abs{\cup \cA} \leq k^d/2$ voter types are sampled, let $W = f^*(\cA)$ be the optimal committee chosen based on $\cA$.
	Even if the chosen $W$ perfectly partitions $\cup \cA$ by support, it will intersect only approximately $k^d(\nf{1}{2} + \nf{1}{2}(1-\nf{1}{e})) = k^d(1 - \Omega(1))$ types.
	Hence for $k \geq 10$ there will be some $c \in \cands$ with all $k$ of its types uncovered by $W$, witnessing a \jr{} violation.
\end{example}

\begin{restatable}[]{observation}{complexitymdependence}
	\label{thm:sample-complexity-depends-on-m}
	For sufficiently large $k$ and given $\delta \leq \nf{1}{2}$, for all $d \in \mathbb{N}$ there exist elections $E_d$ for which the sample complexity of \jr{} is $\Omega(k^d)$.
\end{restatable}

This does not contradict our positive results in \Cref{sec:upperbounds} precisely because no dependence on $m$ is specified.
Our task is now to determine the lower bound dependence on $m$ that \Cref{ex:sample-complexity-m-dependence} implies.
\begin{proposition}
	\label{prop:jr-complexity-log-m-dependence}
	For any $n$ and any $\delta < 1$, the sample complexity of \jr{} is at least $\Omega \left( k \cdot \frac{\log m}{\log \log m}\right)$.
\end{proposition}
\begin{proof} %[Proof of \Cref{prop:jr-complexity-log-m-dependence}]
	In \Cref{ex:sample-complexity-m-dependence} the number of samples necessary to satisfy \jr{} for any $\delta <1$ is strictly more than $k^d/2$, while the number of candidates is $m = \binom{k^d}{k^{d-1}}$.
	Therefore $\log m = \Theta(dk^{d-1} \log k)$, so for any constant $d\geq 2$ we have for some constants $c$, $c'$ that the number of necessary samples is
	\[
		t \geq k^d/2 \geq c\cdot  k \cdot \frac{\log m}{\log k} \geq c'\cdot  k \cdot \frac{\log m}{\log \log m}. \qedhere
	\]
\end{proof}

%%%%
\subsection{Impossibility for Droop JR}
\label{sec:lowerbounds-droop}

The analyses of $\eps$-\mkc and \samplepav suffice to show that the sample complexity of \jre{} increases boundedly as $\eps$ decreases.
Conversely, our lower bound instances can be modified to show that in the limit as $\eps$ approaches $0$ and we approach Droop \jr{}, \jre{} becomes unboundedly hard to satisfy.
This unlearnability of Droop \jr{} can be made precise in the following sense:
\begin{corollary}[Droop \jr{} is unlearnable]
\label{thm:droop-jr-unlearnable}
	For any $\delta \leq \nf{1}{k+1}$ and any function $f(k, m, \delta)$, there exists an instance and an $\eps > 0$ for which \jre{} (and therefore Droop \jr{}) is not $(f(k, m, \delta), \delta)$-satisfiable.
\end{corollary}
\begin{proof}
	Fix $\delta$ and $t = f(k,m,\delta)$ and apply \Cref{thm:small-eps-lower-bound} or \Cref{thm:k3-lower-bound-jr} with small enough $\eps$.
\end{proof}

This can be illustrated more vividly in the case of uniform samples from a finite population $[n]$.
\begin{observation}
\label{claim:droop-jr-unlearnable-finite-population}
	Even sampling without replacement, for $\delta < 1/(k+1)$ Droop \jr{} requires $t=n$.
\end{observation}
\begin{proof}[Proof of \Cref{claim:droop-jr-unlearnable-finite-population}]
	For any fixed $k$, consider an instance as in \Cref{fig:partition-all-small} with $n = (k+1)^2$ and $k$ blocs of size $k+1$, and one of size $k+2$.
	The last bloc alone reaches the Droop quota, and after $n-1$ samples all blocs will have the same number of sampled supporters and be indistinguishable.
\end{proof}
This is outside the regime for which we established the asymptotic equivalence of sampling with and without replacement (\Cref{claim:with-vs-without-replacement}), and so the requirements for sampling distinguishable ballots with replacement is still worse: it requires $t = \Omega(n \log n)$ by the coupon collector's problem.

%%%%
\subsection{Certifying Proportional Committees}
\label{sec:lowerbounds-certifying}

Consider the task of \emph{certifying} that a specific choice of $k$ candidates $W \subseteq \cands$ satisfies \jr{}.
That is, given a committee $W$, how many samples are necessary to decide whether or not $W$ satisfies \jr{}, with success probability at least $1-\delta$?
This is more in the spirit of \emph{property testing}, wherein we are given a combinatorial object and would like to efficiently decide, given sample access to some data, whether it has a property of interest \cite{goldreich98property} \cite{goldreich17introduction}.

Perhaps surprisingly, this is impossible even for \jr{}, as well as for $\alpha$-\jr{} for $\alpha < 1$.
This is for the same reason as we saw in \Cref{sec:lowerbounds-droop}.
For a given committee $W$, there is no way to know how many samples will be necessary to decide whether the maximum marginal coverage of an unchosen candidate is $<1/k$ or $\geq 1/k$, without first knowing how \emph{far} $W$ is from satisfying \jre{}.

If instead we have a committee which either satisfies Droop \jr{} or doesn't satisfy Hare \jr{}, then distinguishing the two requires $O(\eps^{-2} \log \nf{m}{\delta})$ samples by implementing a single greedy step.
Is this optimal via \Cref{ex:hard-partition}?
Although parameterizing the difficulty of testing whether a committee satisfies even $\alpha$-\jr{} by its distance from feasibility is perhaps unsatisfying, gaps of this nature are often necessary in the formulations of property testing problems.

This perhaps offers us a distinct way to understand how \jr{} differs from $\nf{1}{k^2}-\MkC$; our lower bound in \Cref{sec:coverage} apparently demonstrates that testing whether a given $k$-subset is coverage-optimal or $\eps$-far from coverage-optimal requires $\tilde\Omega(k \eps^{-2})$ samples.

%%%
\section{Improved Bounds for Structured Domains}
\label{sec:vcdim}

Having established worst-case upper and lower bounds on the sample complexity of satisfying multiwinner proportionality axioms, it is natural to ask whether the landscape changes if we restrict our attention to structured families of election instances, and if so, how.

A common beyond-worst-case assumption for approval elections is that they arise from \emph{candidate interval} (CI) domains.
This means the approval matrix exhibits the consecutive ones property (C1P): there is a global ordering $\pi$ of $\cands$ such that any $A \sim \prof$ forms a contiguous subsequence with respect to $\pi$.
This can be understood as the approval analog of the property of single-peakedness for ranked preference profiles \cite{elkind15structure}.
Our lower bound instances (\Cref{fig:partition-all-small}) are in fact partition instances and hence satisfy CI, so the sample complexity of satisfying \jre{} for even partition and CI profiles remains $\Omega(k^{-1} \eps^{-2})$.

Conversely, and more generally, we might hope that the PAC learning view of approximate coverage or $\pavsc$ maximization (\Cref{alg:eps-mkc-brute-force} and \Cref{alg:sample-pav}) might offer complexity improvements when the space of solutions is constrained.
For instance, an approval profile is \emph{$d$-Euclidean} if there is an embedding $\rho: \cands \rightarrow \mathbb{R}^d$ such that each approval set $A \sim \cP$ takes the form $A = \{c \in \cands: \|a - \rho(c)\|_2 \leq r\}$ for some choice of $a$ and $r$ \cite{elkind15structure}.
Here the VC dimension of the space of individual candidates is $d+1$ \cite{har2016towards}, and so the space of unions of $k$ candidates has VC dimension $O(dk \log k)$ \cite{blumer89learnability}.
As a result, for $d$-Euclidean domains we may replace the $\log m$ term in our coverage and $\pavsc$ guarantees (\Cref{thm:eps-mkc-sample-complexity} and \Cref{prop:sample-pav-correct}) with $\log k$; an uninspiring improvement unless $m$ is very large.

%%%%
\subsection{Popular Candidates and High-Coverage Committees}
Fortunately, revisiting our sample-efficient approach (\Cref{alg:sample-ws-ls-pav}) allows us to make improved claims for a distinct class of instances: those containing one or more high-support candidates.
\begin{restatable}[JR for popular-candidate instances]{theorem}{bwclargecandub}
\label{thm:large-candidate-ub}
	For any $\eps>0$ and any constant $\beta >0$, for any election $E$, if there is a candidate $c$ for which $\probover{A \sim \prof}{c \in A} \geq \beta$, then only $\tilde O(k^{2}) $ samples are necessary to identify a \jre{} committee for $E$.
\end{restatable}

Beyond-worst-case guarantees on the sample complexity of \jr{} are also possible for any instance where there exists a committee with high overall voter coverage.
\begin{restatable}[JR for high-coverage instances]{theorem}{bwclargecovgub}
	\label{thm:large-covg-ub}
	For any $\eps>0$ and any constant $\gamma < 1$, for any election $E$, if there is a committee $W$ for which $\probover{A \sim \prof}{A \cap W \neq \emptyset} \geq 1 - \nf{\gamma}{k+1}$, then only $\tilde O(k^{2}) $ samples are necessary to identify a \jre{} committee for $E$.
\end{restatable}

%%%%
\subsection{The Structure of Hard Instances}
In fact, we can say more.
Our positive result for relaxed \jr{} (\Cref{cor:sample-complexity-ejrp-relaxed-ub}) also tells us something about the hard instances.
Any election $E$ for which \jr{} is harder than $\tilde O(k^3)$ must `look like' our hard instance (\Cref{fig:partition-all-small}), in the sense that it must have committees that approximately equipartition the electorate according to candidate support:
\begin{restatable}[Structure of hard instances]{proposition}{hardinstancestructure}
\label{claim:hard-instance-characterization}
	Let $E=(C, \cP, k)$ be an election.
	Either a \jr{} committee can be identified for $E$ using $\tilde O(k^3)$ samples, or there
	exists a committee
	$W$ for which candidates have
	\begin{itemize}
		\item Approximately $\nf{1}{k}$-exclusive support: $\probover{A \sim \cP}{A \cap W = \{c\}} \geq \frac{1}{k} \left(1 - \frac{1}{\sqrt{k}} \right)$ for all $c \in W$, and
		\item Small pairwise overlap: $\probover{A \sim \cP}{\abs{A \cap W} \geq 2} \leq \frac{1}{\sqrt{k}}$.
	\end{itemize}
\end{restatable}

This is a striking amount of structure, and suggests that instances for which \jr{} is sample-intensive to satisfy should be rare in practice.
We will also see evidence for this in \Cref{sec:experiments}.

\section{Implications for Ranked Ballots}
\label{sec:ranked}

Although our main focus is on approval-based committee voting, our primary lower bound is versatile enough to imply sample complexity lower bounds on finding proportional committees in the setting of multiwinner voting based on ranked preferences.
In this domain, a standard criterion is that of solid coalition proportionality, due to \citet{dummett84voting}.
Like \jr{} this is defined with respect to a quota, which is generally taken to be a value between a $\nf{1}{k+1}$ and a $\nf{1}{k}$-proportion of the total electorate size.
The former is more stringent and is referred to as the Droop quota \cite{droop81methods}, and the latter as the Hare quota.

In our notation, a multiwinner election from ranked preferences is again given by $E=(\cands, \prof, k)$, where $\prof$ is now a distribution over the set of rankings $\Sigma_{\cands}$ of $\cands$, and each sampled ballot $\sigma\sim\cP$ is a (we will suppose) strict and complete ranking of $\cands$.
Let $\sigma_{\leq r} = \{\sigma(1), \ldots, \sigma(r)\}$ be the $r$-prefix of $\sigma$.
\begin{definition}[Proportionality for solid coalitions]
	For $\eps > 0$, a committee $W \in \comms$ satisfies  \psce{} with respect to an election $E$ if, for all $C' \subseteq \cands$ and for all $\ell \in [k]$,
	\[
		\text{if} \qquad \quad  \Pr_{\sigma \sim \prof}\left[\sigma_{\abs{C'}} = C' \right] \geq \ell \cdot \left(\frac{1}{k+1} + \eps \right),
		\quad \qquad \text{then} \qquad \quad
		\abs{C' \cap W} \geq \min\left(\ell, \abs{C'}\right).
	\]
\end{definition}
As usual, note that Hare \psc{} corresponds to \psce{} for $\eps = (k (k+1))^{-1}$.

Borrowing notation from \citet{brill23robust}, for a given ranked election $E$ with approval profile $\prof$, let $E^r=(\cands, \prof^r, k)$ denote the \emph{approval} election derived by taking the top $r$ ranks from each ballot: for $\sigma \sim \prof$, we will define $A^r \sim \prof^r$ by $A^r \defeq \sigma_{\leq r}$.
We need only $r=1$.
\begin{observation}
	\label{obs:psc-implies-jr-for-top}
	For any $\eps >0$ and any $E$, if $W\in\comms$ satisfies \psce{} for $E$ then $W$ satisfies \jre{} for $E^1$.
\end{observation}
This is a natural consequence of the definitions of \jre{} and \psce{}, and can also be seen as a corollary of \cite[Proposition 2]{aziz20expanding} or \cite[Section 4]{brill23robust}.

Using \Cref{obs:psc-implies-jr-for-top}, \Cref{thm:k3-lower-bound-jr} implies the same lower bound for \psce{}.
\begin{corollary}
\label{thm:ranked-psc-lower-bound}
	For all $k\geq 2$ and all sufficiently small $\eps$, there are instances with $m, \nf{1}{\delta} = \poly(k)$ for which the sample complexity of \psce{} is $\Omega(\eps^{-2}k^{-1})$.
	In particular, there exist $c_1, c_2, c_3$ such that for all $k \geq 2$ and all $\eps \leq c_1 k^{-1}$, if $\delta \leq c_2 k^2 \eps^2$ then no sample-based rule $(t,\delta)$-satisfies \psce{} for $t \leq c_3\eps^{-2}k^{-1}$.
\end{corollary}

Hence the sample complexity of Hare \psc{} is also $\tilde\Omega(k^3)$.
Applying \Cref{obs:psc-implies-jr-for-top} to \Cref{thm:droop-jr-unlearnable}, it also follows that Droop \psc{} outcomes are unlearnable from samples in ranked-ballot elections.

\psc{} is satisfied by the single transferable vote (STV), as well as by the method of equal shares (MES) \cite{peters21proportional} and others in the expanding approvals rule (EAR) family \cite{aziz20expanding}.
Prior work has largely studied the tractability of recovering winning \cite{micha20can} or losing \cite{imber21probabilistic} candidates for single-winner rules from samples.
But determining the STV outcome is seemingly much more difficult than determining a \psc{} outcome, and while this is in some cases empirically tractable, beyond-worst-case assumptions appear necessary for efficient guarantees \cite{iceland24sampling}.

As mentioned in \Cref{sec:greedyrules} in regards to MES and Phragm\'en's rule, it is less clear how to construct and analyze sample-based versions of share-based rules.
Within a given round of EAR and conditioning on the prior choices of candidates, obtaining an $\eps$-accurate estimate of the number of voters supporting each remaining candidate in all top-$\ell$-votes ballot prefixes requires only $\tilde O(\eps^{-2})$ sampled ballots.
But the task of a sample-based EAR is in fact to estimate the amount of virtual credit available for each candidate, which depends on the uncertain support sizes of the previously chosen candidates.
While it seems possible this can be done, we leave this and the broader question of sample complexity upper bounds on \psce{} as a future challenge.

\section{Experiments}
\label{sec:experiments}

Worst-case bounds are all well and good, but for real instances, how quickly does the probability of satisfying \jr{} or \ejr{} increase with the number of sampled approval ballots $t$, for various $k$?
Which rules appear most sample-efficient in practice, and how many randomly sampled voters' preferences appear to be sufficient in order to identify a proportional winning committee?

%%%
\subsection{Experimental Approach}
We set out to answer these questions for the approval preferences expressed by voters in recent real-world participatory budgeting instances.
These datasets are curated and made available by PabuLib, an open library of participatory budgeting data \cite{faliszewski23participatory}.\footnote{\url{https://pabulib.org/}}
For each given participatory budgeting election, our approach is to ignore project costs and subsample voters.
\citet{szufa22how} have also used this approach to construct a statistical culture of approval elections.\footnote{\citet[Figure 4b]{szufa22how} also observe that finding \pav{} winners for even 1000-voter, 50-candidate elections takes only seconds, which comports with our experiments and bodes well for the practicality of even non-polytime sample-based rules.}

We choose a few non-representative instances to study.
We select the first three based on tests indicating they are challenging from the perspective of \jr{} and \ejrp{} for (utilitarian) approval voting (\av{}), and therefore likely not too easy for other rules.
These are ``Poland Warszawa 2017 rejon B'', a knapsack-voting PB election held in Warsaw in 2016; and ``Netherlands Amsterdam 621'' and ``Netherlands Amsterdam 631'', two knapsack-voting PB elections held in different parts of Amsterdam in 2022.
We also consider a pair of instances identified by \citet[Figure 8]{boehmer26understanding} as ones for which it is highly unlikely that a uniformly random committee satisfies \jr{} or \ejrp{}, for committee sizes in the range we consider.
These instances are ``Netherlands Amsterdam 622'', a 2022 PB election with each ballot constrained to at most 5 approvals, and ``Poland Warszawa 2022 Mokot\'ow'', a PB election with each ballot constrained to at most 15 approvals.

For each instance, we created $500$ ballot samples $\cA_t$ for $t=10, 20, \ldots, 250$ approval ballots uniformly at random and with replacement.\footnote{Our implementation uses the \texttt{abcvoting} Python package \cite{lackner23abcvoting}.}

On each sample, we ran (utilitarian) approval voting (\av{}), \pav{}, \seqpav{}, maximum coverage (\chambcou{}), and the method of equal shares (\mes{}) for even $2 \leq k \leq 30$ and evaluated whether the resulting committees satisfied \jr{} and \ejrp{}.
These rules were chosen in order to give good coverage to a range of committee voting methods and exhibit different combinations of attributes of interest.
\pav{} and \chambcou{} are two Thiele rules we provide guarantees for, which differ in terms of proportionality axioms---both satisfy \jr{}, but \pav{} satisfies \ejrp{} while \chambcou{} does not.
\mes also has strong proportionality guarantees, though we suspect it is more challenging to implement via samples; it is interesting to test this intuition.
\seqpav{} is a greedy method similar to rules we study, but it does not have the same theoretical guarantees and, in this implementation but unlike in our rules, uses all samples for every step.
Finally, \av{} is often used in practice but not at all oriented towards proportional outcomes; we include it in part to gauge the difficulty of satisfying proportionality on a given instance.

For each instance, we choose a particular value of $k$ to highlight, for instance $k=12$ in \Cref{fig:warsza_increasing_n_k12}.
Each is chosen by looking at the performance of the above rules on $n=50$ samples for even $2 \leq k \leq 30$ and choosing the committee size which appears most challenging; for instance, $k=12$ was chosen for ``Poland Warszawa 2017 rejon B'' based on \Cref{fig:warsza_increasing_k_n50}.
Plots for the other instances of success probability as a function of the number of samples for difficult committee sizes, and more in-depth views of the ``Poland Warszawa 2017 rejon B'', appear in \Cref{app:experiments}.

%%%
\begin{figure}
	\begin{subfigure}{.48\textwidth}
		\centering
		\includegraphics[width=\linewidth]{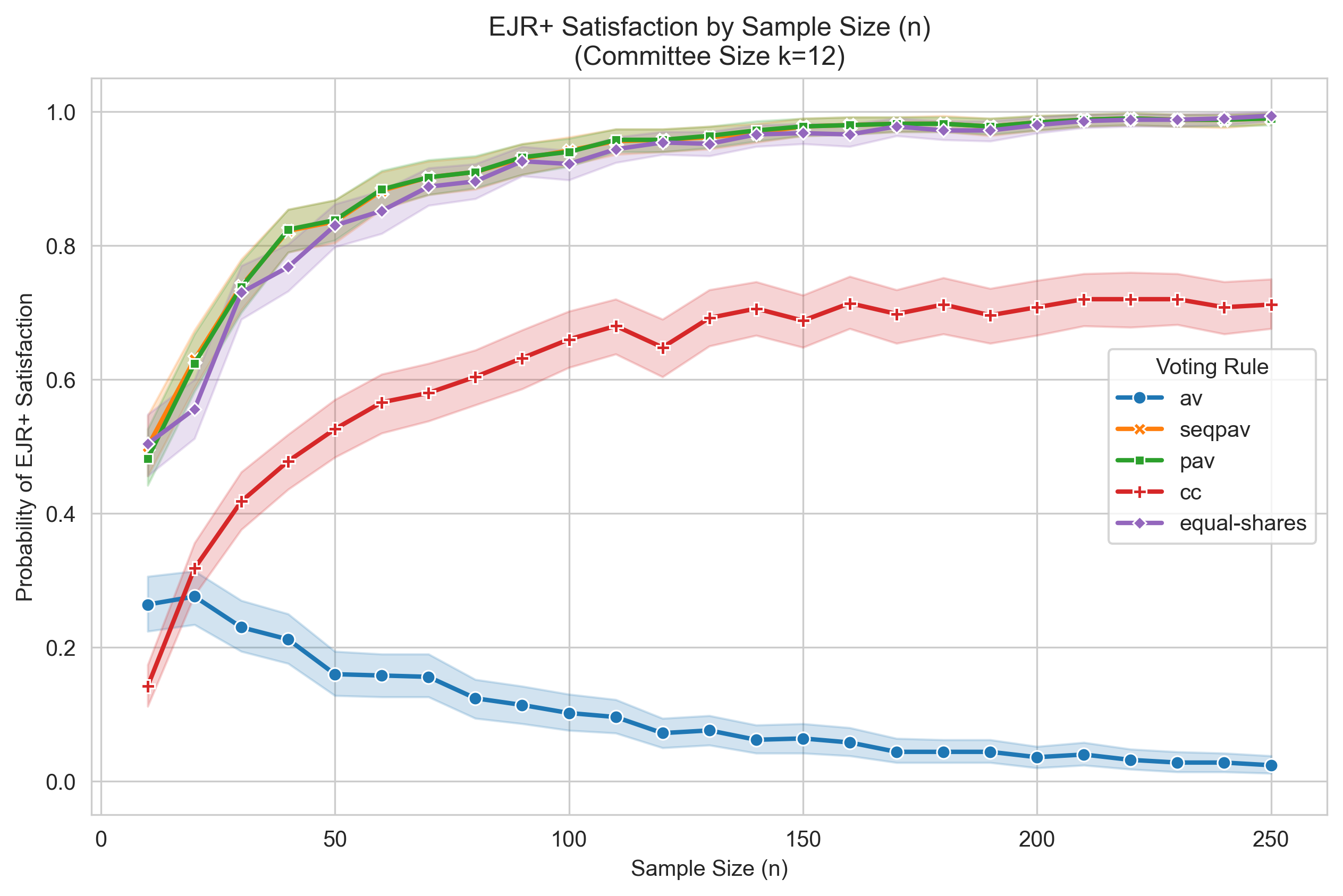}
		\caption{Probability of \ejrp{} for $k=12$}
		\label{fig:warsza_k12_ejrp}
	\end{subfigure}%
	\hfill
	\begin{subfigure}{.48\textwidth}
		\includegraphics[width=\linewidth]{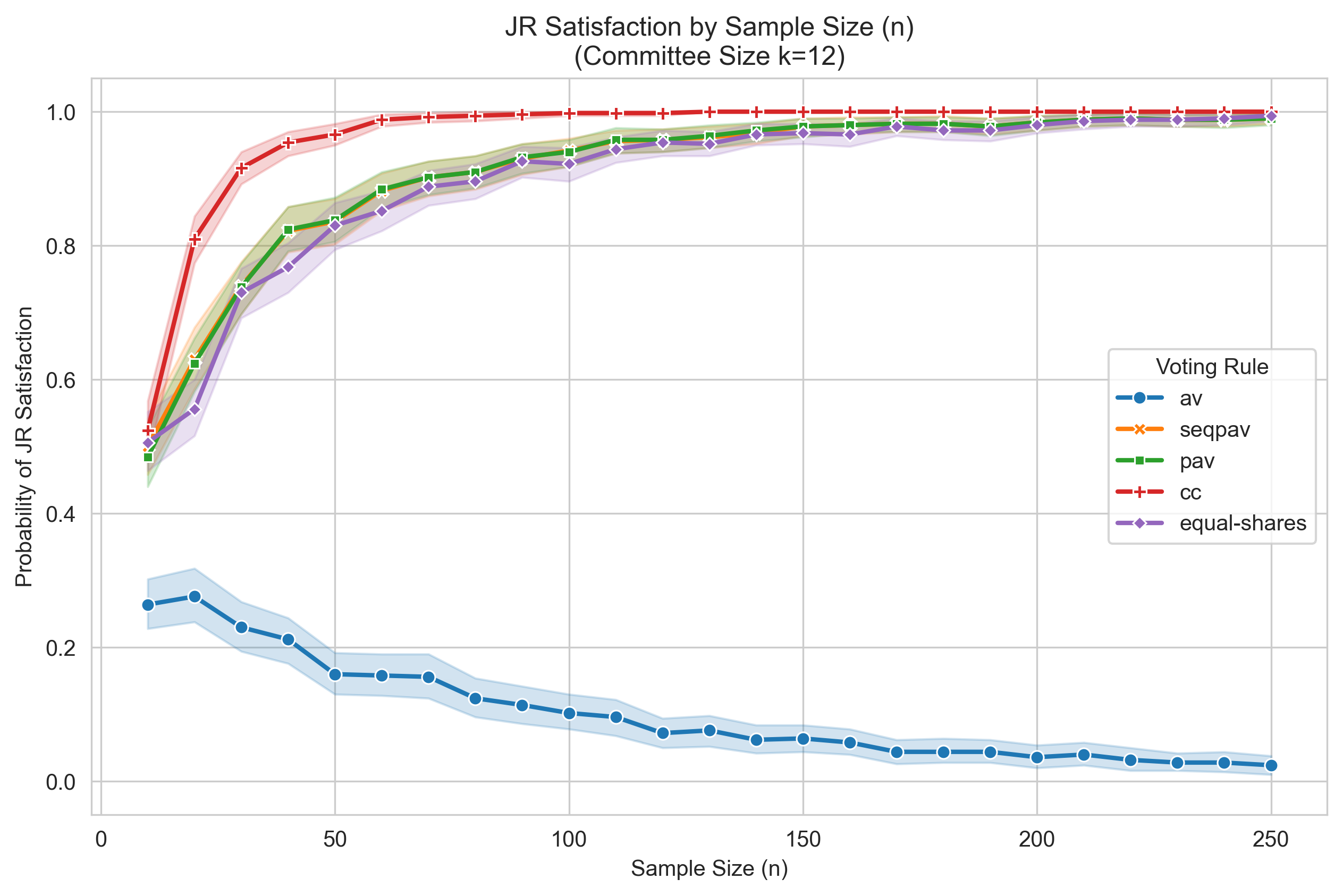}
		\caption{Probability of \jr{} for $k=12$}
		\label{fig:warsza_k12_jr}
	\end{subfigure}
	\caption{Success probability of satisfying \jr{} and \ejrp{} for various rules on approvals from the ``Poland Warszawa 2017 rejon B'' participatory budgeting election, as the number of samples increases. Empirical probabilities are out of 500 trials for various rules selecting a committee of size $k =12$, out of 53 candidates overall.}
	\label{fig:warsza_increasing_n_k12}
\end{figure}

%%%
\begin{figure}
	\begin{subfigure}{.48\textwidth}
		\centering
		\includegraphics[width=\linewidth]{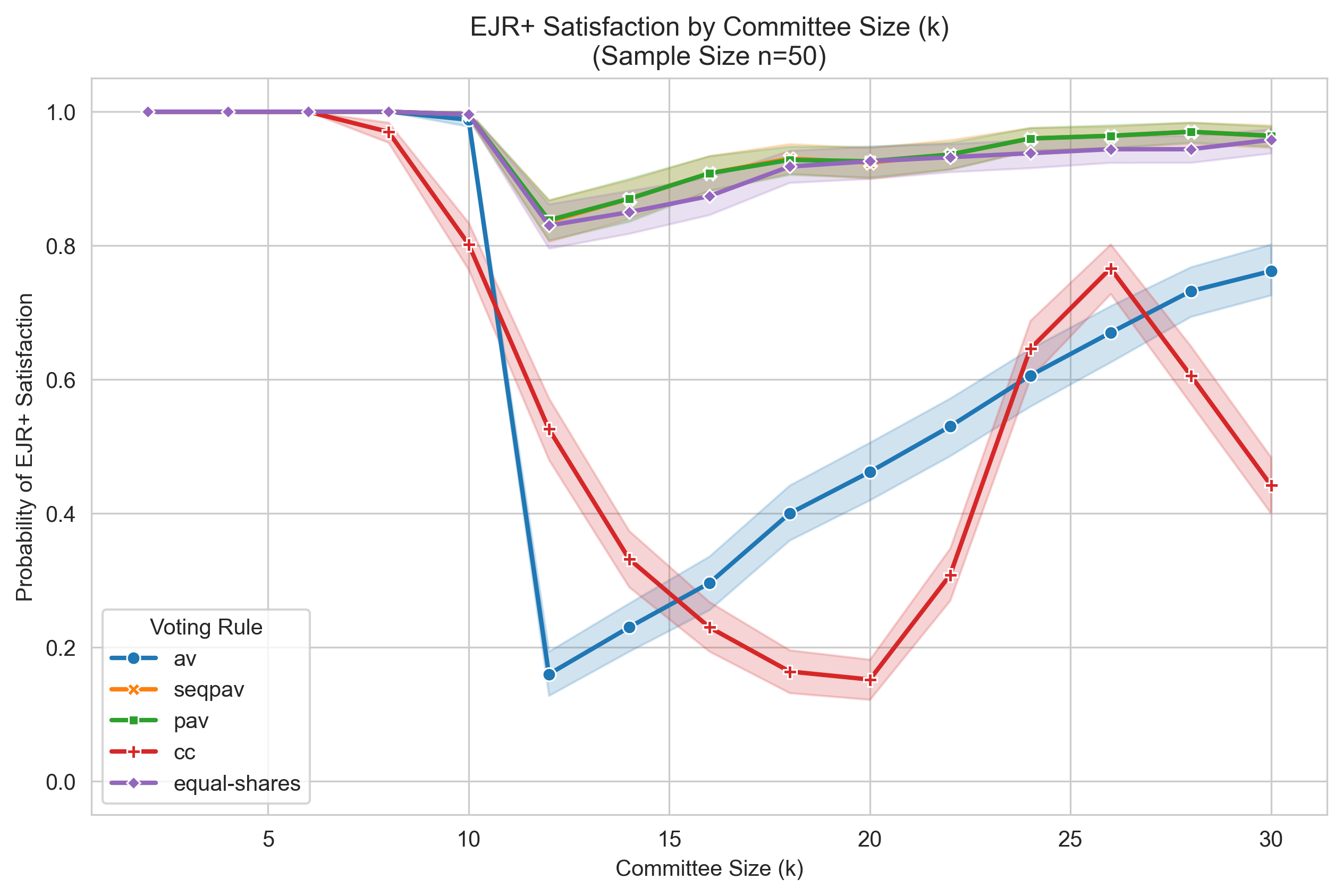}
		\caption{Probability of \ejrp{} for $n=50$}
		\label{fig:warsza_n50_ejrp}
	\end{subfigure}%
	\hfill
	\begin{subfigure}{.48\textwidth}
		\includegraphics[width=\linewidth]{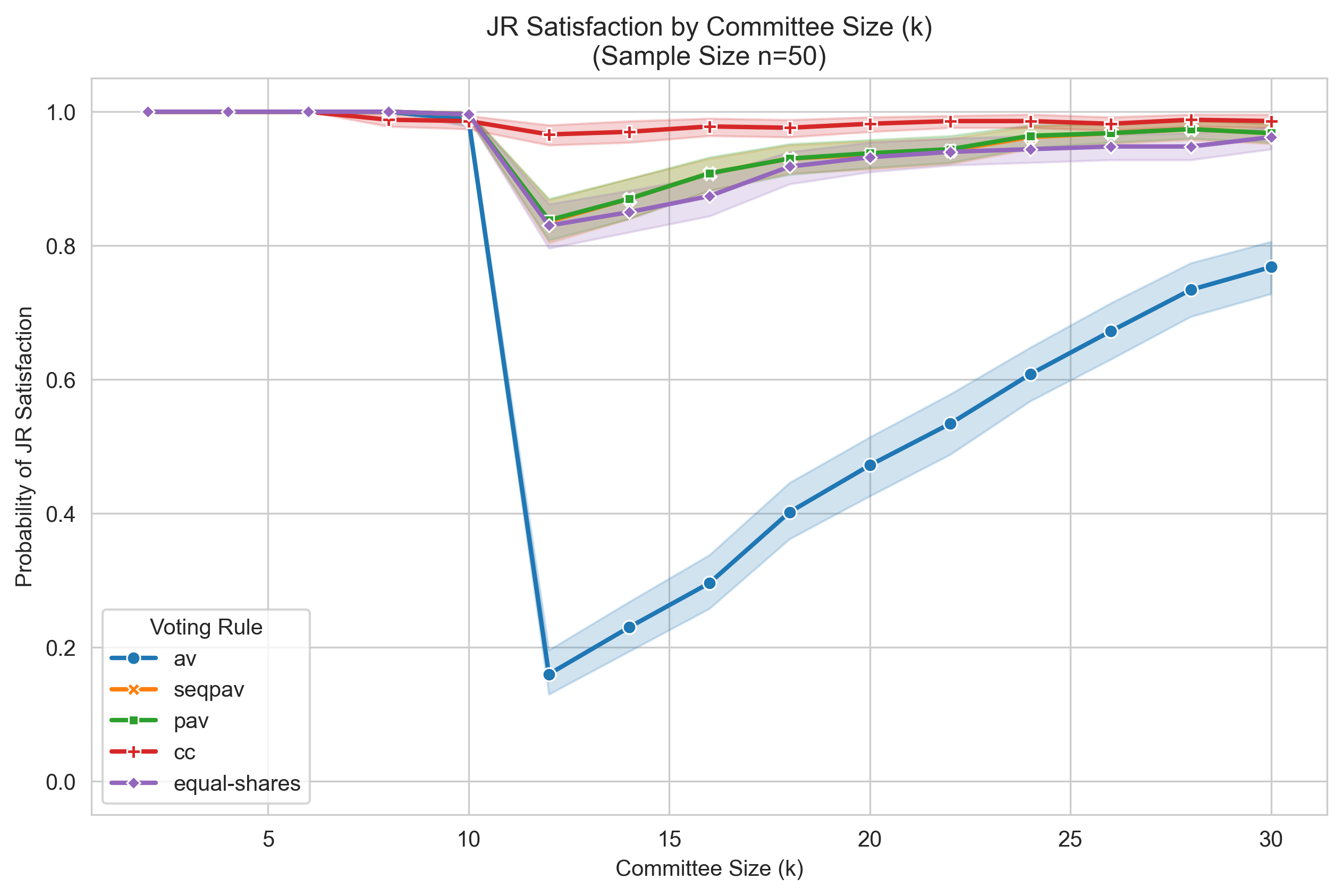}
		\caption{Probability of \jr{} for $n=50$}
		\label{fig:warsza_n50_jr}
	\end{subfigure}
	\caption{Success probability of satisfying \jr{} and \ejrp{} for various rules on approvals from the ``Poland Warszawa 2017 rejon B'' participatory budgeting election, for even $k$ up to $k=30$, out of 53 candidates overall. Empirical probabilities are out of 500 trials for various rules from $n=50$ samples.}
	\label{fig:warsza_increasing_k_n50}
\end{figure}

%%%
\subsection{Discussion of Results}

All of the rules that provably satisfy the axioms we consider appear to exhibit similar convergence rates in the number of sampled approvals needed to do so, and converge significantly more quickly than our worst-case or even beyond-worst-case guarantees would imply.
Note that since the same sets of samples were used for all rules, their absolute performances may jointly correlate with how informative the sampled ballots for a given value of $t$ are.
It is also possible that these real approval instances are not challenging enough to discern between reasonable rules.

With that said, it is interesting to note the performance of rules on axioms they are not guaranteed to satisfy.
\seqpav{} performs remarkably well (and often indistinguishably from \pav{}) for both \jr{} and \ejrp{}, while \chambcou with regards to \ejrp{} is surprisingly poor on all instances considered, even in many cases relative to \av{}.
One possible explanation is that sets of multiple candidates with high but highly correlated support are prevalent.
Finally, it bears noting that \mes{} often appears to require more samples than the score-based rules.
This could provide some indication that share-based approaches are less sample-efficient than score-based approaches to satisfying \jr{}.

The strong empirical performance of all tested proportional rules on these elections could indicate that PB and other real-world approval instances exhibit other structure which explains the tractability of identifying proportional outcomes from samples.
If so, it remains to establish whether the beyond-worst-case results in \Cref{sec:vcdim} characterize these instances, or whether other parameterizations provide better theoretical alignment to these empirical results.

\section{Conclusion}

In this work we have attempted to characterize the number of sampled ballots necessary to satisfy standard proportionality axioms in approval-based committee voting.
We have presented new lower bounds for \jr{} as well as upper bounds for significant strengthenings, and we asymptotically separate the sample complexity of these axiom satisfaction problems from the sample complexity of their natural coverage counterparts.

The most immediate outstanding question is the correct dependence on $k$, and we do not have strong reason to believe either of the present bounds is tight.
A better understanding of the correct dependence could also help establish whether the tradeoff we observe between sample and computational complexity among the rules we consider is inherent to the problem of proportional committee selection.
Structured domains also merit further study.
On the one hand, our lower bound is simple enough that it holds in many traditional restricted domains of interest.
On the other, the empirical tractability of the problem invites new characterizations of when and how our lower bound can be circumvented.

The proportionality axioms considered here enjoy a wide range of refinements, generalizations, and applications.
Within the \jr{} family, we have not established upper bounds for axioms that are incomparable with \ejrp{} such as \bjr{} and \fpjr{}.
Moving towards the most immediate applications, participatory budgeting project selection from sampled approvals and committee voting from ranked ballots can be readily studied experimentally but were unaddressed here.
Developing rules that exhibit strong sample complexity performance and guarantees for these settings could be of practical significance, and constitutes a natural direction for future work.

\bibliographystyle{plainnat}
\bibliography{refs}

\newpage
\appendix
\crefalias{section}{appendix}
\crefalias{subsection}{appendix}
\crefalias{subsubsection}{appendix}
%%%
\section{Supplementary Material and Proofs for \texorpdfstring{\Cref{sec:prelims}}{Section 2}}
\label{app:prelims}

\distinguishingsourcesnew*

This fact will be quite useful to us.
Both it and the following proof are standard; the proof is included for convenience and completeness.

\begin{proof}[Proof of \Cref{lem:separating-means-from-samples-new}]
	The upper bound is straightforward: just sample enough, take the empirical maximizer, and apply a concentration inequality together with a union bound over the $m$ indices.
	In particular, for $t$ i.i.d. samples $X = (X^1, \ldots, X^t) \sim \cD^t$, for each coordinate $j \in [m]$ compute the empirical probability estimate $\hat{p}_j \defeq \frac{1}{t} \sum_{k} X^k_j$.
	The algorithm (coordinate distinguisher) chooses $\hat{i} = \arg\max_{j} \hat{p}_j$ as its best guess for an $\eps$-approximate probability-maximizing index, breaking ties arbitrarily.

	This approach fails if $p_{\hat{i}} \leq p_{i^*}-\eps$.
	If we let $J \subset [m]$ be the indices $j$ for which $p_j \leq p_{i^*} - \eps$, then we fail if $\hat{p}_j \geq \hat{p}_{i^*}$ for any $j \in J$.
	For each $j \in J$, consider the random variable $\hat q_{i^* j} \defeq \hat p_{i^*} - \hat p_j$.
	Then failure occurs exactly when $\hat q_{i^* j}\leq 0$ for any $j \in J$.
	Since $t \cdot \hat q_{i^* j}$ is a sum of $t$ independent random variables $Q^r \defeq X^r_{i^*} - X^r_j$ for $r \in [t]$, we can apply Hoeffding:
	\[
		\probover{X \sim \cD}{ \hat q_{i^* j} \leq 0}
		\leq \prob{t\cdot  \hat q_{i^* j} \leq t\cdot \expect{ \hat q_{i^*j}} - t\cdot  \eps}
		\leq \exp \left( -\frac{2 t^2 \eps^2}{4 t} \right) = \exp\left( -\eps^2 t/2\right) ,
	\]
	where we used that $p_{i^*} \geq p_j + \eps$ for all $j \in J$ by definition.
	Union bounding over $J$ and requiring that this upper bound on the failure probability is at most $\delta$, a sufficient condition
	\[
	    \abs{J} \cdot \exp \left( -\eps^2 t/2\right) \leq
	    m \cdot \exp \left( -\eps^2 t/2\right) \leq \delta.
	\]
	Rearranging reveals this bound is satisfied for $t \geq \nf{2}{\eps^2} \log \nf{m}{\delta}$, as claimed.

	We now turn to the lower bound of $\Omega(\nf{1}{\eps^2} \log \nf{m}{\delta})$, which is slightly more involved.
	It is established in two parts, by proving first an $\Omega(\nf{1}{\eps^2} \log m)$ lower bound and then an $\Omega(\nf{1}{\eps^2} \log \nf{1}{\delta})$ lower bound, and then taking their maximum.
	To make things easier for both ourselves and the algorithm, we will additionally assume that $p_{i^*}$ is at least $\eps$ larger than all other $p_{i'}$ for $i' \neq i$, so that it is the only $\eps$-approximately-maximal index.

	Having made this additional assumption, we need to construct a family of distributions it is hard to distinguish between.
	Let this be $\mathfrak{D} =\{\cD_1, \ldots, \cD_m\}$, where each $\cD_i$ is a product distribution over Bernoullis with means $p_{i} = \nf{1}{2} + \eps$ and $p_j = \nf{1}{2}$ for $j \neq i$.
	We will imagine the process of (a) sampling a distribution $\cD_i \sim \mathfrak{D}$ uniformly at random, then (b) taking $t$ i.i.d. samples $X \sim \cD_i^t$ from it.
	Finally, let $\Psi: \{X\} \rightarrow [m]$ be an estimator that attempts to recover the distribution index $i$ from the samples.
	Our goal is to lower bound the failure probability of all such $\Psi$.
	(Note that by Yao's principle we may assume $\Psi$ is deterministic.)

	The $\log m$-dependent lower bound proceeds via Fano's inequality, which in the case of multiple hypothesis testing lower bounds the error of \emph{any} estimator of the hidden index in terms of the mutual information between the samples and the true underlying index and the number of possible indices \cite[Equation 4]{scarlett21introductory}.
	\begin{lemma}[Fano's Inequality]
		\label{lem:fano}
		For any estimator $\Psi$ of the index $i \sim [m]$ in the multiple hypothesis testing setup above,
		\[
			\probover{i \sim [m], X^t \sim \cD_i^t}{\Psi(X^t) \neq i} \geq 1 - \frac{I(i; X^t) + \log 2}{\log m},
		\]
		where $I(X; Y) \defeq H(X) - H(X \vert Y)$ is the mutual information between $X$ and $Y$ (and $H(\cdot)$ is entropy).
	\end{lemma}

	Given this uniformly random choice of $i$, the mutual information from these independent samples can in turn be bounded by the Kullback-Leibler divergence (see e.g. \cite[Lemma 28]{tibshirani17minimax}) by
	\[
		I(i; X^t) \leq t \cdot I(i; X) \leq t \cdot \max_{j \neq k} \dkl(\cD_j \| \cD_k).
	\]
	Finally, standard bounds on the KL divergence for our constructed Bernoullis give
	\begin{align*}
		\dkl(\cD_j \| \cD_k) &= \dkl(\Bern(\nf{1}{2}) \| \Bern(\nf{1}{2} + \eps)) + \dkl(\Bern(\nf{1}{2} + \eps) \| \Bern(\nf{1}{2}) ) \\
		&\leq \frac{\eps^2}{\nf{1}{4}} + \frac{\eps^2}{\nf{1}{4} - \eps^2} \leq 10 \eps^2,
	\end{align*}
	provided that $\eps \leq \nf{1}{4}$.
	Combining these inequalities with \Cref{lem:fano}, for any failure probability $\delta \leq 1/3$, we obtain that the number of samples satisfies $t = \Omega(\frac{1}{\eps^2}\log m)$, as claimed.

	The $\log \nf{1}{\delta}$-dependent lower bound proceeds from considering the family of distributions constructed above in the case of $m = 2$.
	In the two-hypothesis case, for exact hypothesis testing, we may bound the estimation error in terms of the statistical distance between our two distributions: for any estimator $\Psi$ of the uniformly sampled index $i \sim [2]$ from $t$ samples from $\cD_i$,
	\[
		\delta = \probover{i \sim \{0,1\}, X^t \sim \cD_i^t}{\Psi(X^t) \neq i} \geq \frac{1}{2}\left(1 - \dtv(\cD_0, \cD_1) \right).
	\]
	(See the proof of \Cref{thm:small-eps-lower-bound} for details.)
	From here we may use either the Tsybakov or the Bretagnolle–Huber bound to relate $\dtv$ to $\dkl$ \cite{canonne22short}.
	We will opt for Tsybakov's \cite{tsybakov09introduction}, which states
	\[
		\dtv(\cD_0, \cD_1) \leq 1 - \frac{1}{2} \exp \left(- \dkl(\cD_0, \cD_1)\right).
	\]
	Combining, we have
	\begin{align*}
		\delta &\geq \frac{1}{4} \cdot \exp \left(- \dkl(\cD_0, \cD_1)\right) \\
		&= \frac{1}{4} \cdot \exp \left(- t \cdot \left(\dkl(\Bern(\nf{1}{2}) \| \Bern(\nf{1}{2} + \eps)) + \dkl(\Bern(\nf{1}{2} + \eps) \| \Bern(\nf{1}{2}) \right) \right) \\
		&\geq \frac{1}{4} \cdot \exp \left(- t \cdot 10 \eps^2 \right),
	\end{align*}
	by the same bound as above.
	Rearranging gives $t = \Omega(\frac{1}{\eps^2} \log \nf{1}{\delta})$, again as claimed.

	Taken together, $t = \Omega\left(\nf{1}{\eps^2} \max\left(\log m, \log \frac{1}{\delta}\right)\right) = \Omega\left(\nf{1}{\eps^2} \left(\log m + \log \frac{1}{\delta}\right)\right) = \Omega\left(\nf{1}{\eps^2} \log \nf{m}{\delta}\right)$.
\end{proof}

\withvswithoutreplacement*

\begin{proof}[Proof of \Cref{claim:with-vs-without-replacement}]
	The first part, simulating $t$ samples w.r. from $t$ samples w.o.r., is simple.
	Given w.o.r. samples $X = (x_1, \ldots, x_t)$, construct the w.r. output sequence $Y = (y_1, \ldots, y_t)$ by first letting $y_1 = x_1$.
	At this point, our index is $i = 2$ we will follow the following process for $t-1$ more steps: in each step, with probability $(i-1)/n$ choose the next $y_j$ to be one of $(x_1, \ldots, x_{i-1})$ uniformly at random.
	Otherwise, with probability $1-(i-1)/n$ choose the next $y_j$ to be $x_i$ and increment $i$.
	If $X$ is sampled uniformly from $[n]$ without replacement, then $Y$ is distributed according to $t$ samples from $[n]$ with replacement.

	% \medskip
	For the second part of the claim we use rejection sampling: draw $t'$ samples with replacement, and retain only the first $t$ distinct elements.
	If there are at least $t$ distinct elements then this is clearly a uniform sample without replacement; this fails only if there are fewer than $t$ distinct elements.

	For $k \in [t]$, let the random variable $Z_k$ be the number of i.i.d. draws from $[n]$ required to find a new distinct element after having already found $k$ distinct elements.
	This follows a geometric distribution with success probability $p_k = 1 - \nf{k}{n}$.
	Since we assume $t \le (1-\eps)\cdot n$, these success probabilities are lower-bounded by $p_k \ge 1 - \frac{t}{n} \ge \eps$.

	Therefore the number of samples $T$ required to find $t$ distinct items is stochastically dominated by the sum of $t$ independent geometric random variables, each with success probability $\eps$.
	Therefore the expected number of samples is $\expect{T} \le t/\eps$.
	Using a tail bound  for the sum of the dominating geometric variables (e.g. \cite[Theorem 2.3]{janson18tail}), the probability $T$ exceeds $t' = \frac{t}{\eps} + \Delta$ decays exponentially in $\Delta$.
	In particular, $\Delta = \Theta(\frac{c}{\eps} \log n)$ guarantees failure probability at most $n^{-c}$.
\end{proof}

%%%
\section{Supplementary Material and Proofs for \texorpdfstring{\Cref{sec:coverage}}{Section 3}}
\label{app:coverage}

\Citet{aziz15justified} remark in passing that their analysis of \pav{} extends to show that \chambcou{} satisfies \jr{}.
Here we provide explicit proofs of slight generalizations of this fact.

\cvgdroopjr*

\begin{proof}[Proof of \Cref{obs:cvg-droop-jr}]
	Let $W$ maximize voter coverage $\cvg(W) \defeq \probover{A\sim \prof}{A \cap W \neq \emptyset}$.
	Suppose for contradiction that there exists a candidate $c \not \in W$ such that
	\[
		\cvg(c \vert W) \defeq \probover{A\sim \prof}{c \in A \land A \cap W = \emptyset} > \nf{1}{k+1}.
	\]
	Then by averaging, there is a $c' \in W$ such that $\cvg(c' \vert W \setminus \{c'\}) < \nf{1}{k+1}$, since
	\[
		\sum_{c' \in W} \cvg(c' \vert W \setminus \{c'\}) \leq \cvg(W) \leq 1 - \cvg(c \vert W) < \nf{k}{k+1}.
	\]
	Therefore $\cvg(W \setminus \{c'\} \cup \{c\}) > \cvg(W)$, contradicting the coverage-optimality of $W$.

	Hence no such $c'$ exists, and so $W$ satisfies Droop \jr{}.
\end{proof}

\cvgharejr*
\begin{proof}[Proof of \Cref{obs:approx-cvg-hare-jr}]
	This follows the same argument as \Cref{obs:cvg-droop-jr}.

	Suppose $W$ $\eps$-approximately maximizes voter coverage, meaning $\cvg(W) \geq \cvg(W^*) - \eps$.
	Suppose also that there is a (Hare) \jr{} violation for $W$, meaning there exists a candidate $c \not \in W$ with
	\[
		\cvg(c \vert W) = \probover{A\sim \prof}{c \in A \land A \cap W = \emptyset}
		\geq \nf{1}{k}.
	\]
	Again by averaging, there is a $c' \in W$ such that $\cvg(c' \vert W \setminus \{c'\}) \leq \nf{1}{k} - \nf{1}{k^2}$, since
	\[
		\sum_{c' \in W} \cvg(c' \vert W \setminus \{c'\}) \leq \cvg(W) \leq 1 - \cvg(c \vert W)
		\leq \frac{k-1}{k}.
	\]
	Combining, we have
	\[
		\cvg(W^*) \geq \cvg(W \setminus \{c'\} \cup \{c\}) \geq \cvg(W) - \cvg(c' \vert W\setminus \{c'\}) + \cvg(c \vert W) \geq \cvg(W) + \nf{1}{k^2}.
	\]
	By assumption on $W$, this means $\eps \geq \nf{1}{k^2}$.
	Therefore if $W$ admits a \jr{} violation then $\eps \geq \nf{1}{k^2}$.
\end{proof}

We now prove the main pair of lemmas in the section.

\epsmkcsampleub*

\begin{proof}[Proof of \Cref{lem:eps-mkc-sample-upper-bound}]
	This argument directly follows the upper bound of \Cref{lem:separating-means-from-samples-new}.
	We treat each $k$-subset $C$ as a Bernoulli with mean $p_C \defeq \frac{1}{\abs{\cU}} \cvg(C)$.
	There is some k-subset $C^*$ with maximum coverage $p^* \defeq \cvg(C^*)$, and \Cref{alg:eps-mkc-brute-force} succeeds if and only if its output $C_A$ satisfies $p_{C_A} \geq p^* - \eps$.
	our coverage estimates $\hat{c}_C$ from our $t$ samples might be correlated, but the union bound doesn't mind.

	Our failure event is contained by the event that either (a) $C^*$ has unusually bad sample coverage $\hat c_{C^*} \leq p^* t - \eps t/2$ or (b) any of the `bad' $C$ with true coverage $p_C \leq p^* - \eps $ exhibit unusually good sample coverage $\hat c_{C} \geq p^* t - \eps t/2$.
	If neither (a) nor (b) transpires, then the algorithm will pick a $C_A$ with sample coverage at least $p^* t - \eps t/2$, and all such sets will not be `bad'.

	By Hoeffding, the probability of (a) over the samples $\bu = (u^1, \ldots, u^t)$ is
	\begin{equation}
		\probover{\bu \sim \cU^t}{c_{C^*} \leq p^* t - \eps t/2} = \probover{\bu \sim \cU^t}{c_{C^*} \leq \expect{c_{C^*}} - \eps t/2} \leq \exp\left(- \eps^2 t/2\right).		 \notag
	\end{equation}
	Identically, for any bad $C$, the probability of it making (b) true is at most
	\begin{equation}
		\probover{\bu \sim \cU^t}{c_{C} \geq p^* t - \eps t/2} = \probover{\bu \sim \cU^t}{c_{C} \geq \expect{c_{C}} + \eps t/2} \leq \exp\left(- \eps^2 t/2\right).		 \notag
	\end{equation}

	Since all $k$-subsets besides $C^*$ could be bad, union bounding over all $C \in \binom{\cS}{k}$ gives an upper bound on our bad event probability of $\binom{\abs{\cS}}{k} \cdot \exp\left(- \eps^2 t/2\right) \leq \exp\left(k \log \abs{\cS} - \eps^2 t/2\right)$.
	Asserting
	\[
	\exp\left(k \log \abs{\cS} - \eps^2 t/2\right) \leq \delta
	\]
	and solving for the minimum feasible $t$ yields $t = O\left(\nf{1}{\eps^2} \left( k \log \abs{\cS} + \log \nf{1}{\delta} \right)\right)$, as claimed.
\end{proof}

\epsmkcsamplelb*
\begin{proof}[Proof of \Cref{lem:eps-mkc-sample-lower-bound}]
    For this we will choose $\delta = \nf{1}{4}$.
    For given $k$, $m$, and $\eps$, we will construct a hard instance as the concatenation (disjoint union) of $k$ non-overlapping subinstances of equal size and identical structure.
    Without loss of generality let $\eps = 2^{-r}$ for some $r \in \mathbb N$.
    We will also assume that $k$ divides $m$, let $m' \defeq \abs{\cS}/k$, and use $n = kn'$ elements, where $n' \defeq \cdot 2^{r+m'-1}$ is the number of elements in each disjoint subinstance.

	Denote each subinstance by $(\cU_i, \cS_i)$ for $i \in [k]$, and let $\cS\defeq \cup_i \cS_i$ and $\cU \defeq \cup_i \cU_i$.
	Each subinstance will have $\abs{\cU_i} = n'$ elements and a unique good set $S_i^*$ of size $n' \cdot (\nf{1}{2} + 4 \eps)$, and $m' - 1$ bad sets $\{S_{ij}\}_j$ of size $n' /2$.
	The sets in $\cS_i$ are incident on $\cU_i$ in such a way as to exhibit independent/product structure with respect to the elements.
	In particular, for any subfamily of bad set indices $T \subseteq [m'-1]$ let there be $2^{r-1} + 4$ elements contained in both $S_i^*$ and exactly the sets $\{S_{ij}\}_{j \in T}$, and $2^{r-1} - 4$ elements not in $S_i^*$ and in exactly the sets $\{S_{ij}\}_{j \in T}$.
	This is $n' = 2^{r+m'-1}$ elements overall.

	Since each unique good set in each disjoint subinstance is additively $4\eps \cdot n'$ larger than the other sets in the subinstance, the instance solutions $C \in \binom{\cS}{k}$ which are additively within $\eps \cdot n$ of optimal overall coverage (and therefore satisfy $\eps$-\MkC) are exactly the $C$ which contain at least $\nf34$ of the $k$ unique good subinstance sets.

	To succeed with constant probability, we therefore need to succeed in good set identification on at least $3k/4$ subinstances.
	By \Cref{lem:separating-means-from-samples-new}, let $t' = \Theta(\nf{1}{\eps^2} \log m')$ be the number of samples \emph{from subinstance $i$} necessary to determine $S_i^*$ with probability at least $\nf12$.

	If the overall number of samples is $t = t' \cdot \nf k 3$, then at most $k/3$ subinstances $i$ see $t_i \geq t'$ samples, and at least $2k/3$ subinstances $i'$ see $t_{i'} < t'$ samples.
	Fix a sample-based algorithm $A$, and optimistically suppose it succeeds at finding the big set on all $k/3$ subinstances $i$ in this first group.
	By the definition of $t'$, \emph{any} algorithm $A$ succeeds at identifying the good set in each $i'$ in this second group independently with probability at most $\nf12$.
	Therefore at least half the time, $A$ succeeds on at most $k/3$ of the subinstances $i'$ that see few samples.
	Since $2k/3 < 3k/4$, in this event we conclude that $A$ fails to find a $k$-subset within $\eps n$ of maximium true coverage.

	We conclude that the failure probability of any algorithm $A$ on this instance is at least $\delta = \nf12$ when given $t = \Omega(kt') = \Omega\left(\nf{k}{\eps^2} \log m' \right)$ samples.
	Because $\log m' = \log \abs{\cS} - \log k = \Omega(\log \abs{\cS})$ since $\abs{\cS} = k^{1 + \Omega(1)}$ by assumption, the claim follows.
\end{proof}

%%%
\section{Supplementary Material and Proofs for \texorpdfstring{\Cref{sec:upperbounds}}{Section 4}}
\label{app:upper-bounds}

We begin with a general-purpose fact, used in the proofs of correctness for
\samplecspav (\Cref{thm:cs-pav-correctness}),
\samplepav (\Cref{prop:sample-pav-correct}),
\samplelspav (\Cref{prop:sample-ls-pav-correct}), and
\samplewslspav (\Cref{thm:ws-ls-pav-correctness}).

\pavtojerp*

\begin{proof}[Proof of \Cref{prop:pav-score-to-ejrp-guarantee}]
	We follow the proof of \citet{aziz18complexity}, and prove the contrapositive.
	Suppose that a given committee $W$ does not satisfy \ejrpe{}.
	We will show that there exists a $c' \in W$ and $c \in C \setminus W$ such that $\Delta(W, c,c') \geq \eps$.

	Since $W$ does not satisfy \ejrpe{} (\Cref{def:ejrp-eps}), there is some $\ell \in [k]$ and some candidate $c^* \in C \setminus W$ such that $\probover{A \sim \cP}{c^* \in A \land \abs{A \cap W} \leq \ell - 1} \geq \ell \cdot \left( \frac{1}{k+1} + \eps \right)$.
	Let $S\subseteq 2^C$ denote the support of this event, i.e. $S \defeq \{A \in 2^C: \: c^* \in A \land \abs{A \cap W} \leq \ell - 1\}$.
	How much does $\pavsc$ increase if we add $c^*$?
	Let $W^* \defeq W \cup \{c^*\}$.
	Then we can lower bound the score of this too-large committee by
	\begin{align}
		\pavsc(W^*) &= \probover{A \sim \cP}{\sum_{j=1}^{\abs{A \cap W^*}} \frac{1}{j}}  \notag \\
		&= \probover{A \sim \cP}{\sum_{j=1}^{\abs{A \cap W^*}} \frac{1}{j} \cdot \1\{A \not\in S\}} +  \probover{A \sim \cP}{\sum_{j=1}^{\abs{A \cap W^*}} \frac{1}{j} \cdot \1\{A \in S\}}  \notag \\
		&\geq \probover{A \sim \cP}{\sum_{j=1}^{\abs{A \cap W}} \frac{1}{j} \cdot \1\{A \not\in S\}} +  \probover{A \sim \cP}{\left( \sum_{j=1}^{\abs{A \cap W}} \frac{1}{j} + \frac{1}{\ell}\right) \cdot \1\{A \in S\}}  \notag \\
		&= \probover{A \sim \cP}{\sum_{j=1}^{\abs{A \cap W}} \frac{1}{j} } + \frac{1}{\ell} \cdot \cP(S)  \notag\\
		&\geq \pavsc(W) + \left(\frac{1}{k+1} + \eps \right). \label{eq:pav-pivot-eq1}
	\end{align}
	Next consider the expected score impact of ejecting a candidate from $W^*$ uniformly at random.
	This is given by
	\begin{align}
		\frac{1}{k+1} \sum_{c' \in W^*}\left(\pavsc(W^*) - \pavsc(W^* \setminus \{c'\}) \right) &= \frac{1}{k+1} \cdot \sum_{c' \in W^*} \probover{A\sim \cP}{\abs{A \cap W^*}^{-1} \cdot \1\{c' \in A\}} \notag\\
		&= \frac{1}{k+1} \cdot \probover{A \sim \cP}{\sum_{c' \in A \cap W^*}  \abs{A \cap W^*}^{-1} }  \notag\\
		&= \frac{1}{k+1} \cdot \probover{A \sim \cP}{ A \cap W^* \neq \emptyset }  \notag\\
		&\leq \frac{1}{k+1}. \notag
		\intertext{Therefore there is some $c' \in W^*$ for which this holds, meaning that}
		\pavsc(W^*) - \pavsc(W^* \setminus \{c'\}) &\leq \frac{1}{k+1}. \label{eq:pav-pivot-eq2}
	\end{align}
	(Note also that this $c' \neq c^*$.)
	Combining \eqref{eq:pav-pivot-eq1} and \eqref{eq:pav-pivot-eq2} yields
	\[
	\Delta(W, c^*, c') = \pavsc(W^* \setminus \{c'\}) - \pavsc(W) \geq \eps,
	\]
	meaning that a sufficiently good pivot exists.
	This concludes the proof.
\end{proof}

\relaxedejrpub*

\begin{proof}[Proof of \Cref{cor:sample-complexity-ejrp-relaxed-ub}]
	Since $\alpha = (\frac{k}{k+1} + k\eps)^{-1}$, we may write constant $\alpha < 1$ as $ \alpha = \frac{1}{1+\beta}$ for some constant $\beta > 0$.
	Then
	\[
		 \eps = \frac{1}{k}\left( \frac{1}{\alpha} - \frac{k}{k+1}\right) = \frac{1}{k}\left(\beta + \frac{1}{k+1}\right) = \Omega(k^{-1}).
	\]
	Substituting this into the sample complexity guarantee of \Cref{thm:cs-pav-correctness} yields the claim.
\end{proof}

%%%
\subsection{Sample-based adaptations of existing rules}
\label{app:upper-bounds-known}

The contents of this subsection are summarized by the following claim:

\standardrulecomplexity*

We divide the proof of \Cref{obs:complexity-of-standard-rules} into \Cref{prop:sample-greedy-correct}, \Cref{prop:sample-pav-correct}, and \Cref{prop:sample-gjcr-correct}, below.

\begin{algorithm}[H]
	\caption{\samplepav{}}
	\label{alg:sample-pav}
	\KwIn{Candidates $C$, $k$, $\eps$, failure probability $\delta$}
	\KwOut{committee $W$ satisfying \ejrpe{} w.p. $1-\delta$}
	Sample $\cA = (A^1, \ldots, A^t) \sim \cP^t$ i.i.d. for $t = \nf{2}{\eps^2} (1 + \log k)^2(k \log m + \log \nf{1}{\delta})$ \;
	\For{$W \in \binom{C}{k}$}{
			$\hat s_W \leftarrow \pavsc(W; \cA)$ \;
		}
	\Return $W \in \arg\max_{W \in \binom{C}{k}} \hat s_W$
\end{algorithm}

\begin{proposition}
\label{prop:sample-pav-correct}
	\samplepav (\Cref{alg:sample-pav}) satisfies \ejrp{} with probability $1-\delta$ and has sample complexity $\tilde O(\nf{k}{\eps^2})$.
	(For \ejrp{} = \ejrpe{}, it uses $\tilde O(k^5)$ samples.)
\end{proposition}

\begin{proof}[Proof of \Cref{prop:sample-pav-correct}]
	The number of samples is specified by the definition of the algorithm.
	It remains to show success with high probability.

	The first step of this argument closely resembles \Cref{lem:separating-means-from-samples-new}.
	Let $\cE$ be the failure event that the set $W$ we choose has score $\pavsc(W) \leq \max_{W'} \pavsc(W') - \eps$.
	As in the proof of \Cref{lem:separating-means-from-samples-new}, observe that if $\cE$ holds then either (a) the score-maximizing committee $W^*$ has sample score at most $\pavsc(W^*) - \eps/2$ or (b) any of the bad committees $W \in \cW_B$ with $\pavsc(W) \leq \pavsc(W^*) - \eps$ has sample score at least $\pavsc(W^*) - \eps/2$.
	We apply Hoeffding and union bounds to upper bound the probability of these events, and therefore of $\cE$.

	For ease of notation, let $s^* \defeq \pavsc(W^*)$.
	By Hoeffding, the probability of (a) over the samples $\cA = (A^1, \ldots, A^t)$ is
	\begin{equation}
		\probover{\cA \sim \cP^t}{\hat s_{W^*} \leq s^* - \eps /2}
		= \prob{t \hat s_{W^*} \leq t \expect{\hat s_{W^*}} - \eps t /2}
		\leq \exp\left(- \frac{\eps^2 t^2/2 }{t H_k^2 } \right) = \exp\left(- \frac{\eps^2 t }{2 H_k^2 } \right),		 \notag
	\end{equation}
	Where $H_k$ is the $k$th harmonic number.
	Identically, for any bad $W \in \cW_B$, the probability of it making (b) true is at most
	\begin{equation}
		\probover{\cA\sim \cP^t}{\hat s_{W} \geq s^*  - \eps /2} = \prob{t \hat s_{W} \geq t \expect{\hat s_W} + \eps t/2} \leq \exp\left(- \frac{\eps^2 t}{2 H_k^2} \right).		 \notag
	\end{equation}

	Since all $k$-committees besides $W^*$ could be bad, union bounding over all $W$ gives an upper bound on our bad event probability of $\binom{\abs{C}}{k} \cdot \exp\left(- \eps^2 t/2 H_k^2\right) \leq \exp\left(k \log m - \eps^2 t/2H_k^2\right)$.
	Asserting
	\[
	\exp\left(k \log m - \eps^2 t/2 H_k^2\right) \leq \delta
	\]
	and solving for the minimum feasible $t$ yields $t \geq 2 (\log k + 1)^2\nf{1}{\eps^2} \left( k \log m + \log \nf{1}{\delta} \right)$.
	Since our number of samples satisfies this bound, our failure probability is therefore upper bounded by $\delta$.

	The second step is to argue that if the $\pavsc(W)$ for the output $W$ is sufficiently high (i.e. provided $\cE$ does not hold), then $W$ satisfies \ejrp{}.
	For this we can leverage \Cref{prop:pav-score-to-ejrp-guarantee}.
	In particular, if $\pavsc(W) > \max_{W'} \pavsc(W') - \eps$ then there is certainly no pivot $(c,c')$ such that $\Delta(W, c, c') \geq \eps$.
	Therefore by \Cref{prop:pav-score-to-ejrp-guarantee} $W$ satisfies \ejrpe{}.
\end{proof}

\begin{algorithm}[H]
	\caption{\samplelspav{}}
	\label{alg:sample-ls-pav}
	\KwIn{Candidates $C$, $\eps$,  initial committee $W_{init}$, failure probability $\delta$}
	\KwOut{committee $W$ satisfying \ejrpe{} w.p. $1-\delta$}
	$W \gets W_{init}$ \;
	\For{$\nf{3}{\eps} \cdot (\log k + 1)$ steps }{
		Sample $\cA = (A^1, \ldots, A^t) \sim \cP^t$ i.i.d. for $t = \nf{18}{\eps^2} (\log k + 1)^2(\log \nf{m}{\delta} + \log \nf{1}{\eps}  + 2(\log k + 1))$ \;
		$Q \gets \left\{ (c,c') \in (C \setminus W) \times W: \: \Delta(W, c, c'; \cA) \geq \frac{2\eps}{3} \right\}$ \;
		 \eIf{$Q \neq \emptyset$}{
			$W \gets W \cup \{c\} \setminus \{c'\}$ for arbitrary $(c, c') \in Q$ \;
		}{
			\textbf{break} \;
		}
	}
	\Return $W$ \;
\end{algorithm}

\begin{proposition}
	\label{prop:sample-ls-pav-correct}
	\samplepav (\Cref{alg:sample-ls-pav}) satisfies \ejrp{} with probability $1-\delta$ and has sample complexity $\tilde O(\nf{1}{\eps^3})$.
	(For \ejrp{} = \ejrpe{}, it uses $\tilde O(k^6)$ samples.)
\end{proposition}

\begin{proof}[Proof of \Cref{prop:sample-ls-pav-correct}]
	The number of samples is specified by the definition of the algorithm.
	It remains to show success with high probability.

	This is similar to the proof of \Cref{prop:sample-pav-correct}, except that we will union bound over more events.
	In essence, we would like to guarantee that (a) whenever there is a pivot $(c,c')$ such that $\Delta(W, c, c') \geq \eps$, we choose a pivot, and that (b) any pivots $(c,c') \in Q$ increase $\pavsc(W)$ by at least $\eps/3$.
	As long as both of these things remain true, we will perform at most $\frac{\pavsc(W^*)- \pavsc(W_{init})}{\eps/3} \leq 3 H_k/\eps < 3 (\log k + 1)/\eps$ pivots before no remaining pivots confer $\pavsc$ benefit $\eps$ or more.
	At this point we will terminate, our final committee will satisfy \ejrpe{} by \Cref{prop:pav-score-to-ejrp-guarantee}.

	It remains only to bound our failure probabilities, and argue that the number of samples used is sufficient.
	For each round, the number of pivots under consideration is $\abs{(C \setminus W) \times C} \leq km$.
	As before, within a given round our conditions (a) and (b) are met so long as (i) no $\eps$-good pivots appear worse than $2\eps/3$, and (ii) no $\eps/3$-bad pivots appear better than $2\eps/3$.
	For fixed $W$ and a pivot $q = (c, c')$, let $\hat\Delta_q \defeq \Delta(W, c, c', \cA)$ denote the empirical estimate of its $\pavsc$ impact.
	In the first case, for good pivot $q$ by Hoeffding we have
	\begin{equation}
		\probover{\cA\sim \cP^t}{\hat s_{q} \leq \eps - \eps/3 }
		\leq  \prob{t \hat s_{q} \leq t \expect{\hat s_{q}} - \eps t/3}
		\leq \exp\left(- \frac{2 \eps^2 t^2/9}{t (2H_k)^2} \right)
		= \exp\left(- \frac{ \eps^2 t}{18 H_k^2} \right).		 \notag
	\end{equation}
	The probability of a $\eps/3$-bad pivot $q$ appearing at least $2\eps/3$-good can be upper bounded identically.
	Union bounding over the $km$ pivots, we have
	\[
		km \cdot \exp\left(- \frac{\eps^2 t}{18 H_k^2} \right) \leq \delta',
	\]
	for per-round failure probability $\delta'$.
	To make sure our failure is at most $\delta$ over all rounds, we set $\delta' = \delta / (\nf{2}{\eps} H_k)$ and solve to obtain
	\[
		t \geq \frac{18 H_k^2}{\eps^2} \left(\log k + \log m + \log \nf{1}{\delta} + \log \nf{1}{\eps} + H_k + \log 2 \right).
	\]
	As our per-round number of samples satisfies this inequality, our success is bounded accordingly.
\end{proof}

\begin{algorithm}[H]
	\caption{\samplegjcr{}}
	\label{alg:sample-gjcr}
	\KwIn{Candidates $C$, $k$, $\eps$, failure probability $\delta$}
	\KwOut{committee $W$}
	$W \gets \emptyset$ \;
	\For{$\ell \in \{k, k-1, \ldots, 2, 1\}$}{
		\While{\texttt{TRUE}}{
			Sample $\cA = (A^1, \ldots, A^t) \sim \cP^t$ i.i.d. for $t =\nf{2}{\eps^2} (\log \abs{C} + \log \nf{k}{\delta})$ \;
			$S \gets \left\{c \in C \setminus W : \abs{\{A \in \cA : c \in A \:\wedge\: \abs{A \cap W} < \ell\}} \ge t \cdot \ell \cdot \left(\frac{1}{k+1} + \nf{\eps}{2}\right) \right\}$ \;
			\eIf{$S \neq \emptyset$}{
				$W \gets W \cup \{c\}$ for arbitrary choice of $c \in S$\;
			}{
				\textbf{break} \;
			}
		}
	}
	\Return $W$ \;
\end{algorithm}

\begin{proposition}
	\label{prop:sample-gjcr-correct}
	\samplegjcr (\Cref{alg:sample-gjcr}) satisfies \ejrp{} with probability $1-\delta$ and has sample complexity $\tilde O(\nf{k}{\eps^2})$.
	(For \ejrp{} = \ejrpe{}, it uses $\tilde O(\nf{k}{\eps^2})$ samples.)
\end{proposition}

\begin{proof}[Proof of \Cref{prop:sample-gjcr-correct}]
	This proof proceeds in much the same way as the others.
	On line 5, let $\hat Q^\ell(c,W) = \{A \in \cA : c \in A \:\wedge\: \abs{A \cap W} < \ell\}$ and let $\hat q^\ell(c,W) \defeq \frac{1}{t} \abs{\hat Q^\ell(c,W)}$ note that $\hat q^\ell(c,W)$ is an estimator for
	\[
		q^\ell(c,W) \defeq \probover{A\sim\prof}{c \in A \:\wedge\: \abs{A \cap W} < \ell}.
	\]

	We define our two failure events in line 5 to be that, for fixed already-chosen candidates $W$, either (a) we include a $c \in \hat Q^\ell(c,W)$ with insufficient $\ell$-demand given $W$, meaning that $q^\ell(c,W) \leq \ell \frac{1}{k+1}$, or (b), we fail to include a $c$ with high $\ell$-demand given $W$, meaning that $q^\ell(c,W) \geq \ell \left( \frac{1}{k+1} + \eps\right)$.

	Applying the same Hoeffding bounds as elsewhere suffices to union bound over the at most $m$ possible bad events per round, and over the at most $k$ times that line 5 is invoked.

	The fact that line 5 is invoked at most 5 times and that, conditioned on no failure events having transpired, the output committee satisfies \ejrpe{}, both follow from the proof of correctness of the algorithm.
	The latter is direct from the definition of the algorithm and the definition of \ejrpe{}, again provided no failure events occur.
	That the $W$ output by \Cref{alg:sample-gjcr} satisfies $\abs{W} \leq k$ follows the charging argument made in \cite{brill23robust} for the Hare quota or in \cite{casey25justified} for the Droop quota.
\end{proof}

\begin{algorithm}[H]
	\caption{\samplegreedy{}}
	\label{alg:sample-greedy}
	\KwIn{Candidates $C$, $k$, $\eps$, failure probability $\delta$}
	\KwOut{committee $W$ satisfying $\jr{}$ w.p. $1-\delta$}
	$W \gets \emptyset$ \;
	\For{$i \in [k]$}{
		Sample $\cA_i = (A^1, \ldots, A^t) \sim \cP^t$ i.i.d. for $t = \nf{2}{\eps^2} (\log m + \log \nf{1}{\delta})$ \;
		$c_i \gets \arg \max_{c \in C \setminus W} \abs{ \{A \in \cA_i: A \cap W = \emptyset  \: \wedge\: c \in A\}}$ \;
		$W \gets W \cup \{c_i\}$ \;
	}
	\Return $W$ \;
\end{algorithm}

\begin{proposition}
	\label{prop:sample-greedy-correct}
	\samplegreedy (\Cref{alg:sample-greedy}) satisfies \ejrpe{} with probability $1-\delta$ and has sample complexity $\tilde O(\nf{k}{\eps^2})$.
	(For \ejrp{}, it uses $\tilde O(k^5)$ samples.)
\end{proposition}

\begin{proof}[Proof of \Cref{prop:sample-greedy-correct}]
	This is simply an implementation of the greedy algorithm for \jr{} using empirical marginal coverage estimates.

	For each of $k$ steps, apply \Cref{lem:separating-means-from-samples-new} with $\delta' = \delta/k$, as in the proof of \Cref{prop:sample-ls-pav-correct}.
	By using fresh samples for each step and union bounding over the failure probability for each of the $k$ steps, we use $ O\left(k/\eps^2 \cdot \log \frac{mk}{\delta} \right)$ samples. Since $k \leq m$, the claim follows.
\end{proof}

%%%
\subsection{A direct combination of \pav and \lspav}
\label{app:upper-bounds-pav-pivot}

Here we show that a direct combination of \samplepav and \samplelspav suffices to improve upon the sample complexity of standard rules, as presented in \Cref{app:upper-bounds-known}.

\begin{algorithm}[h]
	\caption{\samplewslspav{}}
	\label{alg:sample-ws-ls-pav}
	\KwIn{Candidates $C$,  $k$, tolerance $\eps$, failure probability $\delta$}
	\KwOut{committee $W$ satisfying $\ejrp{}_\eps$ w.p. $1-\delta$}

	Sample $\cA = (A^1, \ldots, A^t) \sim \cP^t$ i.i.d. for $t = 2 (\log k + 1)^2 \frac{\eps^{-2}}{k^{2/3}} (k \log \abs{C} + \log \nf{1}{\delta} + 1)$ \label{line:sample-ws-ls-pav-first-samples} \;
	\For{$W \in \binom{C}{k}$}{
		$\hat s_W \leftarrow \pavsc(W; \cA)$ \;
	}
	$W \gets \arg\max_{W \in \binom{C}{k}} \hat s_W$ \label{line:sample-ws-ls-pav-step-output-committee} \;
	\For{$3\cdot k^{1/3}$ steps}{
		Sample $\cA = (A^1, \ldots, A^t) \sim \cP^t$ i.i.d. for $t =18 (\log k + 1)^2 \eps^{-2} ( \log \nf{\abs{C}}{\delta} + 2 \log k + 2)$ \label{line:sample-ws-ls-pav-second-samples} \;
		$Q \gets \left\{ (c,c') \in (C \setminus W) \times W: \: \Delta(W, c, c'; \cA) \geq t \cdot \frac{2\eps}{3} \right\}$ \;
		\eIf{$Q \neq \emptyset$}{
			$W \gets W \cup \{c\} \setminus \{c'\}$ for arbitrary $(c, c') \in Q$ \;
		}{
		\textbf{break} \;
		}
	}
	\Return $W$ \;
\end{algorithm}

\begin{theorem}
	\label{thm:ws-ls-pav-correctness}
	\samplewslspav{} (\Cref{alg:sample-ws-ls-pav}) satisfies \ejrpe{} with probability at least $1-\delta$.
	Furthermore, it does so using $O\left( \eps^{-2} k^{1/3} \left(\log^2 k + \log \nf{m}{\delta} \right)\right)$ samples.
\end{theorem}

Recalling $\ejrp{}$ corresponds to $\ejrp{}_\eps$ for $\eps = \nf{1}{k^2}$, this implies \jr{} has sample complexity $\tilde O(k^{13/3})$.

Why this particular division between the accuracy of \samplepav and subsequent \samplelspav pivots?
Intuitively, if $\eps = k^{-\alpha}$ for some fixed $\alpha$, suppose we deploy \samplepav to accuracy $\eps' = k^{-\beta}$ before switching to pivots.
Choosing $\beta = \alpha$ reduces to just running \samplepav, while $\beta = 0$ corresponds to relying on \samplelspav.

What is the sample complexity as a function of $\beta$?
If the committee $W$ on line 5
is within $k^{-\beta}$ of optimal \pav{} score, then we must do at most $\eps'/\eps = k^{\alpha - \beta}$ pivots of gain $\eps$ before arriving at a approximately \pav{}-maximal committee.
Therefore our overall sample complexity is
\[
\max\left(\tilde O \left(\frac{k}{\eps'^2}\right), \tilde O \left(\frac{\eps'}{\eps} \cdot \frac{1}{\eps^2}\right)  \right)
= \max\left(\tilde O \left(k^{1 + 2 \beta}\right), \tilde O \left(k^{\alpha - \beta + 2\alpha} \right)  \right).
\]
The first term is increasing in $\beta$ and the second is decreasing.
Balancing, $\beta = \alpha - \nf{1}{3}$, or $\eps' = k^{1/3} \eps$.

\begin{proof}[Proof of \Cref{thm:ws-ls-pav-correctness}]
	The sample complexity claim is clear from the way in which \Cref{alg:sample-ws-ls-pav} is structured.
	The \pav{} step uses  $\tilde O(\eps^{-2} k^{1/3})$ samples (line 1),
	while each of the $2 k^{1/3}$ local search pivot steps uses $\tilde O(\eps^{-2})$ samples (line 7).
	Summing yields the stated sample bound.

	The main task is to prove correctness.
	This amounts to showing two things: first, that both the approximate \samplepav{} step and all of the pivot steps are sufficiently good with respect to their impact on $\pavsc(W)$ for our running committee $W$.
	Second, that if all steps are sufficiently good then we terminate with a committee satisfying \ejrpe{}.
	We will tackle these in order.

	First, let the event $\cE_{\pav{}}$ denote the failure event for the first step, and let $\{\cE_\ell\}$ denote the failure events for the pivot steps $\ell \in [3k^{1/3}]$.
	Our goal  will be to show that the probability of a failure over the course of the algorithm is at most $\delta$.
	First, we will define $\cE_{\pav{}}$ to be the event that $\pavsc(W) \leq \max_{W'} \pavsc(W') - \eps'$, where we will define $\eps' = \eps \cdot k^{1/3}$.
	Upper bounding this event closely follows the analysis for \samplepav (\Cref{prop:sample-pav-correct}).
	As in the analysis of \samplelspav (\Cref{prop:sample-ls-pav-correct}), our per-round failure event $\cE_\ell$ will be that either (a) we neglect to choose an $\eps/3$-good pivot when an $\eps$-good pivot is available, or (b) we choose a pivot that is less than $\eps/3$-good.

	The second step is to show that if neither 	$\cE_{\pav{}}$ nor any of the $\cE_\ell$ occur, then \samplewslspav{} satisfies \ejrpe{}.
	To this end, suppose that neither $\cE_{\pav{}}$ nor any of the $\cE_\ell$ hold.
	Then the committee $W$ chosen in line 5 satisfies $\hat s_W \geq \hat s_{W^*} - \eps'$, and at most $\eps'/2\eps = 3k^3$ many $\eps/3$-good pivots can be performed starting from $W$.
	Since none of the $\cE_\ell$ hold, so long as an $\eps$-good pivot is possible we will make a $\eps/3$-good pivot, and when we halt no $\eps$-good pivots will be possible.
	Taken together, this means that when we halt it is with $Q=\emptyset$, and no $\eps$-good pivots remain.
	Therefore by \Cref{prop:pav-score-to-ejrp-guarantee} our chosen $W$ satisfies \ejrpe{}.

	It therefore remains only to bound the likelihood of failure.
	We begin with $\cE_{\pav{}}$.
	Observe that if $\cE_{\pav{}}$ holds then either (i) the $\pavsc$-maximizing committee $W^*$ has sample $\hat s_{W^*} \leq \pavsc(W^*) - \nf{\eps'}{2}$ or (ii) at least one of the low-scoring committees $W$ with $\pavsc(W) \leq \pavsc(W^*) - \eps'$ has sample score $\hat s_{W} \geq \pavsc(W^*) - \nf{\eps'}{2}$.
	Let $s^* \defeq \max_{W} \pavsc(W)$ for convenience.
	Then by Hoeffding, the probability of (i) over the samples $\cA = (A^1, \ldots, A^t)$ is
	\begin{equation}
		\probover{\cA \sim \cP^t}{\hat s_{W^*} \leq s^* - \nf{\eps'}{2}}
		\leq \prob{t \cdot \hat s_{W^*} \leq t\cdot \expect{\hat s_{W^*}} - t\cdot \nf{\eps'}{2}}
		\leq \exp\left(- \frac{\eps'^2 t^2/2 }{t H_k^2 } \right) = \exp\left(- \frac{\eps'^2 t }{2 H_k^2 } \right),		 \notag
	\end{equation}
	where $H_k \defeq \sum_{j=1}^k \nf{1}{j}$ is the $k$th harmonic number, and a single sample's impact on $t \cdot \hat s_W$ is between $-H_k$ and $H_k$.
	Similarly, for any low-score $W$ the probability of (ii) is at most
	\begin{equation}
		\probover{\cA\sim \cP^t}{\hat s_{W} \geq s^*  - \nf{\eps'}{2}}
		\leq \prob{t \cdot \hat s_{W} \geq t \cdot \expect{\hat s_W} + t\cdot \nf{\eps'}{2}}
		\leq \exp\left(- \frac{\eps'^2 t}{2 H_k^2} \right).		 \notag
	\end{equation}
	Union bounding over all $W$ and introducing a failure probability $\delta/2$ yields
	\[
	\probover{\cA\sim \cP^t}{\cE_{\pav{}}}
	\leq \binom{m}{k} \cdot \exp\left(-\frac{ \eps'^2 t}{2 H_k^2}\right)
	\leq \exp\left(k \log m - \frac{ \eps'^2 t}{2H_k^2}\right)
	\leq \nf{\delta}{2}.
	\]
	Substituting $\eps' = \eps k^{1/3}$ and $H_k \leq \log k + 1$ and rearranging, we find that this is satisfied provided
	\[
	t \geq 2 (\log k + 1)^2\frac{1}{\eps^2 k^{2/3}} \left( k \log m + \log \nf{1}{\delta}, \right).
	\]
	which holds for our chosen number of samples on line 1.

	With our remaining failure probability budget of $\nf{\delta}{2}$, we now turn to the events $\{\cE_\ell\}$.
	Within each round $\ell$, the number of pivots under consideration is $\abs{(C \setminus W) \times C} \leq km$.
	And within this round, according to our goals we can identify the sufficient conditions for success that (a) no $\eps$-good pivots appear worse than $2\eps/3$, and (b) no $\eps/3$-bad pivots appear better than $2\eps/3$.
	For fixed $W$ and a pivot $q = (c, c')$, let $\hat\Delta_q \defeq \Delta(W, c, c', \cA)$ denote the empirical estimate of its $\pavsc$ impact calculated in line 8.
	In the first case, for a $\eps$-good pivot $q$ we can use Hoeffding to bound the probability it is not included in $Q$:
	\begin{equation}
		\probover{\cA\sim \cP^t}{\hat \Delta_{q} \leq \eps - \nf{\eps}{3} }
		\leq  \prob{t \cdot \hat \Delta_{q} \leq t \cdot \expect{\hat \Delta_{q}} - t \cdot \nf{\eps}{3} }
		\leq \exp\left(- \frac{2 \eps^2 t^2/9}{t (2H_k)^2} \right)
		= \exp\left(- \frac{ \eps^2 t}{18 H_k^2} \right).		 \notag
	\end{equation}
	We can similarly use Hoeffding to upper bound the probability of a $\nf{\eps}{3}$-bad pivot $q$ appearing in $Q$.
	Union bounding over both this and the at most $km$ $\nf{\eps}{3}$-bad pivots, we have that the probability of failing to include an $\eps$-good pivot in $Q$ or including a $\nf{\eps}{3}$-bad pivot in $Q$ is at most
	\[
	\sum_{\eps\text{-good } q} \prob{q \not \in Q} + \sum_{\nf{\eps}{3}\text{-bad } q} \prob{q \in Q}
	\leq  km \cdot \exp\left(- \frac{\eps^2 t}{18 H_k^2} \right)
	\leq \delta',
	\]
	where $\delta'$ is a per-round failure probability.
	To make sure our failure is at most $\nf{\delta}{2}$ over all $3k^{1/3}$ rounds, we set $\delta' = \delta / (6 k^{1/3})$ and rearrange to obtain
	\[
	t \geq \frac{18 H_k^2}{\eps^2} \left(\frac{4}{3} \log k + \log 6 + \log \nf{m}{\delta} \right).
	\]
	Since the number of samples we choose on line 7 satisfies this bound, our failure probability over all pivot steps is at most $\nf{\delta}{2}$.

	This establishes that the failure probability of \Cref{alg:sample-ws-ls-pav} is at most $\delta$ and concludes the proof.
\end{proof}

%%%
\section{Supplementary Material and Proofs for \texorpdfstring{\Cref{sec:lowerbounds}}{Section 5}}
\label{app:lower-bounds}

\smallepslb*

\begin{proof}[Proof of \Cref{thm:small-eps-lower-bound}]
	We will use the indistinguishable family of instances in \Cref{fig:partition-all-large}.
	Consider that the instance $\cI_i$ is chosen from $\cI$ by sampling $i \sim [k+1]$ uniformly at random.
	We will abuse notation slightly and
	As we have just argued, for samples $\cA = (A^1, \ldots, A^t) \sim \prof_i^t$ a sample-based rule $f:  (2^{\cC})^t \rightarrow \comms$ succeeds at finding a \jre{} committee if and only if the corresponding classifier $\Psi:  (2^{\cC})^t \rightarrow [k+1]$ correctly identifies which $\cI_i$ the samples are drawn from.

	As in the second part of the proof of the lower bound in \Cref{lem:separating-means-from-samples-new}, we will consider only two sources and our proof will follow a standard approach.
	Concretely, consider the strictly easier problem for $\Psi$ of identifying the source $i \sim \{1,2\}$, and assume its output is in $ \{1,2\}$ accordingly.
	The first step is to lower bound the failure probability of any deterministic $\Psi$ in terms of $\dtv(\cP_1^t, \cP_2^t)$.
	Let $F_1 \defeq f^{-1}(1)$ and $F_2 \defeq f^{-1}(2)$; note that $F_1  \sqcup F_2$ partitions the space of observed samples $(2^{\cC})^t$.
	Then letting the failure probability be $\delta$, we have
	\begin{align}
		\delta  = \probover{i \sim \{1,2\}}{\Psi(\cA) \neq i}
		& = \frac{1}{2}\cdot \probover{\cA\sim \prof_1^t}{\Psi(\cA) = 2} + \frac{1}{2} \cdot \probover{\cA\sim \prof_2^t}{\Psi(\cA) = 1} \notag \\
		& = \frac{1}{2} - \frac{1}{2} \left(\prof_1^t(F_1)- \prof_2^t(F_1)\right) \notag \\
		& \geq \frac{1}{2} -  \frac{1}{2}\cdot \dtv\left(\cP_1^t, \cP_2^t \right), \label{eq:tv-bounds-error-LB}
	\end{align}
	where the last step follows from the definition of $\dkl$.
	From here, we will bound $\dtv$ in terms of $\dkl$ in order to relate the statistical distance for general $t$ to the distance for a single sample.
	By Pinsker's inequality,
	\footnote{If we wished to establish a dependence on $\log \nf{1}{\delta}$, we would use the Bretagnolle–Huber bound or Tsybakov's bound, as in the proof of \Cref{lem:separating-means-from-samples-new}. But for our purposes Pinsker's inequality is enough.}
	\begin{align}
		\dtv\left(\cP_1^t, \cP_2^t \right)
		&\leq \sqrt{\frac{1}{2} \dkl(\cP_1^t \| \cP_2^t)}   \leq \sqrt{\frac{1}{2} t\cdot \dkl(\cP_1 \| \cP_2)},  \label{eq:pinsker-for-product-dist-LB}
	\end{align}
	where the second step uses the helpful fact that $\dtv$ \emph{tensorizes}, meaning that for product distributions $\dkl(\cP^t \| \cQ^t) \leq t \cdot \dkl(\cP \| \cP)$.
	Finally, what is the KL divergence between $\prof_1$ and $\prof_2$?
	Letting $r \defeq \frac{1}{k+1} + \eps$ and $\eps' = (k+1) \eps$, we have
	\begin{align}
		\dkl(\prof_1 \| \prof_2) &= \sum_{i \in [k+1]} \prof_1({i}) \cdot \log \left( \frac{\prof_1({i})}{\prof_2({i})} \right) \notag \\
		&= r \cdot \log \left(\frac{r}{r - \eps'}\right) + (r-\eps')  \cdot \log \left(\frac{r - \eps'}{r} \right)  \notag \\
		&= \eps' \cdot \log \left(\frac{1}{1 - \nf{\eps'}{r}}\right)  \notag \\
		&\leq \eps' \cdot 2\frac{\eps'}{r} ,  \label{eq:kl-upper-bound-for-LB}
	\end{align}
	provided that $\nf{\eps'}{r} \leq \nf{1}{2}$, which holds so long as $\eps \leq (k+1)^{-2}/2$.

	Finally, combining \eqref{eq:tv-bounds-error-LB} and \eqref{eq:pinsker-for-product-dist-LB} with this bound \eqref{eq:kl-upper-bound-for-LB} and simplifying slightly yields
	\[
	\delta \geq \frac{1}{2} - \frac{1}{2} \left(\frac{t}{2}  \eps^2 (k+1)^3 \right)^{1/2}.
	\]
	For constant $\delta < \nf{1}{2}$ this implies $t \geq 2(1 - 2\delta)^2 \eps^{-2} (k+1)^{-3}$, and so $t = \Omega(\eps^{-2} k^{-3})$ as claimed.
\end{proof}

\optimalrulestructure*

\begin{proof}[Proof of \Cref{lem:optimal-part-rule-structure}]
	Given $t$ samples $\cA = (A^1, \ldots, A^t) \sim \cJ_i$ for unknown $i$, let the empirical support of each candidate $j \in [k+1]$ be denoted by $\bar\mu_j \defeq \frac{1}{t} \sum_{\ell \in [t]} \1\{j \in A^\ell\}$, and observe $ \expect{\bar\mu_j} = \frac{1}{k+1} + \eps$ for $j = i$ and $ \expect{\bar\mu_j} = \frac{1}{k+1} - \frac{\eps}{k}$ for $j \neq i$.
	Next observe that in order for $f^*$ to maximize its probability of choosing a \jre{} committee, it must once again maximize its probability of correctly identifying the unknown index $i$.
	This is because, for chosen output committee $f^*(\cA)$, if $i \in f^*(\cA)$ then $f^*$ satisfies \jre{}, and if $i \not \in f^*(\cA)$ then it does not.

	We now constrain the structure of $f^*$.
	First, since all permutations of any given $\cA = (A^1, \ldots, A^t)$ are equally likely, it is without loss of generality to assume that the optimal rule $f^*$ is a function only of $(\bar\mu_1, \ldots, \bar\mu_{k+1})$.
	Similarly, since all $\cJ_i \in \cJ$ are equally likely to be chosen, all indexing permutations of $\bar \mu = (\bar\mu_1, \ldots, \bar\mu_{k+1})$ are equally likely.
	Therefore it is without loss of generality to assume that $f^*$ is a function of the order of the values in $\bar\mu$, with ties broken lexicographically.

	Let the order statistic $\bar\mu'_{(1)}$ be the empirical minimum support over all candidates/blocs.
	Also let $\bar n_j \defeq t \cdot \bar\mu'_j$ for all $j \in [k+1]$ be the corresponding histogram of support arising from the samples $\cA$, and $\bar n_{(1)}, \ldots \bar n_{(k+1)}$ its corresponding order statistics.
	For a given sequence of samples and corresponding $\bar \mu$, let $\sigma: [k+1] \rightarrow [k+1]$ map the order statistics to the indices that realize them; that is, $\sigma(j)$ is the coordinate for which $\bar\mu_{\sigma (j) } = \bar\mu_{(j)}$.
	Also, let $g^*$ denote the index that is \emph{unchosen} by $f^*$, so that $f^*(\cA) = [k+1] \setminus \{g^*(\cA)\}$.
	We will show that $g^*(\cA) = \sigma(1)$ without loss of generality.

	To do this, assume for the sake of contradiction that $f^*$ does something else on a profile which is not the same up to lexicographic tiebreaking.
	In other words, there is a support histogram $\bar \mu^\bot$ such that $g^*(\bar \mu^\bot) = \sigma^{-1}(\ell)$, but $\bar\mu^\bot_{(1)} < \bar\mu^\bot_{(\ell)}$, meaning that the empirical support of candidate $g^*(\bar \mu^\bot)$ on $\bar \mu^\bot$ is strictly non-minimal.
	Consider the rule $f'$ and corresponding $g'$ that agrees with $f^*$ except on $\cA$ for which $\bar\mu(\cA) = \bar\mu^\bot$.
	We will show that by instead choosing $g'(\bar\mu^\bot) = \sigma^{-1}(1)$ and outputting $f'(\bar \mu^\bot) = [k+1] \setminus \sigma^{-1}(1)$, $f'$ strictly improves upon $f^*$ in terms of error probability.

	Note that for any $\bar\mu =\bar\mu(\cA)$ derived from samples from unknown $\cJ_i$, both $f^*$ and $f'$ succeed so long as neither $g^*(\bar\mu) = i$ nor $g^*(\bar\mu) =i$, where $i$ is the distinguished index for $\cJ_i$.
	Therefore the difference in errors $\err(f^*) - \err(f')$ is given by
	\begin{align}
		&\phantom{{}={}} \probover{i \sim [k+1], \cA \sim \prof_i^t}{ g^*(\bar\mu^\bot) = i \:\land\: \bar\mu(\cA) = \bar\mu^\bot}
		- \probover{i \sim [k+1], \cA \sim \prof_i^t}{ g'(\bar\mu^\bot) = i \:\land\: \bar\mu(\cA) = \bar\mu^\bot}   	 \notag \\
		&=\frac{1}{k+1} \sum_i \left(
		\probover{\cA \sim \prof_i^t}{ \sigma(\ell) = i \:\vert\: \bar\mu(\cA) = \bar\mu^\bot}
		- \probover{\cA \sim \prof_i^t}{ \sigma(1) = i \:\vert\: \bar\mu(\cA) = \bar\mu^\bot}
		\right) \probover{\cA\sim\prof_i^t}{\bar\mu(\cA) = \bar\mu^\bot}	 \notag \\
		&=\frac{1}{k+1} \left(
		\probover{\cA\sim\prof_{\sigma(\ell)}^t}{\bar n(\cA) = \bar n^\bot}
		- \probover{\cA\sim\prof_{\sigma(\ell)}^t}{\bar n(\cA) = \bar n^\bot} \right).	 \notag
	\end{align}
	Therefore it suffices to show this difference is nonnegative.
	We do this by observing that $\bar n = t\cdot \bar \mu$ follows the multinomial distribution $\Multin(p_{i,1}, \ldots, p_{i,k+1})$ for each $\cJ_i$, where $p_{i,j} = \frac{1}{k+1} + \eps$ for $j = i$ and $p_{i,j} = \frac{1}{k+1} - \nf{\eps}{k}$ for $j \neq i$.
	Therefore
	\begin{align}
		\frac{\probover{\cA\sim\prof_{\sigma(\ell)}^t}{\bar n(\cA) = \bar n^\bot}}
		{ \probover{\cA\sim\prof_{\sigma(1)}^t}{\bar n(\cA) = \bar n^\bot}}
		&= \frac{
			\frac{t!}{\prod_j \bar n_j !} \prod_j p_{\sigma(\ell),j}^{\bar n_j}
		}{
			\frac{t!}{\prod_j \bar n_j !} \prod_j p_{\sigma(1),j}^{\bar n_j}
		}
		= \frac{
			p_{\sigma(\ell),\sigma(1)}^{\bar n_{(1)}} p_{\sigma(\ell),\sigma(\ell)}^{\bar n_{(\ell)}}
		}{
			p_{\sigma(1),\sigma(1)}^{\bar n_{(1)}} p_{\sigma(1),\sigma(\ell)}^{\bar n_{(\ell)}}
		}
		= \left(\frac{\frac{1}{k+1} + \eps}{\frac{1}{k+1} - \nf{\eps}{k}}\right)^{\bar n_{(\ell)}- \bar n_{(1)}}
		> 1. \notag
	\end{align}
	Hence the rule $f'$ has strictly smaller error than $f^*$, a contradiction.
	We conclude that for $\cJ_i \sim \cJ$, the \jre{}-optimal rule always chooses all but the sample-support-minimizing candidate as its output.
\end{proof}

\complexitymdependence*

\begin{proof}[Proof of \Cref{thm:sample-complexity-depends-on-m}]
	Let $T = [k^2]$ be the voter types and let $S \subseteq T$ be the voter types observed within the sample $\cA$, and let $\bar S \defeq T \setminus S$.
	If $t \leq k^{d}/2$ then necessarily $\abs{S} \leq \abs{T}/2$.
	Each candidate $c\in\cands$ has an incidence $S_c = \{A \in T: c \in A\}$ on the sampled set of types and likewise on $\bar S$, with $\abs{S_c} + \abs{\bar S_c} = k^{d-1}$.

	Suppose that on these samples $f^*$ chooses a collection of winning candidates $W$ with observed incidences $S_W \defeq \{S_c\}_{c \in W}$ and unobserved incidences $\bar S_W$.
	Let $\cW\vert_{S_W}$ be the collection of all $W \in \cW$ that share these observed incidences.
	Since $f^*$ cannot distinguish between them from $\cA$ alone, all $W \in \cW\vert_{S_W}$ are equally likely to be its output.
	But since every $k$-subset of voter types corresponds to a candidate, the if $W \sim \cW\vert_{S_W}$ on $\bar S$ then the incidences of each of the $c \in W$ on $\bar S$ are independent random variables.
	If $r_c \defeq k - \abs{S_c}$ for each $c \in \cands$, then the expected number of voters in $\bar S$ approving of at least one candidate in $W$ is given by
	\[
	\abs{\{A \in T: A \cap W \neq \emptyset\}} = \sum_{A \in T} \1\{ A \cap W \neq \emptyset \} = \sum_{A \in T} R_A,
	\]
	Where $R_A \defeq \lor_{c \in W} R_{A,c}$ is the event that $A\in T$ is in the incidence of $c$.
	For each $A$, the $R_{A,c}$ over all $c$ are independent events, with $\probover{W}{R_{A,c}} = \frac{r_c}{\abs{\bar S}} \leq r_c \frac{2}{k^d} $.
	Since they are independent,
	\[
	\probover{W}{R_A} = 1 - \prob{\land_c \neg R_{A,c}} = 1-\prod_c \left(1-\prob{R_{A,c}}\right) \leq 1 - \exp\left( -2 \sum_c r_c \frac{2}{k^d}\right) \leq 1-e^{-4},
	\]
	provided $4r_c \leq k^d$ (e.g. if $k \geq 4$).
	This then holds since $1-x \geq e^{-2x}$ for $x \in [0,1/2]$ and since $\sum_{c \in W} r_c \leq k^d$.
	Therefore there is at least a constant chance of each voter type $A \in \bar S$ remaining uncovered by the choice of $W$.
	Since the events $\{R_A\}_{A \in \bar S}$ are negatively associated, the number of uncovered $A$ will exhibit strong concentration around its mean $\abs{\bar S} \probover{W}{R_A}$.
	In particular, let $U_W$ denote the number of voter types in $\bar S$ not incident to any $c \in W$. Then by Hoeffding,
	\[
	\probover{W}{ U_W \geq \left(1 - \nf{1}{2} \cdot e^{-4}\right) \cdot \abs{\bar S}} \leq \exp\left(- \Omega\left(k^d\right)\right).
	\]
	For $k$ sufficiently large that $\abs{S} \cdot \frac{1}{2 e^4} \geq \frac{k^d}{4 e^4} \geq k$, this means even  the $W$ chosen by $f^*$ will fail \jr{} with high probability.
	This concludes the proof.
\end{proof}

%%%
\section{Proofs for \texorpdfstring{\Cref{sec:vcdim}}{Section 6}}

\bwclargecandub*
\begin{proof}[Proof of \Cref{thm:large-candidate-ub}]
	We will reduce to \Cref{thm:cs-pav-correctness} in the form of \Cref{cor:sample-complexity-ejrp-relaxed-ub}.

	The key insight is that if we assume that we take the large candidate $c$, then we may satisfy \jre{} for any $\eps >0$ on the profile overall by (a) treating the remainder of the profile that does not support $c$ as its own standalone profile, (b) finding a committee $W'$ with $k' = k-1$ and some $\eps' > \eps$ on this smaller profile, and then (c) choosing $W = W' \cup \{c\}$ as our final committee for the original profile.

	In particular, $W$ satisfies \jre{} if there are no unchosen candidates with support at least $\frac{1}{k+1} + \eps$.
	If $W'$ has no unchosen candidate with this much support, then $W$ satisfies \jre{}.
	But for $W'$ \emph{with respect to its subinstance}, this prohibitively large candidate support size is
	\[
	\frac{1}{1-\beta}\left(\frac{1}{k'+1} + \eps'\right) = \frac{1}{k+1} + \eps.
	\]
	Observing that $k' = k-1$ and solving for $\eps'$, we get
	\[
	\eps' =\eps + \frac{1}{k+1}\left(\frac{\beta}{1-\beta}\right) - \frac{1}{k^2}.
	\]
	For any $\beta >0$ there is some constant $c_1$ such that $\eps' \geq \eps + c_1 \cdot \frac{1}{k}$.
	Plugging this $\eps'$ into \Cref{thm:cs-pav-correctness} gives the claimed bound (and also gives guarantees for $\beta = o_k(1)$ if so desired).
\end{proof}

\bwclargecovgub*

To show this, we will use the following:
\begin{lemma}[Bernstein's inequality]
	\label{lem:bernstein}
	If $X_i$ for $i\in [n]$ are independent with $\abs{X_i} \leq 1$ and $\expect{X_i} = 0$, then for any $t > 0$
	\[
	\prob{ X > t } \leq \exp\left( - \frac{t^2 }{2 \sigma^2 + 2 t/3} \right),
	\]
	where $X \defeq \sum_i X_i$ and $\sigma^2 \defeq \expect{X^2} = \sum_i \expect{X_i^2}$.
\end{lemma}

We now turn to the proof.

\begin{proof}[Proof of \Cref{thm:large-covg-ub}]
	Let $W^*$ be a voter-coverage-maximizing committee.
	Since $\gamma < 1$, any committee $W$ for which $\cov(W) \geq \cov(W^*) - \frac{1-\gamma}{k+1}$ will leave at most a $\nf{1}{k+1}$ proportion of voters uncovered and therefore satisfy \jre{} for any $\eps > 0$.
	Taking $\eps = \frac{1-\gamma}{k+1} = \Theta(k^{-1})$, we have by \Cref{lem:eps-mkc-sample-upper-bound} that such a $W$ can be found using $\tilde O (k^3)$ samples.
	To improve this to $\tilde O(k^2)$ samples, we will continue to use \samplemaxkcover (\Cref{alg:eps-mkc-brute-force}) and provide a tighter analysis via Bernstein's inequality (\Cref{lem:bernstein}).

	In particular, we will repeat the proof of \Cref{lem:eps-mkc-sample-upper-bound} using Bernstein instead of Hoeffding.
	Recall that by assumption we have $1 - \nf{1}{k+1} < 1 - \nf{\gamma}{k+1} < p_{W^*}$, where $p_{W} \defeq \probover{A \sim \prof}{A \cap W \neq \emptyset}$.
	We will define the threshold $q \defeq 1 - \frac{(\gamma + 1)/2}{k+1}$ and let the bad events be either (a) $W^*$ has sample coverage $\hat p_{W^*} \leq q$, or (b) $\hat p_{W} \geq q$ for any $W$ for which $p_W \leq 1-\nf{1}{k+1}$.

	For any $W$ and each approval in the samples $(A_1, \ldots, A_t) = \cA\sim \prof^t$, let $Z_A^W \defeq \1\{A \cap W \neq \emptyset \} -  p_W$.
	This is merely $\Bern(p_W)$ shifted to have mean $0$, and therefore its variance is $\sigma_i^2 = p_W(1-p_W) \leq (1-p_W)$.
	Let $Z^W \defeq \sum_{A\in \cA} Z_A^W$ and let $c \defeq (1-\gamma)/2$.
	Then by Bernstein, for bad event (a) we have
	\begin{align*}
		\probover{\cA\sim\prof^t}{\hat p_{W^*} \leq q} &= \prob*{Z^W \leq -t \cdot \frac{c}{k+1}} \leq  \exp\left(\frac{- t^2 \left( \frac{c}{k+1}\right)^2}{2 t p_{W^*}(1-p_{W^*}) + \frac{2}{3} t \frac{c}{k+1}} \right) \leq \exp\left( -c'\cdot \frac{t}{k}\right)
	\end{align*}
	for some constant $c'$, where we used that $1-p_{W^*} \leq \nf{1}{k+1}$.

	For a bad event of type (b), consider a $W$ for which $p_W \leq 1-\nf{1}{k+1}$.
	For convenience, let $r \defeq 1 - \nf{1}{k+1} - p_W$.
	Then again by Bernstein,
	\begin{align*}
		\probover{\cA\sim\prof^t}{\hat p_{W} \geq q}
		&= \prob*{Z^W \geq t \cdot \left(r + \frac{c}{k+1}\right)}
		\leq  \exp\left(\frac{- t^2 \left(r + \frac{c}{k+1}\right)^2}{2 t p_{W}(1-p_{W}) + \frac{2}{3} t \left(r + \frac{c}{k+1}\right)} \right) \\
		& \leq \exp\left( - c''\cdot t \left( r + \frac{c}{k+1} \right) \right) \leq \exp\left( -c'''\cdot \frac{t}{k}\right)
	\end{align*}
	for some constants $c''$ and $c'''$, where we used that $1-p_W = r + \frac{1}{k+1}$.

	The remainder of the proof proceeds along the same lines as that of \Cref{lem:eps-mkc-sample-upper-bound}, yielding a final bound of $O(k (k \log m  + \log \nf{1}{\delta}))$.
\end{proof}

\hardinstancestructure*

\begin{proof}[Proof of \Cref{claim:hard-instance-characterization}]
	We start with the first property.
	Given a budget of $\tilde O(k^3)$ samples, we can run \Cref{alg:sample-cs-pav} with a minimal tolerance of $\eps = k^{-3/2}$; let us call the resulting committee $W$, and suppose that none of the failure events bounded by $\delta$ occur over the course of the run of \Cref{alg:sample-cs-pav}.

	The remainder of this proof consists of analyzing the structure of $\cP$ under the conditions that $W$ is a successful output according to \Cref{alg:sample-cs-pav} but nevertheless fails to satisfy \jr{}.
	If $W$ fails \jr{}, then there is some $c' \not \in W$ with marginal coverage at least $\nf{1}{k}$.
	In particular, since coverage is submodular, for all $c \in W$ we have
	\[
	\nf{1}{k} \leq \cvg(W \cup \{c'\}) - \cvg(W) \leq \cvg(W \setminus \{c\} \cup \{c'\}) - \cvg(W \setminus \{c\}).
	\]
	On the other hand, if $W$ is a successful output of \Cref{alg:sample-cs-pav} then no pivots with improvement better than $k^{-3/2}$ remain: therefore for all $c \in W$,
	\[
	k^{-3/2} \geq \cvg(W \setminus \{c\} \cup \{c'\}) = \cvg(W \setminus \{c\} \cup \{c'\}) - \cvg(W \setminus \{c\}) - \left( \cvg(W) - \cvg(W \setminus \{c\})\right).
	\]
	Combining these inequalities yields
	\[
	\cvg(W) - \cvg(W \setminus \{c\}) \geq \frac{1}{k} - k^{-3/2}.
	\]
	This establishes the first property.

	Finally, the first property implies the second: if all candidates in $W$ have exclusive support at least $\frac{1}{k} - \frac{1}{k^{3/2}}$, then the probability $A \sim \prof$ is not an exclusive supporter is at most
	\[
	    1 - \sum_{c \in W} \probover{A\sim\prof}{A \cap W = \{c\}} \leq 1 - k\left(\frac{1}{k} - \frac{1}{k^{3/2}} \right) =  k^{-1/2}.
	\]
	This completes the proof.
\end{proof}

%%%
\section{Supplementary Material for \texorpdfstring{\Cref{sec:experiments}}{Section 8}}
\label{app:experiments}

%%%
\begin{figure}[h]
	\begin{subfigure}{.48\textwidth}
		\centering
		\includegraphics[width=\linewidth]{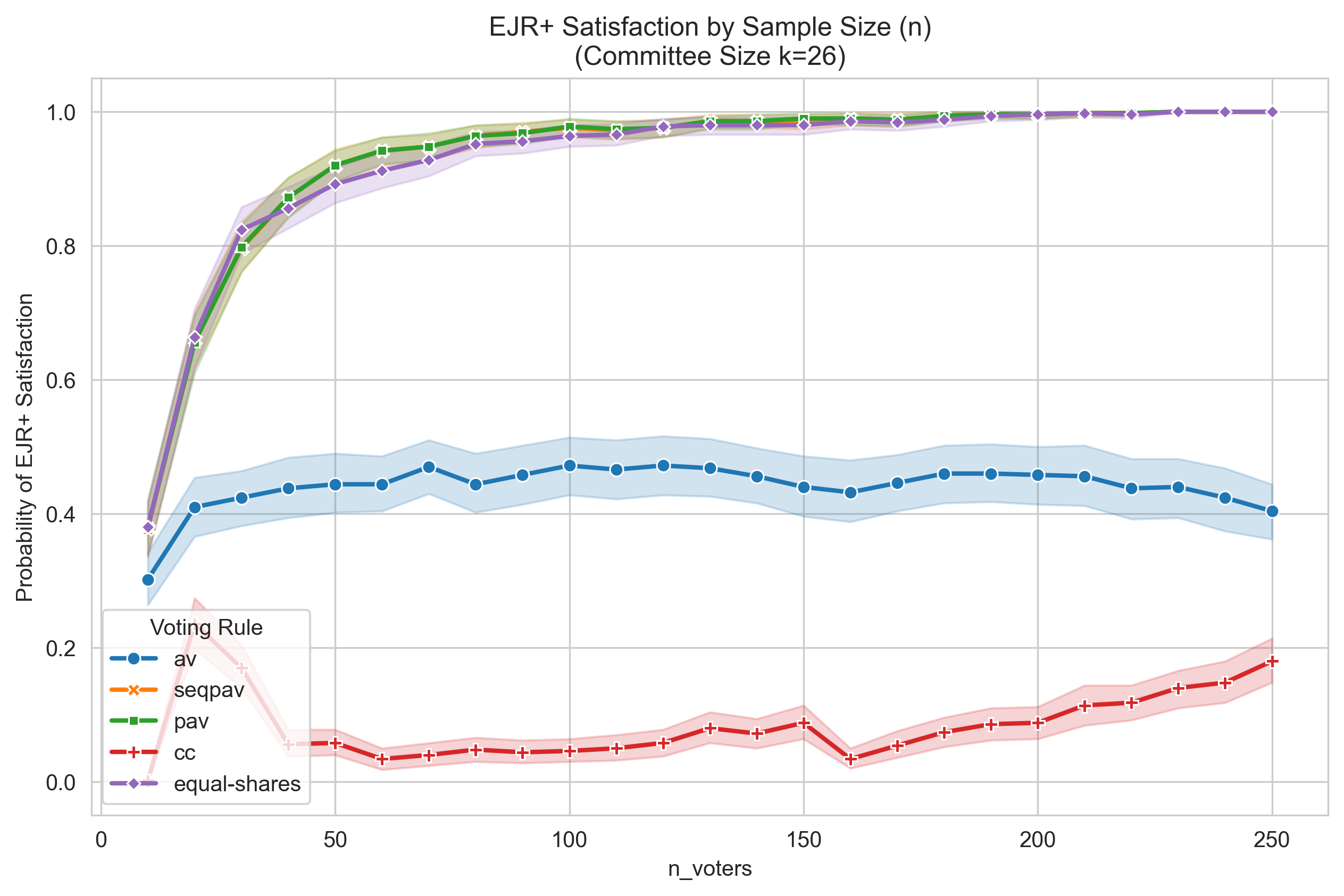}
		\caption{Probability of \ejrp{} for $k=26$}
		\label{fig:mokotow_k26_ejrp}
	\end{subfigure}%
	\hfill
	\begin{subfigure}{.48\textwidth}
		\includegraphics[width=\linewidth]{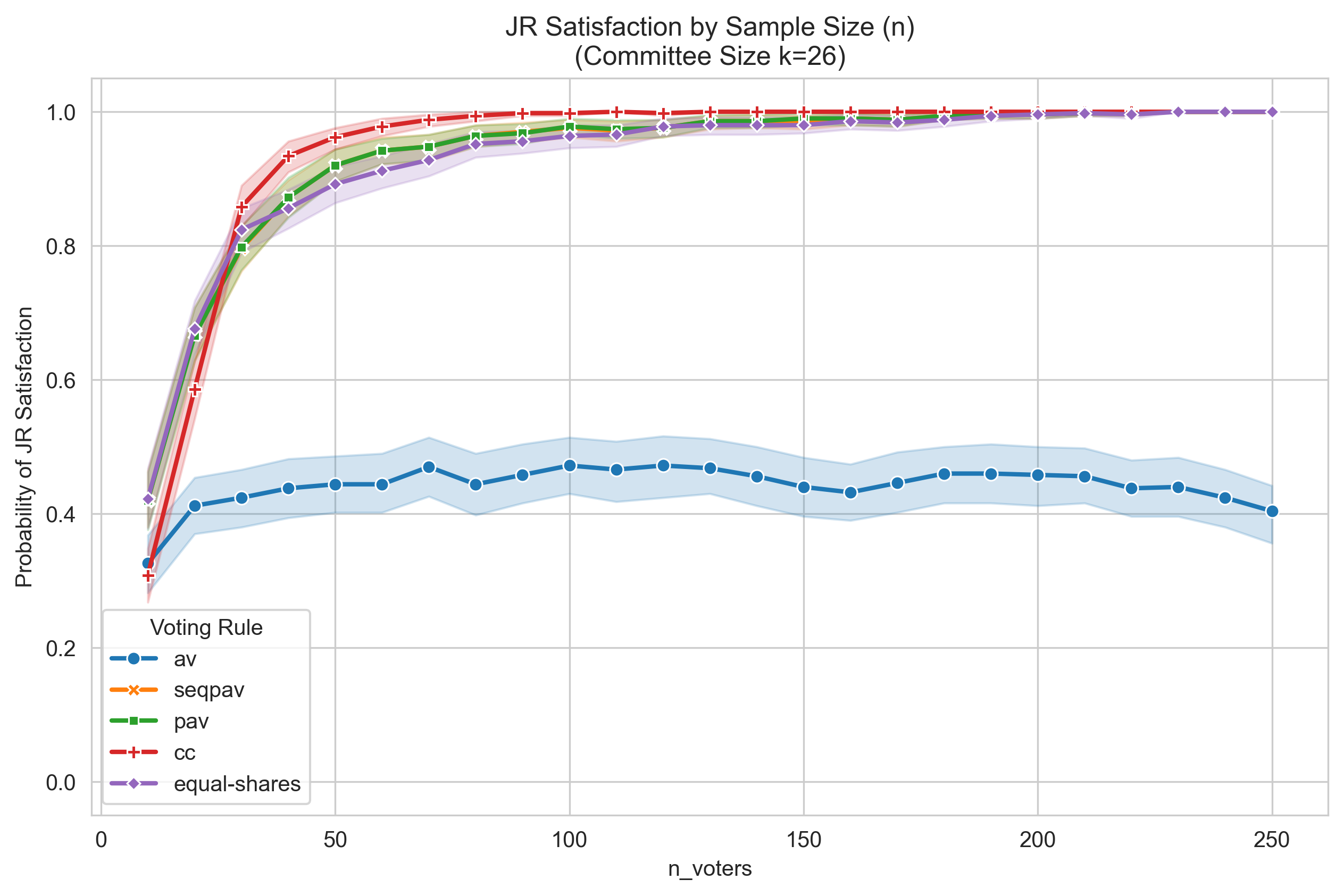}
		\caption{Probability of \jr{} for $k=26$}
		\label{fig:mokotow_k26_jr}
	\end{subfigure}
	\caption{Success probability of satisfying \jr{} and \ejrp{} for various rules on approvals from the ``Poland Warszawa 2022 Mokot\'ow'' participatory budgeting election, as the number of samples increases. Empirical probabilities are out of 500 trials for various rules selecting a committee of size $k=26$, out of 97 candidates overall.}
	\label{fig:mokotow_increasing_n_k26}
\end{figure}

%%%
\begin{figure}[h]
	\begin{subfigure}{.48\textwidth}
		\centering
		\includegraphics[width=\linewidth]{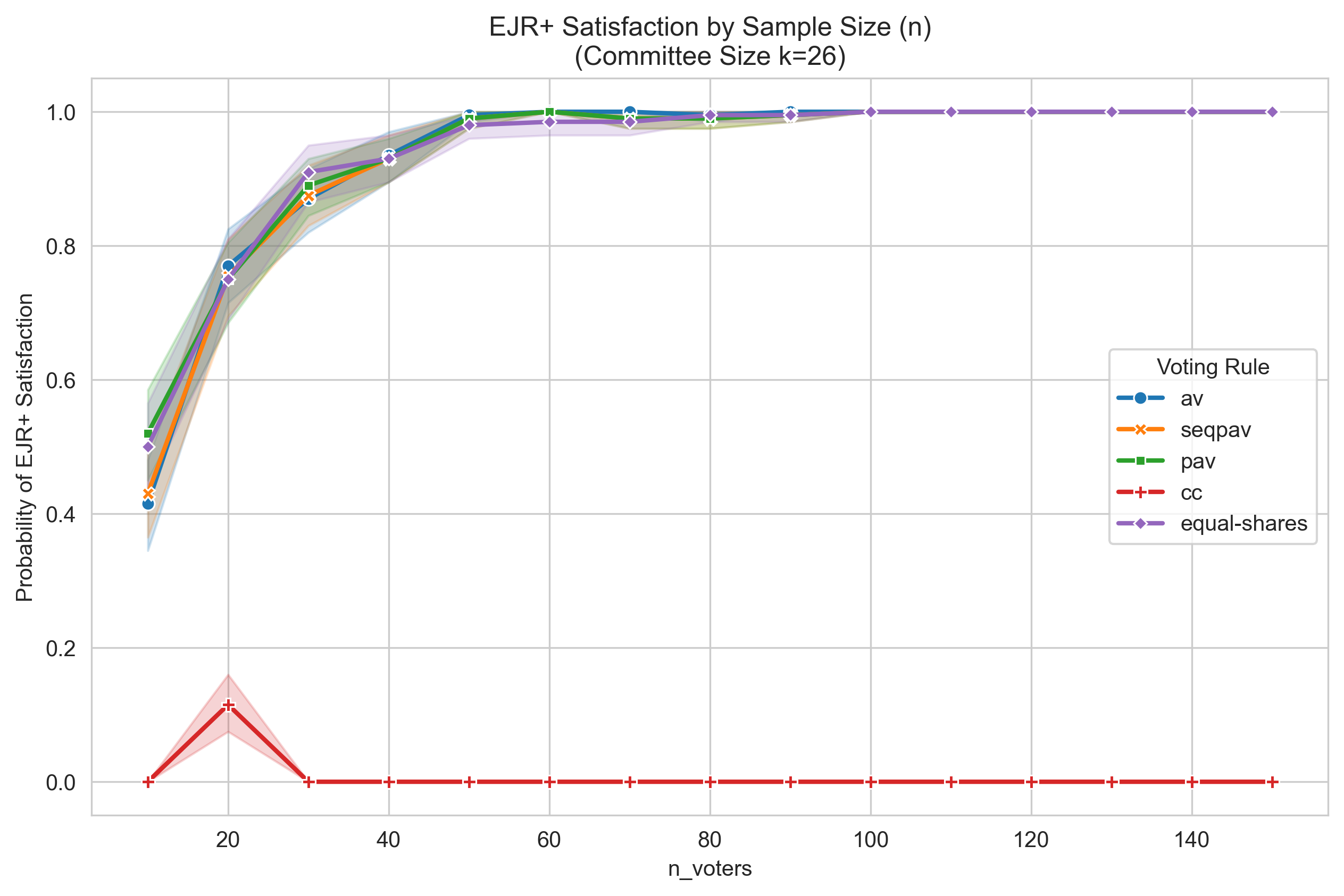}
		\caption{Probability of \ejrp{} for $k=26$}
		\label{fig:amsterdam_622_ejrp}
	\end{subfigure}%
	\hfill
	\begin{subfigure}{.48\textwidth}
		\includegraphics[width=\linewidth]{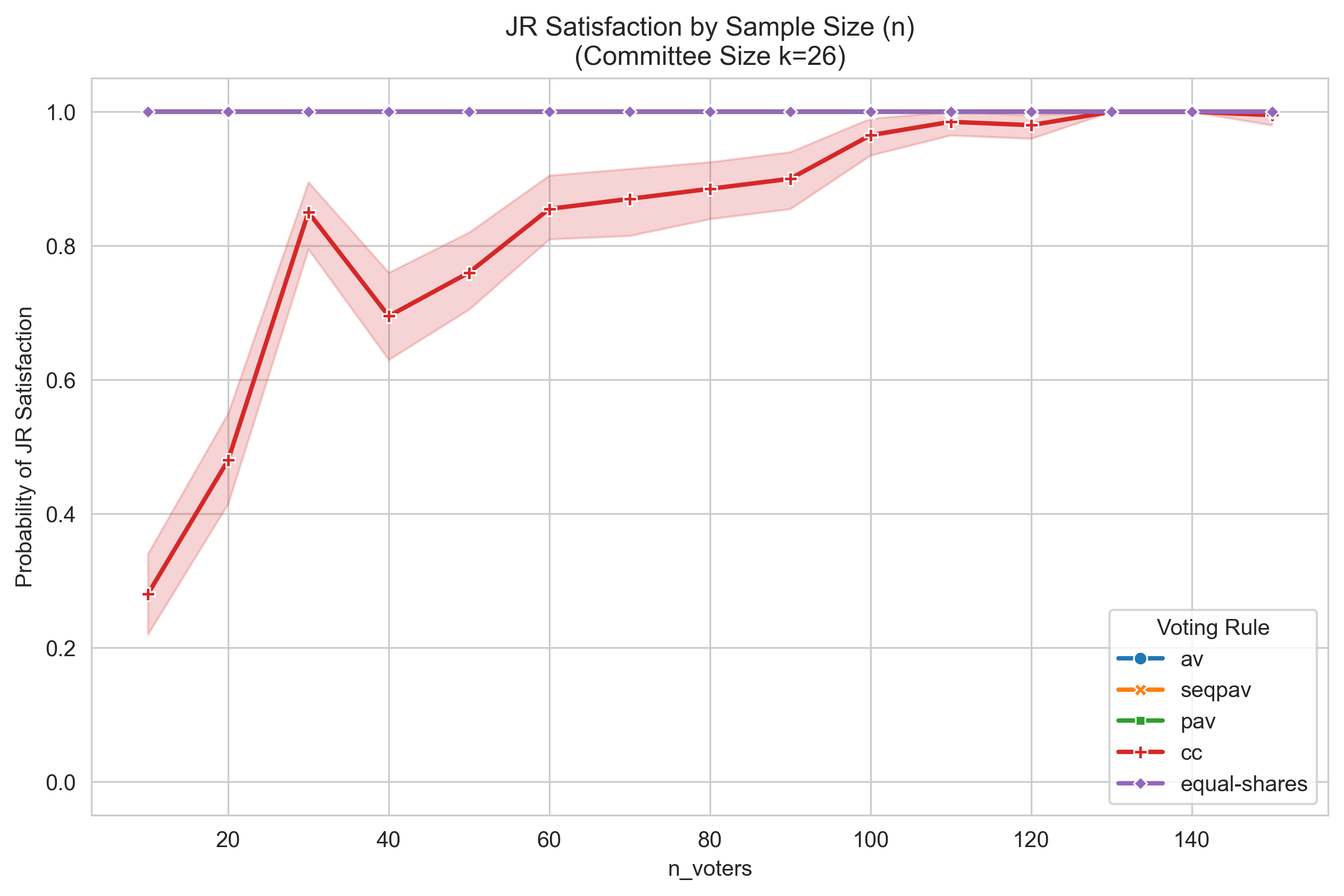}
		\caption{Probability of \jr{} for $k=26$}
		\label{fig:amsterdam_622_k26_jr}
	\end{subfigure}
	\caption{Success probability of satisfying \jr{} and \ejrp{} for various rules on approvals from the ``Netherlands Amsterdam 622'' participatory budgeting election, as the number of samples increases. Empirical probabilities are out of 500 trials for various rules selecting a committee of size $k=26$, out of 67 candidates overall.}
	\label{fig:amsterdam_622_increasing_n_k26}
\end{figure}

%%%
\begin{figure}[h]
	\begin{subfigure}{.48\textwidth}
		\centering
		\includegraphics[width=\linewidth]{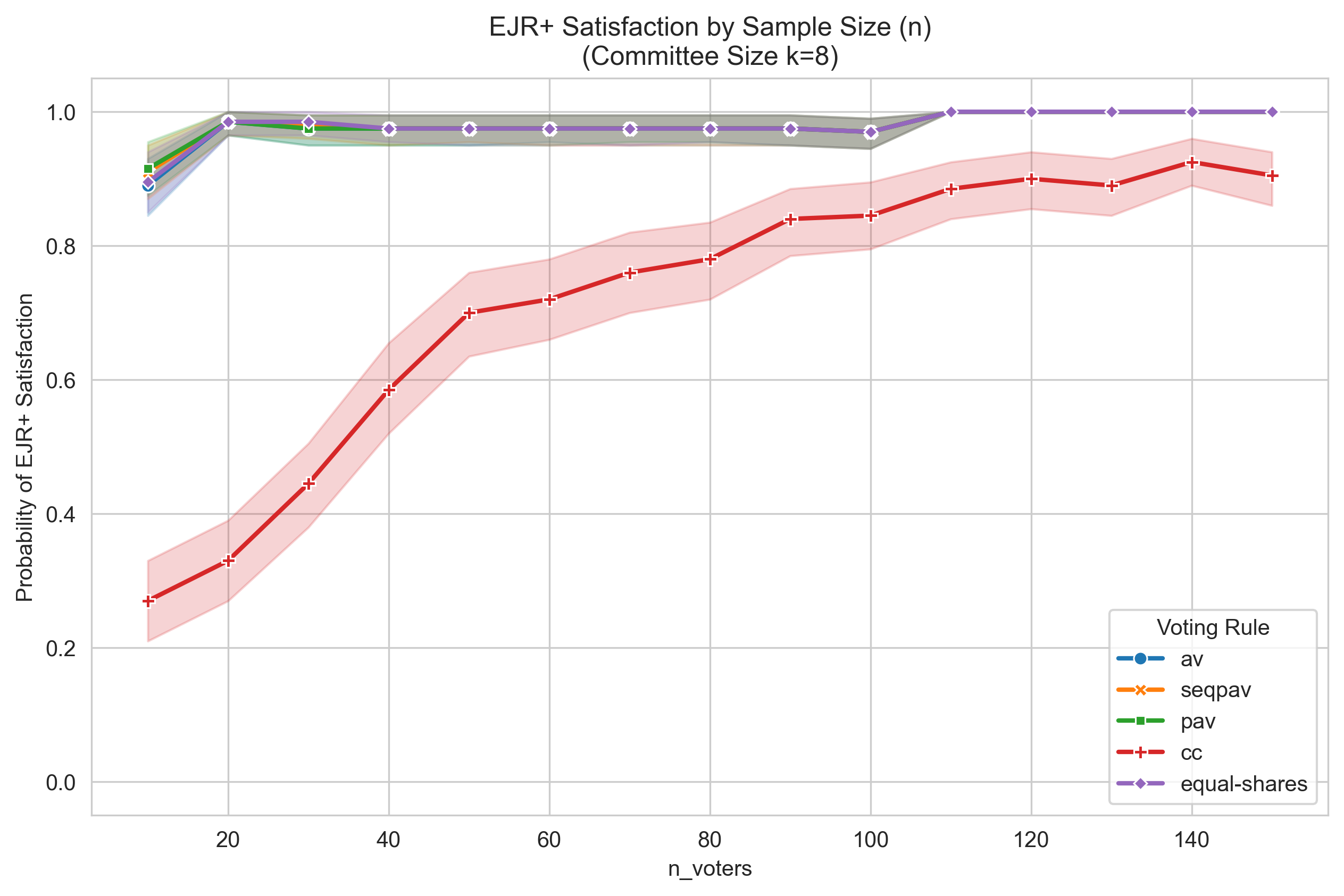}
		\caption{Probability of \ejrp{} for $k=8$}
		\label{fig:amsterdam_622_ejrp}
	\end{subfigure}%
	\hfill
	\begin{subfigure}{.48\textwidth}
		\includegraphics[width=\linewidth]{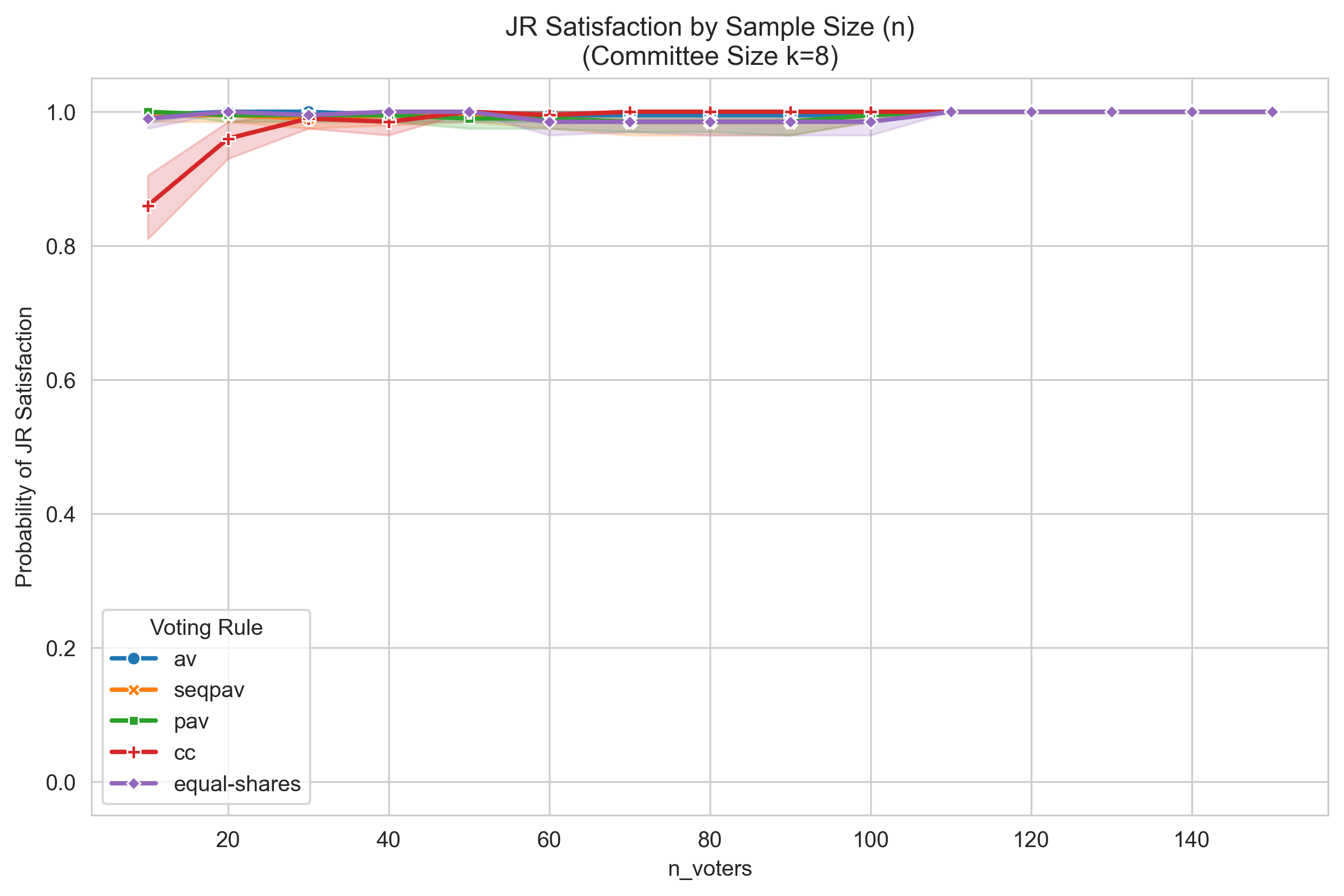}
		\caption{Probability of \jr{} for $k=8$}
		\label{fig:amsterdam_622_k26_jr}
	\end{subfigure}
	\caption{Success probability of satisfying \jr{} and \ejrp{} for various rules on approvals from the ``Netherlands Amsterdam 621'' participatory budgeting election, as the number of samples increases. Empirical probabilities are out of 500 trials for various rules selecting a committee of size $k=8$, out of 60 candidates overall.}
	\label{fig:amsterdam_621_increasing_n_k8}
\end{figure}

%%%
\begin{figure}[h]
	\begin{subfigure}{.48\textwidth}
		\centering
		\includegraphics[width=\linewidth]{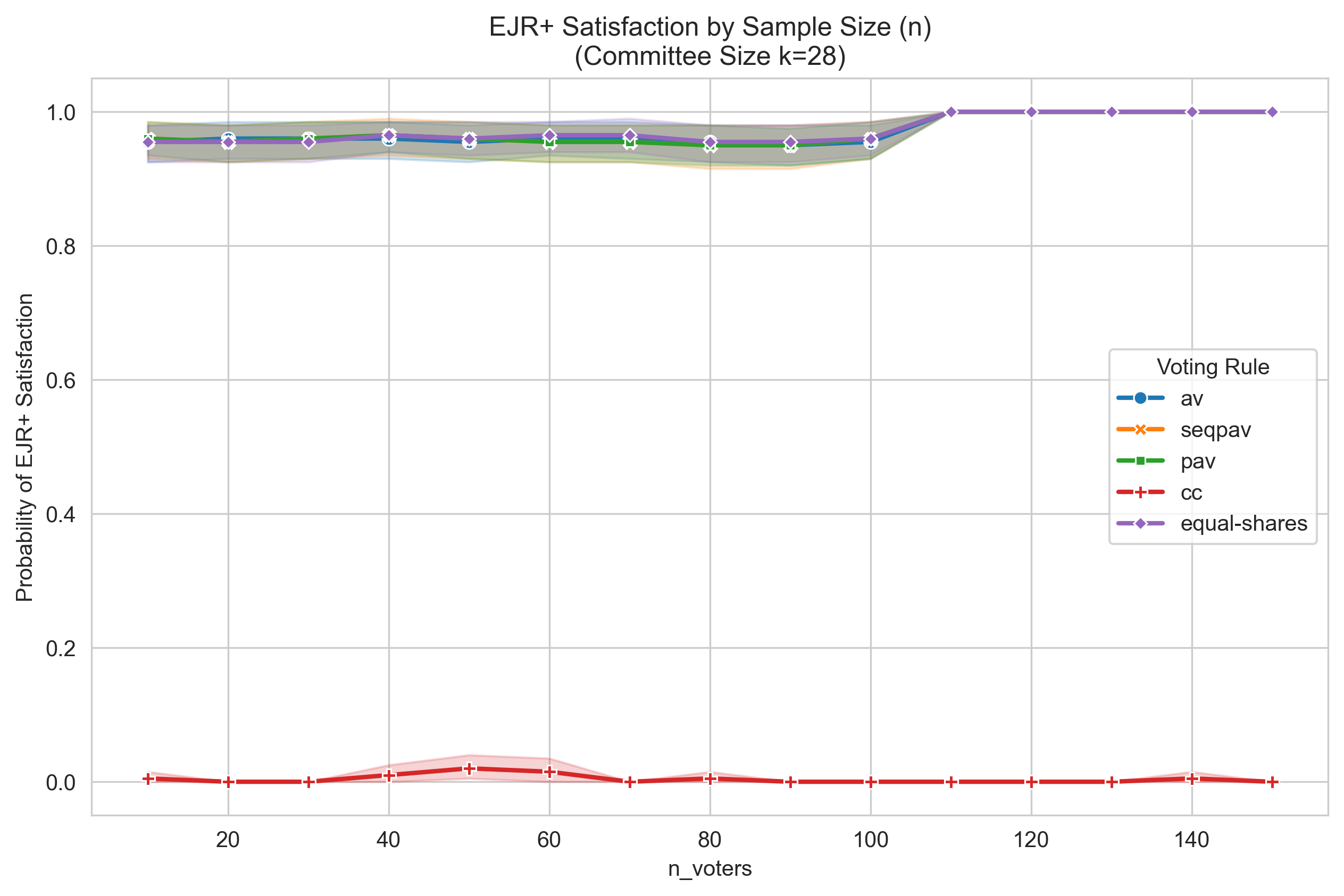}
		\caption{Probability of \ejrp{} for $k=28$}
		\label{fig:amsterdam_631_k28_ejrp}
	\end{subfigure}%
	\hfill
	\begin{subfigure}{.48\textwidth}
		\includegraphics[width=\linewidth]{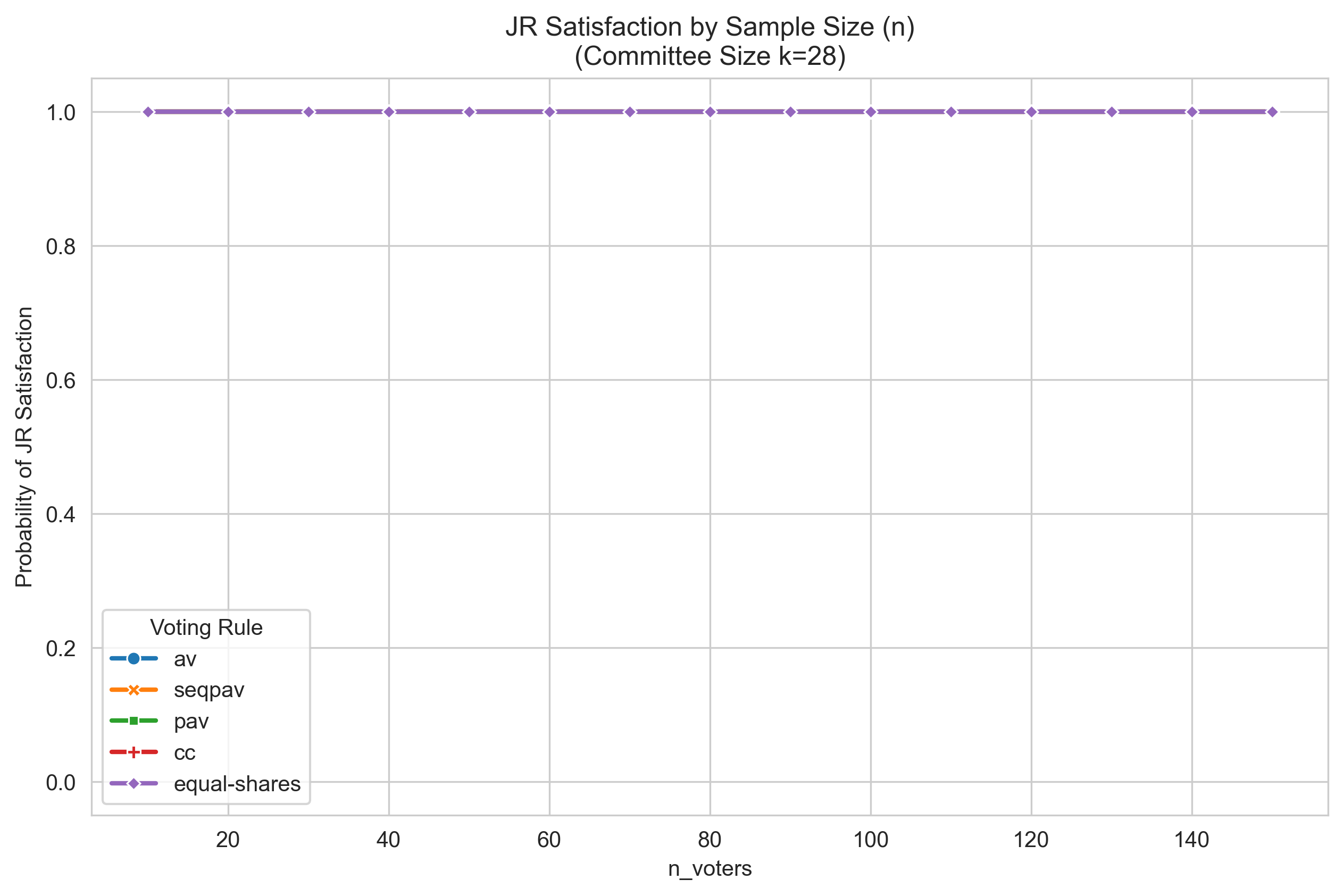}
		\caption{Probability of \jr{} for $k=28$}
		\label{fig:amsterdam_631_k28_jr}
	\end{subfigure}
	\caption{Success probability of satisfying \jr{} and \ejrp{} for various rules on approvals from the ``Netherlands Amsterdam 631'' participatory budgeting election, as the number of samples increases. Empirical probabilities are out of 500 trials for various rules selecting a committee of size $k=28$, out of 62 candidates overall.}
	\label{fig:amsterdam_631_increasing_n_k28}
\end{figure}

%%%
\begin{figure}
	\begin{subfigure}{.45\textwidth}
		\centering
		\includegraphics[width=\linewidth]{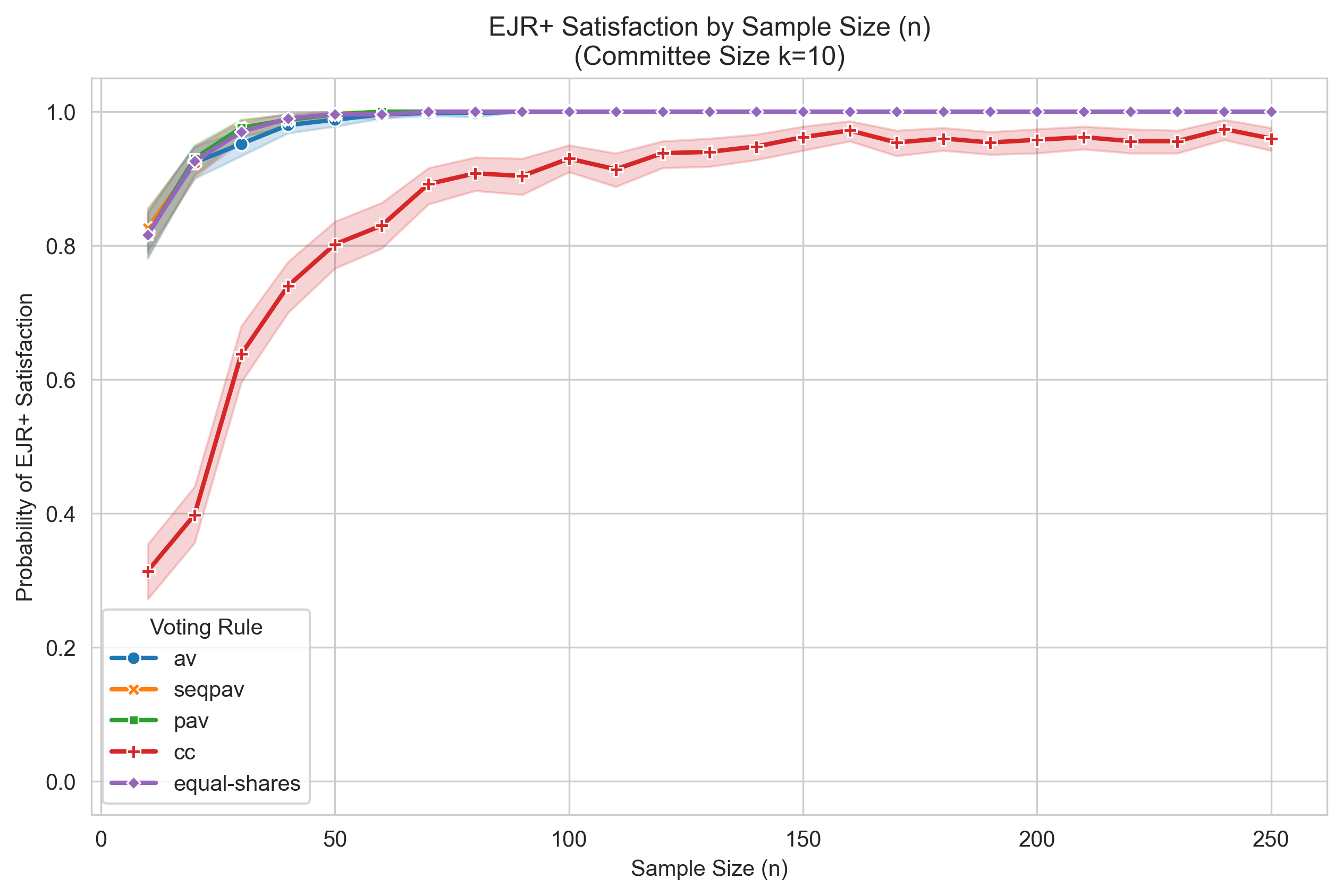}
		\caption{Probability of \ejrp{} for $k=10$}
		\label{fig:warsza_k10_ejrp}
	\end{subfigure}%
	\hfill
	\begin{subfigure}{.45\textwidth}
		\includegraphics[width=\linewidth]{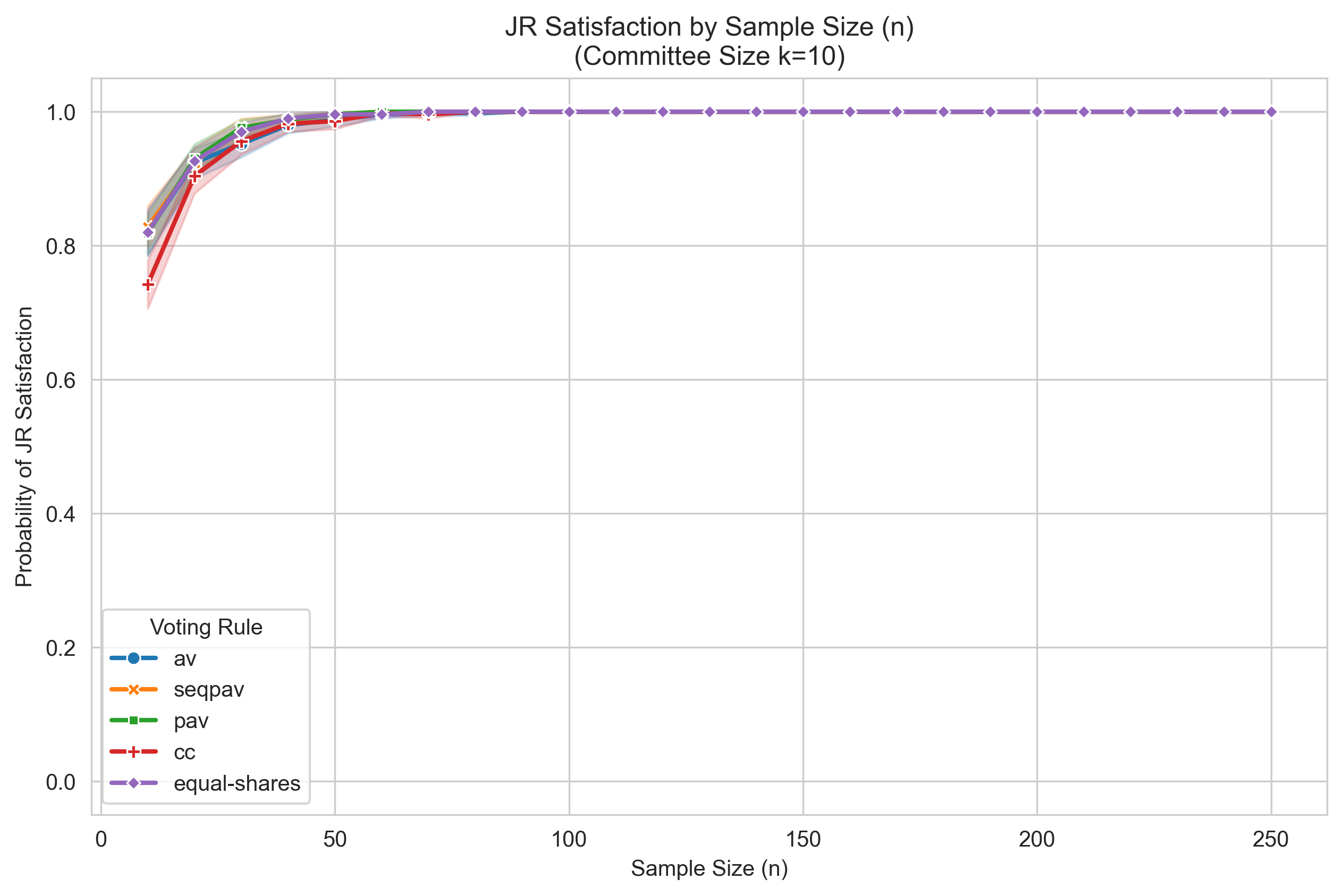}
		\caption{Probability of \jr{} for $k=10$}
		\label{fig:warsza_k10_jr}
	\end{subfigure}
	\\
	\begin{subfigure}{.45\textwidth}
		\centering
		\includegraphics[width=\linewidth]{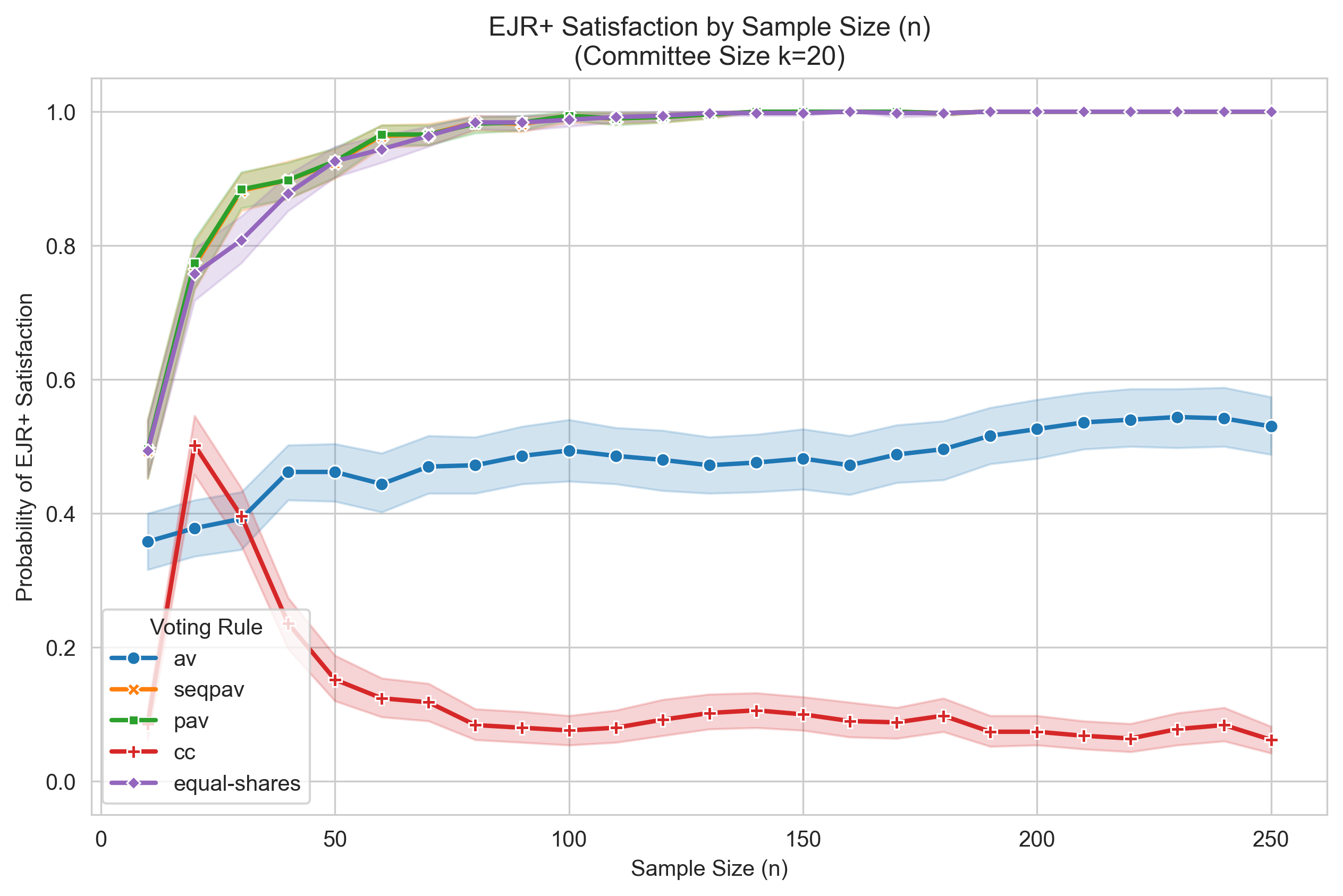}
		\caption{Probability of \ejrp{} for $k=20$}
		\label{fig:warsza_k10_ejrp}
	\end{subfigure}%
	\hfill
	\begin{subfigure}{.45\textwidth}
		\includegraphics[width=\linewidth]{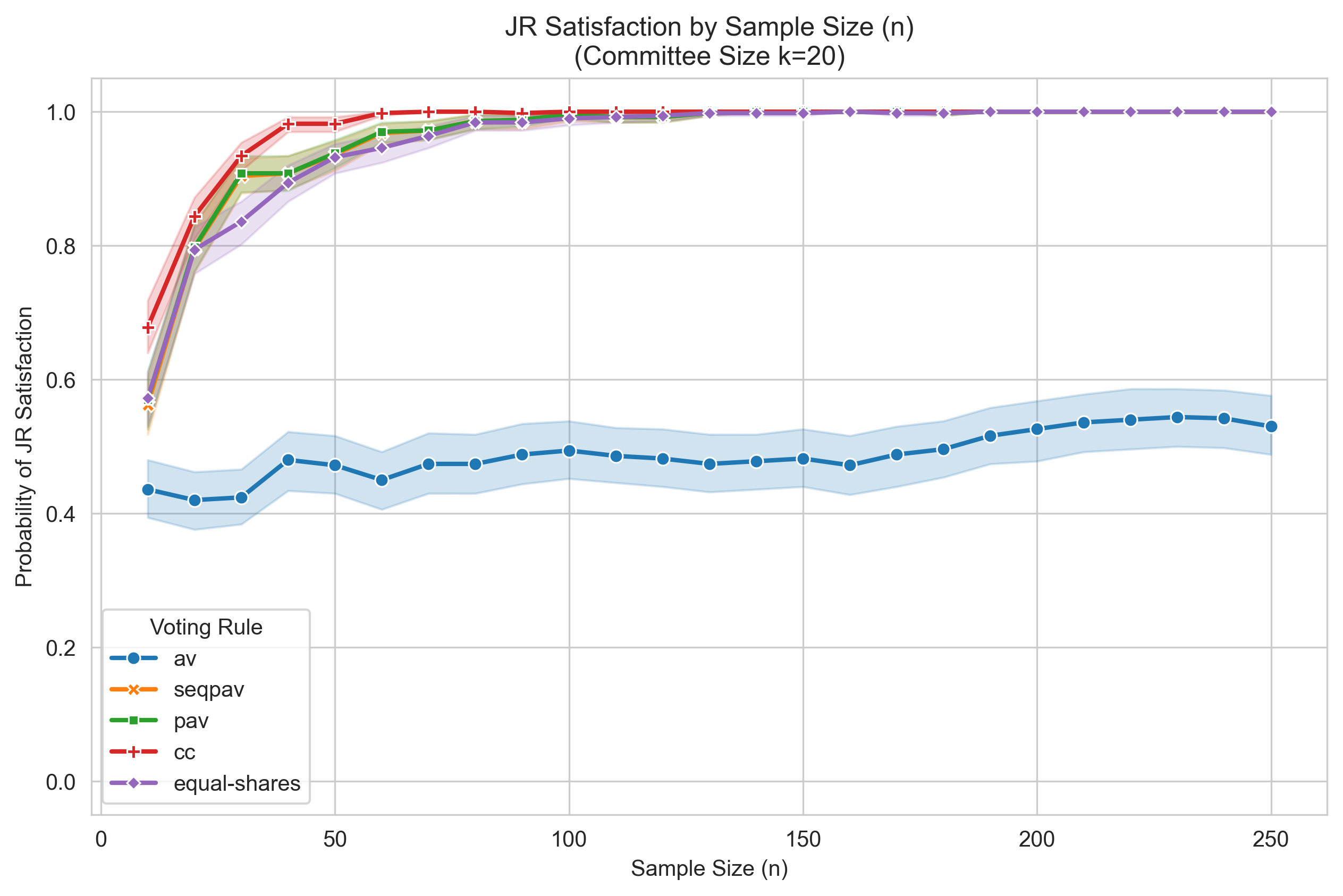}
		\caption{Probability of \jr{} for $k=20$}
		\label{fig:warsza_k10_jr}
	\end{subfigure}
	\\
		\begin{subfigure}{.45\textwidth}
		\centering
		\includegraphics[width=\linewidth]{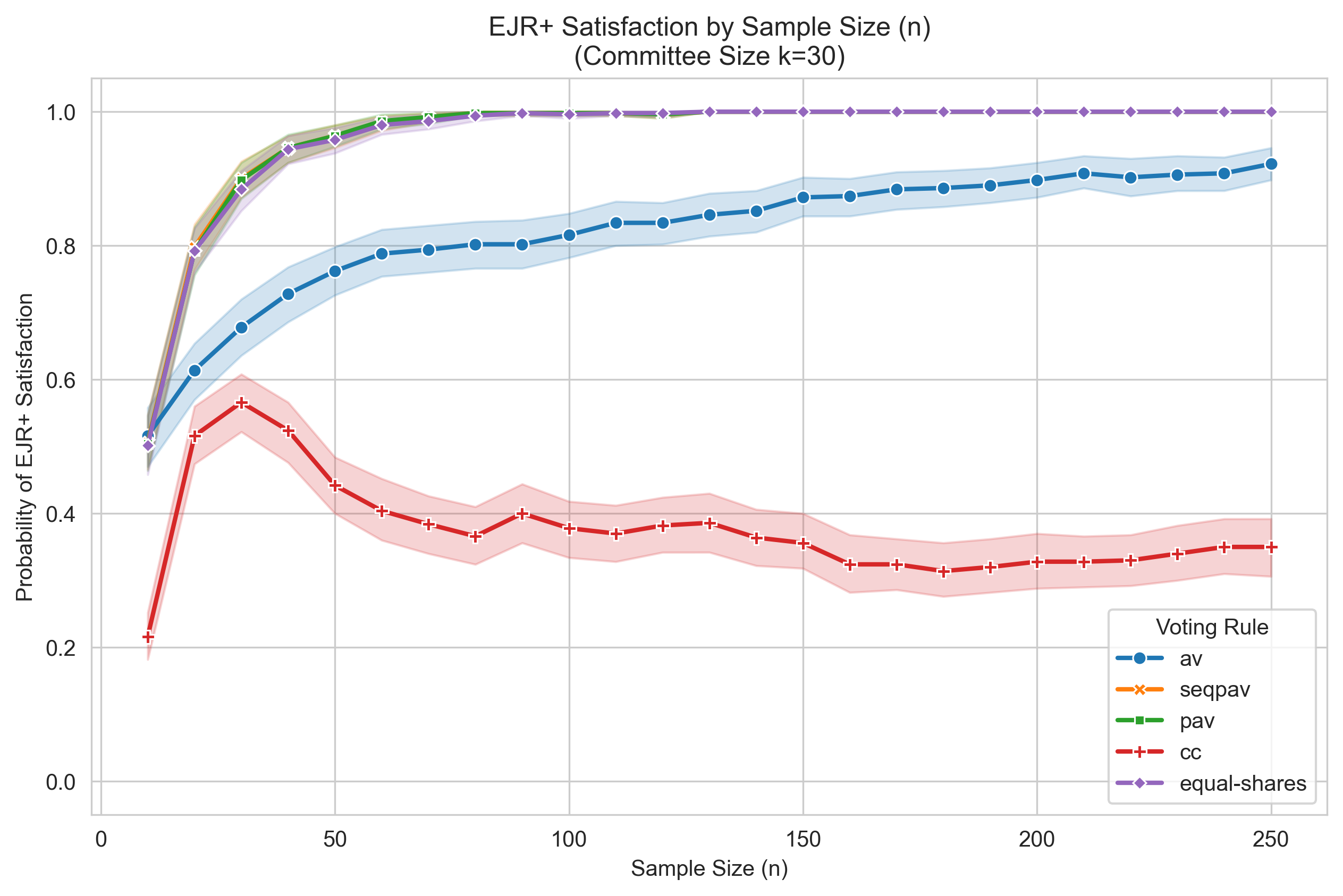}
		\caption{Probability of \ejrp{} for $k=30$}
		\label{fig:warsza_k10_ejrp}
	\end{subfigure}%
	\hfill
	\begin{subfigure}{.45\textwidth}
		\includegraphics[width=\linewidth]{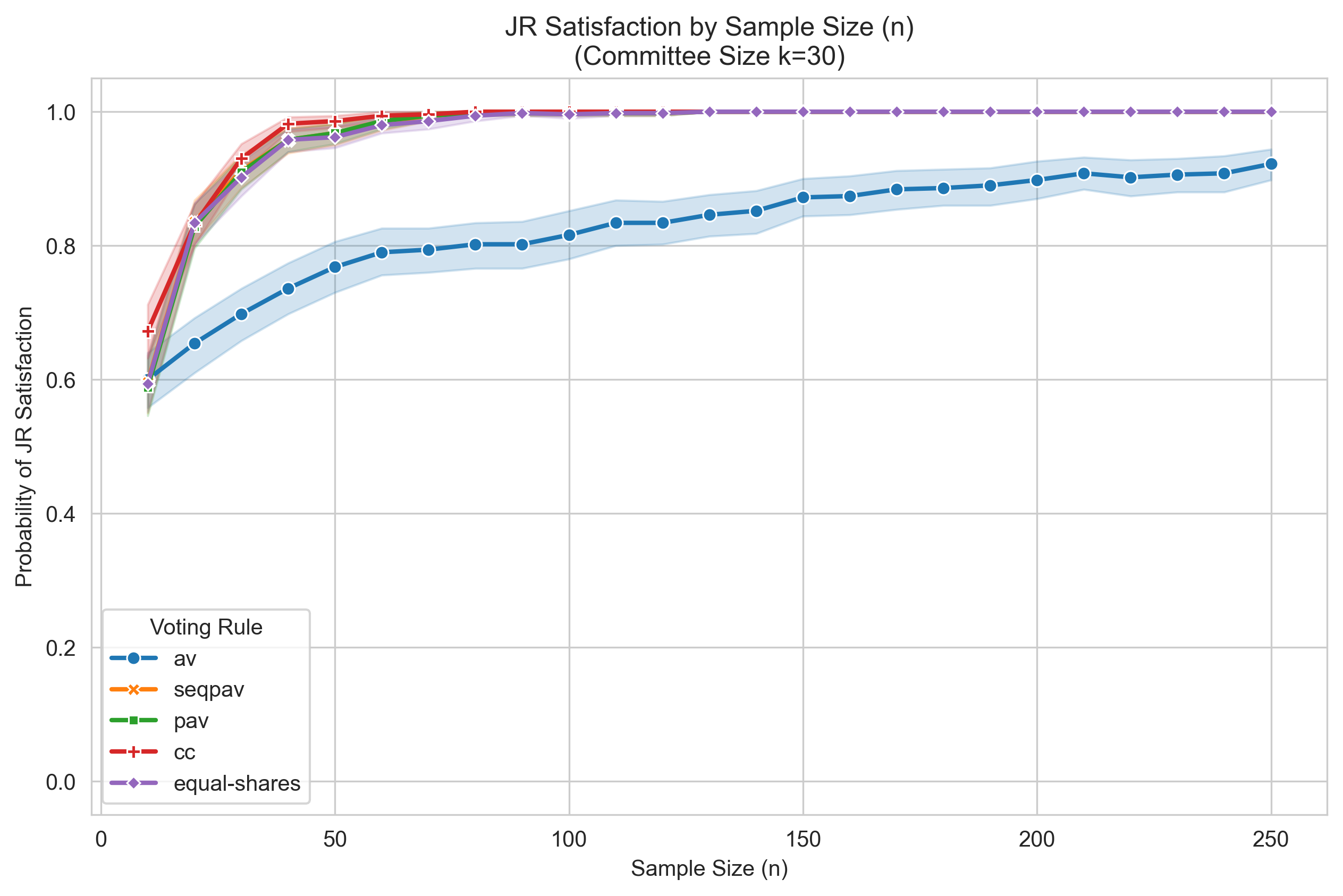}
		\caption{Probability of \jr{} for $k=30$}
		\label{fig:warsza_k10_jr}
	\end{subfigure}
	\caption{Success probability of satisfying \jr{} and \ejrp{} for various rules on approvals from the ``Poland Warszawa 2017 rejon B'' participatory budgeting election, as the number of samples increases. Empirical probabilities are out of 500 trials each for committee sizes $k \in \{10,20,30\}$, out of 53 candidates overall.}
	\label{fig:warsza_increasing_n}
\end{figure}

\end{document}